\documentclass[12pt]{article}
\usepackage[capposition=top]{floatrow}
\usepackage{amssymb,amsmath,amsthm,amsfonts,changepage,eurosym,geometry,ulem,graphicx,caption,color,setspace,sectsty,comment,footmisc,caption,natbib,pdflscape,subcaption,array,hyperref,soul,enumerate}
\allowdisplaybreaks

\newtheorem{theorem}{Theorem}
\newtheorem{corollary}[theorem]{Corollary}
\newtheorem{proposition}{Proposition}
\newtheorem{assumption}{Assumption}
\newtheorem{lemma}{Lemma}

\theoremstyle{plain}

\newcounter{tablenums}

\begin{document}
\begin{titlepage}
\title{\singlespacing Breaking Ties: Regression Discontinuity Design\break Meets Market Design\thanks{ \scriptsize We thank Nadiya Chadha,  Andrew McClintock, Sonali Murarka, Lianna Wright, and the staff of the New York City Department of Education for answering our questions and facilitating access to data. Don Andrews, Tim Armstrong, Eduardo Azevedo, Yeon-Koo Che, Glenn Ellison, Brigham Frandsen, John Friedman, Justine Hastings, Guido Imbens, Jacob Leshno, Whitney Newey, Ariel Pakes, Pedro Sant'Anna, Olmo Silva, Hal Varian and seminar participants at Columbia, Montreal, Harvard, Hebrew University, Google, the NBER Summer Institute, the NBER Market Design Working Group, the FRB of Minneapolis, CUNY, Yale, Hitotsubashi, and Tokyo provided helpful feedback. We're especially indebted to Adrian Blattner, Nicolas Jimenez, Ignacio Rodriguez, and Suhas Vijaykumar for expert research assistance and to MIT SEII team leaders Eryn Heying and Anna Vallee for invaluable administrative support. We gratefully acknowledge funding from the Laura and John Arnold Foundation, the National
Science Foundation (under awards SES-1056325 and SES-1426541), and the W.T. Grant Foundation.}}
\author{At\.{\i}la Abdulkad\.{\i}ro\u{g}lu\thanks{\scriptsize Duke
University and NBER. Email: \href{mailto:atila.abdulkadiroglu@duke.edu}{atila.abdulkadiroglu@duke.edu}} \and Joshua D. Angrist\thanks{\scriptsize MIT and NBER. Email: \href{mailto:angrist@mit.edu}{angrist@mit.edu}} \and Yusuke Narita \thanks{\scriptsize Yale University and NBER. Email: \href{mailto:yusuke.narita@yale.edu}{yusuke.narita@yale.edu}} \and Parag Pathak\thanks{\scriptsize MIT and NBER. Email: \href{mailto:ppathak@mit.edu}{ppathak@mit.edu}} }
\date{\today}
\maketitle
\begin{adjustwidth*}{0.5cm}{0.5cm} 
\begin{abstract}
\singlespacing
\noindent Many schools in large urban districts have more applicants than seats. Centralized school assignment algorithms ration seats at over-subscribed schools using randomly assigned lottery numbers, non-lottery tie-breakers like test scores, or both. The New York City public high school match illustrates the latter, using test scores and other criteria to rank applicants at ``screened'' schools, combined with lottery tie-breaking at unscreened ``lottery'' schools. We show how to identify causal effects of school attendance in such settings. Our approach generalizes regression discontinuity methods to allow for multiple treatments and multiple running variables, some of which are randomly assigned.  The key to this generalization is a local propensity score that quantifies the school assignment probabilities induced by lottery and non-lottery tie-breakers. The local propensity score is applied in an empirical assessment of the predictive value of New York City's school report cards. Schools that receive a high grade indeed improve SAT math scores and increase graduation rates, though by much less than OLS estimates suggest. Selection bias in OLS estimates is egregious for screened schools. \par 
\end{abstract}
\end{adjustwidth*}
\onehalfspacing
\setcounter{page}{0}
\thispagestyle{empty}
\end{titlepage}
\pagebreak \newpage

\onehalfspacing

\section{Introduction}
\thispagestyle{empty}

Large school districts increasingly use sophisticated centralized assignment mechanisms to match students and schools.  In addition to producing fair and transparent admissions decisions, centralized assignment offers a unique resource for research on schools: the data these systems generate can be used to construct unbiased estimates of school value-added. This research dividend arises from the \textit{tie-breaking} embedded in centralized assignment.  Many school assignment schemes rely on the deferred acceptance (DA) algorithm, which takes as input information on applicant preferences and school priorities. In settings where seats are scarce, DA rations seats at oversubscribed schools using tie-breaking variables, producing quasi-experimental assignment of students to schools.

Many districts break ties with a uniformly distributed random variable, often described as a lottery number.  \cite{mdrd1:17} show that DA with lottery tie-breaking assigns students to schools as if in a stratified randomized trial.  That is, conditional on preferences and priorities, the assignments generated by such systems are randomly assigned and therefore independent of potential outcomes.  In practice, however, preferences and priorities, which we call applicant \textit{type}, are too finely distributed for full non-parametric conditioning to be useful. We must therefore pool applicants of different types, while avoiding any omitted variables bias that might arise from the fact that type predicts outcomes.  

The key to type pooling is the \textit{DA propensity score}, defined as the probability of school assignment conditional on applicant type.  In a mechanism with lottery tie-breaking, conditioning on the scalar DA propensity score is sufficient to make school assignment independent of potential outcomes. Moreover, the distribution of the scalar propensity score turns out to be much coarser than the distribution of types.\footnote{The propensity score theorem says that for research designs in which treatment status, $D_i$, is independent of potential outcomes conditional on covariates, $X_i$, treatment status is also independent of potential outcomes conditional on the propensity score, that is, conditional on $E[D_i|X_i]$. In work building on \cite{mdrd1:17}, the DA propensity score has been used to study schools \citep{bergman2018risks}, management training \citep{abebe2019learning}, and entrepreneurship training \citep{argentina}.}

This paper generalizes the propensity score to DA-based assignment mechanisms in which tie-breaking variables are not limited to randomly assigned lottery numbers. Selective exam schools, for instance, admit students with high test scores, and students with higher scores tend to have better achievement and graduation outcomes regardless of where they enroll. We refer to such scenarios as involving \textit{general tie-breaking}.\footnote{Non-lottery tie-breaking embedded in centralized assignment schemes has been used in econometric research on schools in Chile \citep{hastings/neilson/zimmerman:13, zimmerman2016making}, Ghana \citep{ajayi:13}, Italy \citep{fort2016cognitive}, Kenya \citep{lucas/mbiti:14}, Norway \citep{kirkeboen/leuven/mogstad:15}, Romania \citep{pop-eleches/urquiola:13}, Trinidad and Tobago \citep{jackson:10, jackson:12, beuermann2015privately}, and the U.S. \citep{abdulkadiroglu/angrist/pathak:14, dobbie/fryer:11, barrow2016role}. These studies treat individual schools and tie-breakers in isolation, without exploiting centralized assignment. Related methodological work exploring regression discontinuity designs with multiple assignment variables and multiple cutoffs includes \cite{papay2011extending, zajonc2012regression, wong2013analyzing, cattaneo2016interpreting}.}  Matching markets with general tie-breaking raise challenges beyond those addressed in the \cite{mdrd1:17} study of DA with lottery tie-breaking. 

The most important complication raised by general tie-breaking arises from the fact that seat assignment is no longer independent of potential outcomes conditional on applicant type. This problem is intimately entwined with the identification challenge raised by regression discontinuity (RD) designs, which typically compare candidates for treatment on either side of a qualifying test score cutoff. In particular, non-lottery tie-breakers play the role of an RD \textit{running variable} and are likewise a source of omitted variables bias. The setting of interest here, however, is far more complex than the typical RD design: DA may involve many treatments, tie-breakers, and cutoffs.  

A further barrier to causal inference comes from the fact that the propensity score in this general tie-breaking setting depends on the unknown distribution of non-lottery tie-breakers conditional on type. Consequently, the propensity score under general tie-breaking may be no coarser than the underlying high-dimensional type distribution.  When the score distribution is no coarser than the type distribution, score conditioning is pointless. 

These problems are solved here by introducing a \textit{local DA propensity score} that quantifies the probability of school assignment induced by a combination of non-lottery and lottery tie-breakers.  This score is ``local'' in the sense that it  is constructed using the fact that continuously distributed non-lottery tie-breakers are locally uniformly distributed. Combining this property with the (globally) known distribution of lottery tie-breakers yields a formula for the assignment probabilities induced by any DA match. Conditional on the local DA propensity score, school assignments are shown to be asymptotically randomly assigned. Moreover, like the DA propensity score for lottery tie-breaking, the local DA propensity score has a distribution far coarser than the underlying type distribution.

Our analytical approach builds on \cite{hahn2001identification} and other pioneering econometric contributions to the development of non-parametric RD designs. We also build on the more recent local random assignment interpretation of nonparametric RD.\footnote{See, among others, \cite{frolich2007regression, cattaneo2015randomization, cattaneo2017comparing, frandsen2017party, sekhon2017interpreting}; \cite{frolich/huber2019}; and \cite{arai2019testing}.} The resulting theoretical framework allows us to quantify the probability of school assignment as a function of a few features of student type and tie-breakers, such as proximity to the admissions cutoffs determined by DA and the identity of key cutoffs for each applicant. By integrating nonparametric RD with \cite{rosenbaum/rubin:83}'s propensity score theorem and large-market matching theory, our theoretical results provide a framework suitable for causal inference in a wide variety of applications.

The research value of the local DA propensity score is demonstrated through an analysis of New York City (NYC) high school report cards.  Specifically, we ask whether schools distinguished by ``Grade A''  on the district's school report card indeed signify high quality schools that boost their students' achievement and improve other outcomes.  Alternatively, the good performance of most Grade A students may reflect omitted variables bias. The distinction between causal effects and omitted variables bias is especially interesting in light of an ongoing debate over access to New York's academically selective schools, also called screened schools, which are especially likely to be graded A (see, e.g., \cite{brody:19} and \cite{veiga:18}). We identify the causal effects of Grade A school attendance by exploiting the NYC high school match. NYC employs a DA mechanism integrating non-lottery screened school tie-breaking with a common lottery tie-breaker at lottery schools.   In fact, NYC screened schools design their own tie-breakers based on middle school transcripts, interviews, and other factors.  

The effects of Grade A school attendance are estimated here using instrumental variables constructed from the school assignment offers generated by the NYC high school match. Specifically, our two-stage least squares (2SLS) estimators use assignment offers as instrumental variables for Grade A school attendance, while controlling for the local DA propensity score. The resulting estimates suggest that Grade A attendance boosts SAT math scores modestly and may increase high school graduation rates a little. But these effects are much smaller than those the corresponding ordinary least squares (OLS) estimates of Grade A value-added would suggest. 

We also compare 2SLS estimates of Grade A effects computed separately for NYC's screened and lottery schools, a comparison that shows the two sorts of schools to have similar effects. This finding therefore implies that OLS estimates showing a large Grade A screened school advantage are especially misleading. The distinction between screened and lottery schools has been central to the ongoing debate over NYC school access and quality. Our comparison suggests that the public concern with this sort of treatment effect heterogeneity may be misplaced. Treatment effect heterogeneity may be limited, supporting our assumption of constant treatment effects conditional on observables.\footnote{The analysis here allows for treatment effect heterogeneity as a function of observable student and school characteristics. Our working paper shows how DA in markets with general tie-breaking identifies average causal affects for applicants with tie-breaker values away from screened-school cutoffs (\cite{abdulkadiroglu/angrist/narita/pathak:19}). We leave other questions related to unobserved heterogeneity for future work.}   

The next section shows how DA can be used to identify causal effects of school attendance. Section \ref{sec:SDsection} illustrates key ideas in a setting with a single non-lottery tie-breaker. Section \ref{sec:multiscore} derives a formula for the local DA propensity score in a market with general tie-breaking. This section also derives a consistent estimator of the local propensity score. Section \ref{sec:NYC} uses these theoretical results to estimate causal effects of attending Grade A schools.\footnote{Our theoretical analysis covers any mechanism that can be computed by student-proposing DA. This DA class includes student-proposing DA, serial dictatorship, the immediate acceptance (Boston) mechanism \citep{ergin/sonmez:05}, China's parallel mechanisms \citep{chen/kesten:17}, England's first-preference-first mechanisms \citep{pathak/sonmez:13}, and the Taiwan mechanism \citep{dur/pathak/song/sonmez:18}.  In large markets satisfying regularity conditions that imply a unique stable matching, the relevant DA class also includes school-proposing as well as applicant-proposing DA (these conditions are spelled out in \cite{azevedo/leshno:14}). The DA class omits the Top Trading Cycles (TTC) mechanism defined for school choice by \cite{abdulkadiroglu/sonmez:03}.}

\section{Using Centralized Assignment to Eliminate Omitted Variables Bias}\label{sec:setup}

The NYC school report cards published from 2007-13 graded high schools on the basis of student achievement, graduation rates, and other criteria.  These grades were part of an accountability system meant to help parents choose high quality schools. In practice, however, report card grades computed without extensive control for student characteristics reflect students' ability and family background as well as school quality. Systematic differences in student body composition are a powerful source of bias in school report cards. It's therefore worth asking whether a student who is \textit{randomly} assigned to a Grade A high school indeed learns more and is more likely to graduate as a result.

We answer this question using instrumental variables derived from NYC's DA-based assignment of high school seats. The NYC high school match generates a single school assignment for each applicant as a function of applicants' preferences over schools, school-specific priorities, and a set of tie-breaking variables that distinguish between applicants who share preferences and priorities.\footnote{Seat assignment at some of NYC's selective enrollment ``exam schools'' is determined by a separate match.  NYC charter schools use school-specific lotteries. Applicants are free to seek exam school and charter school seats as well as an assignment in the traditional sector.}  Because they're a function of student characteristics like preferences and test scores, NYC assignments are not randomly assigned.  We show, however, that conditional on the local DA propensity score, DA-generated assignments of a seat at school $s$ provide credible instruments for enrollment at school $s$. This result motivates a two-stage least squares (2SLS) specification where the endogenous treatment is enrollment at any Grade A school while the instrument is DA-generated assignment of a seat at any Grade A school.  

Our identification strategy builds on the large-market ``continuum'' model of DA detailed in \cite{mdrd1:17}.  The large-market model is extended here to allow for multiple and non-lottery tie-breakers. To that end, let $s=0,1,...,S$ index schools, where $s=0$ represents an outside option. Applicants are assumed to be identified by an index, $i$, drawn from the unit interval $[0, 1]$.  The large market model is ``large'' by virtue of this assumption. 

Applicant $i$'s preferences over schools constitute a strict partial ordering, $\succ_{i}$, where $a\succ_{i}b$ means that $i$ prefers school $a$ to
school $b$. Each applicant is also granted a priority at every school. For example, schools may prioritize applicants who live nearby or with currently enrolled siblings. Let
$\rho_{is}\in\{1,...,K,\infty\}$ denote applicant $i$'s priority at school $s$,
where $\rho_{is}<\rho_{js}$ means school $s$ prioritizes $i$ over $j$. 
We use $\rho_{is}=\infty$ to indicate that $i$ is ineligible for school $s.$
The vector $\rho_{i}=(\rho_{i1},...,\rho_{iS})$
records applicant $i$'s priorities at each school.
Applicant type is then defined as $\theta_{i}=(\succ_{i},\rho_{i})$, that is, the combination of an applicant's  preferences and priorities at all schools. Let $\Theta_s$ denote the set of types, $\theta$, that ranks $s$.

In addition to applicant type, DA matches applicants to seats as a function of a set of tie-breaking variables. We leave DA mechanics for Section \ref{sec:multiscore}; at this point, it's enough to establish notation for DA inputs. 
Most importantly, our analysis of markets with general tie-breaking requires notation to keep track of tie-breakers.  Let $v\in\{1,...,V\}$ index tie-breakers and let $S_v$ be the set of schools using tie-breaker $v$.
We assume that each school uses a single tie-breaker.
Scalar random variable $R_{iv}$ denotes applicant $i$'s tie-breaker $v$. Some of these are uniformly distributed lottery numbers.  The set of non-lottery $R_{iv}$ used at schools ranked by applicant $i$ are collected in the vector $\mathcal{R}_i$.  
Without loss of generality, we assume that ties are broken in favor of applicants with the smaller tie-breaker value. 
DA uses $\theta_i$, $\mathcal{R}_i$, and the set of lottery
tie-breakers for all $i$ to assign applicants to schools. 

We are interested in using the assignment variation resulting from DA to estimate the causal effect of $C_{i}$, a variable indicating student $i$'s attendance at (or years of enrollment in)  
any Grade A school. Outcome variables, denoted $Y_i$, include SAT scores and high school graduation status. In a DA match like the one in NYC, $C_i$ is not randomly assigned, but rather reflects student preferences, school priorities, tie-breaking variables, as well as decisions whether or not to enroll at school $s$ when offered a seat there through the match. The potential for omitted variables bias induced by the process determining $C_i$ can be eliminated by an instrumental variables strategy that exploits our understanding of the structure of matching markets.  

The instruments used for this purpose are a function of individual school assignments, indicated by $D_i(s)$ for the assignment of student $i$ to a seat at school $s$. Because DA generates a single assignment for each student, a dummy for any Grade A assignment, denoted $D_{Ai}$, is the sum of dummies indicating all assignments to individual Grade A schools. $D_{Ai}$ provides a natural instrument for $C_i$. In particular, we show below that 2SLS consistently estimates the effect of $C_i$ on $Y_i$ in the context of a linear constant-effects causal model that can be written as:
\begin{align}
Y_i & = \beta C_{i} + f_2(\theta_i, \mathcal{R}_i, \delta) + \eta_i\label{RF}, 
\end{align}
where $\beta$ is the causal effect of interest and the associated first stage equation is
\begin{align}
C_{i} &=  \gamma D_{Ai} + f_1(\theta_i, \mathcal{R}_i, \delta) + \nu_{i}\label{FS}.
\end{align}
The terms $f_1(\theta_i, \mathcal{R}_i, \delta)$ and $f_2(\theta_i, \mathcal{R}_i, \delta)$ are functions of type and non-lottery tie-breakers, as well as a bandwidth, $\delta\in \mathbb{R}$, that's integral to the local DA propensity score.

Our goal is to specify $f_1(\theta_i, \mathcal{R}_i, \delta)$ and $f_2(\theta_i, \mathcal{R}_i, \delta)$ so that 2SLS estimates of $\beta$ are consistent. Because \eqref{RF} is seen as a model for potential outcomes rather than a regression equation, consistency requires that $D_{Ai}$ and $\eta_i$ be uncorrelated.  The relevant identification assumption can be written: 
\begin{equation}
E[\eta_i D_{Ai}]\approx 0,\label{CI}
\end{equation}

\noindent where $\approx$ means asymptotic equality as $\delta\rightarrow0$, in a manner detailed below. Briefly, our main theoretical result establishes limiting local conditional mean independence of school assignments from applicant characteristics and potential outcomes, yielding \eqref{CI}. This result specifies $f_1(\theta_i, \mathcal{R}_i, \delta)$ and $f_2(\theta_i, \mathcal{R}_i, \delta)$ to be easily-computed functions of the local propensity score and elements of $\mathcal{R}_i$.   

\cite{mdrd1:17} derive the relevant DA propensity score for a scenario with lottery tie-breaking only.  Lottery tie-breaking obviates the need for a bandwidth and control for components of $\mathcal{R}_i$.  Many applications of DA use non-lottery tie-breaking, however. The task here is to derive the propensity score for elaborate matches like that in NYC, which combines lottery tie-breaking with many school-specific non-lottery tie-breakers. The resulting estimation strategy integrates propensity score methods with the nonparametric approach to RD (introduced by \cite{hahn2001identification}), and the local random assignment model of RD (discussed by \cite{frolich2007regression, cattaneo2015randomization, cattaneo2017comparing, frandsen2017party}, among others). Our theoretical results can also be seen as generalizing nonparametric RD to allow for many schools (treatments), many tie-breakers (running variables), and many cutoffs.  




\section{From Non-Lottery Tie-Breaking to Random Assignment in Serial Dictatorship}\label{sec:SDsection} 


An analysis of a market with a single non-lottery tie-breaker and no priorities illuminates key elements of our approach.  DA in this case is known as \textit{serial dictatorship}.  Like the general local DA score, the local DA score for serial dictatorship depends only on a handful of statistics, including admissions cutoffs for schools ranked, and whether applicant $i$'s tie-breaker is close to cutoffs for schools using non-lottery tie-breakers. Conditional on this local propensity score, school offers are asymptotically randomly assigned.

Serial dictatorship can be described as follows:
\begin{quote}
Order applicants by tie-breaker.  Proceeding in order, assign each applicant to his or her most preferred school among those with seats remaining.  
\end{quote} 
Seating is constrained by a capacity vector, $q=(q_0,q_1,q_2,...,q_S)$, where $q_{s}\in[0, 1]$ is defined as the proportion of the unit interval that can be seated at school $s$. We assume $q_0=1$. 
Serial dictatorship is used in Boston and New York City to allocate seats at selective public exam schools.

Because serial dictatorship relies on a single tie-breaker, notation for the set of non-lottery tie-breakers, $\mathcal{R}_i$, can be replaced by a scalar, $R_{i}$.  As in \cite{mdrd1:17}, tie-breakers for individuals are modelled as stochastic, meaning they are drawn from a distribution for each applicant.  Although $R_{i}$ is not necessarily uniform, we assume that it's distributed with positive density over $[0, 1]$, with continuously differentiable cumulative distribution function, $F^i_R$.  These common support and smoothness assumptions notwithstanding, tie-breakers may be correlated with type, so that $R_{i}$ and $R_{j}$ for applicants $i$ and $j$ are not necessarily identically distributed, though they're assumed to be independent of one another.
The probability that type $\theta$ applicants have a tie-breaker below any value $r$ is $F_R(r| \theta) \equiv E[F^i_R(r)| \theta_i=\theta]$, where $F^i_R(r)$ is $F^i_R$ evaluated at $r$.  

The serial dictatorship allocation is characterized by a set of \textit{tie-breaker cutoffs}, denoted $\tau_s$ for school $s$. For any school $s$ that's filled to capacity, $\tau_s$ is given by the tie-breaker of the last (highest tie-breaker value) student assigned to $s$. Otherwise, $\tau_s=1$, a non-binding cutoff reflecting excess capacity. We say an applicant \textit{qualifies} at $s$ when they have a tie-breaker value that clears $\tau_s$.  Under serial dictatorship, students are assigned to $s$ if and only if they:
\begin{itemize}
\item qualify at $s$ (since seats are assigned in tie-breaker order)
\item fail to qualify at any school they prefer to $s$ (since serial dictatorship assigns available seats at preferred schools first)
\end{itemize}

\noindent In large markets, cutoffs are constant, so stochastic variation in seat assignments arises solely from the distribution of tie-breakers. 

\subsection{The Serial Dictatorship Propensity Score}\label{sec:SDscoredetail}

Which cutoffs matter? Under serial dictatorship, the assignment probability faced by an applicant of type $\theta$ at school $s$ is determined by the cutoff at $s$ and by cutoffs at schools preferred to $s$. By virtue of single tie-breaking, it's enough to know only one of the latter.  In particular, an applicant who fails to clear the highest cutoff among those at schools preferred to $s$ surely fails to do better than $s$.  This leads us to define \textit{most informative
disqualification} (MID), a scalar parameter for each applicant type and school.  MID tells us how the tie-breaker distribution among type $\theta$ applicants to $s$ is truncated by disqualification at the schools type $\theta$ applicants prefer to $s$.   

Because MID for type $\theta$ at school $s$ is defined with reference to the set of schools $\theta$ prefers to $s$, 
we define:
\begin{equation}\label{better}
B_{\theta s}=\{s' \ne s  \hspace*{.1cm} | \hspace*{.1cm}  s' \succ_{\theta}s\} \text{ for each } \theta \in \Theta_{s},
\end{equation}
the set of schools type $\theta$ prefers to $s$.
For each type and school, $MID_{\theta s}$ is a function of tie-breaker cutoffs at schools in $B_{\theta s}$, specifically:
\begin{align}\label{SDMID}
MID_{\theta s}\equiv
\left\{
\begin{array}
[c]{ll}
0 & \text{ if } B_{\theta s} = \emptyset\\
\max\{\tau_{b}  \hspace*{.1cm} | \hspace*{.1cm} b\in B_{\theta s}\} & \text{ otherwise. }
\end{array}
\right.
\end{align}

$MID_{\theta s}$ is zero when school $s$ is ranked first since all who rank $s$ first compete for a seat there.  The second line reflects the fact that an applicant who ranks $s$ second is seated there only when disqualified at the school they've ranked first, while applicants who rank $s$ third are seated there when disqualified at their first and second choices, and so on.  Moreover, anyone who fails to clear cutoff
$\tau_{b}$ is surely disqualified at schools with less forgiving cutoffs.
For example, applicants who fail to qualify at a school with a cutoff of 0.6 are disqualified at a school with cutoff 0.4.

Note that an applicant of type $\theta$ cannot be seated at $s$ when $MID_{\theta s} > \tau_s$.  This is the scenario sketched in the top panel of Figure \ref{fig:SDsketch}, which illustrates the forces determining SD assignment rates. On the other hand, assignment rates when $MID_{\theta s} \leq \tau_s$ are given by the probability that:
$$MID_{\theta s} < R_i \leq \tau_s,$$  
an event described in the middle panel of Figure \ref{fig:SDsketch}.
These facts are collected in the following proposition, which is implied by a more general result for DA proved in the online appendix.  

\begin{proposition}[Propensity Score in Serial Dictatorship]\label{RSDpscore}
Suppose seats in a large market are assigned by serial dictatorship. Let $p_s(\theta)=E[D_i(s)| \theta_i=\theta]$ denote the type $\theta$ propensity score for assignment to $s$.  
For all schools $s$ and $\theta \in \Theta_{s}$, we have:
\begin{align*}
p_{s}(\theta) =  \max \left\{ 0, F_R(\tau_{s}| \theta)-F_R(MID_{\theta s}| \theta) \right\}.
\end{align*}
\end{proposition}

Proposition \ref{RSDpscore} says that the serial dictatorship assignment probability, positive only when the tie-breaker cutoff at $s$ exceeds $MID_{\theta s}$, is given by the size of the group with $R_i$ between $MID_{\theta s}$ and $\tau_s$.  This is $$F_R(\tau_{s}| \theta)-F_R(MID_{\theta s}| \theta).$$
With a uniformly distributed lottery number, the serial dictatorship propensity score simplifies to $\tau_s - MID_{\theta s}$, a scenario noted in Figure \ref{fig:SDsketch}. In this case, the assignment probability for each applicant is determined by $\tau_s$ and $MID_{\theta s}$ alone. Given these two cutoffs, seats at $s$ are randomly assigned. 

\subsection{Serial Dictatorship Goes Local} 

With non-lottery tie-breaking, the serial dictatorship propensity score depends on the conditional distribution function, $F_R(\cdot| \theta)$ evaluated at $\tau_s$ and $MID_{\theta s}$, rather than the cutoffs themselves. This dependence leaves us with two econometric challenges.  First, $F_R(\cdot| \theta)$ is unknown. This precludes computation of the propensity score by repeatedly sampling from $F_R(\cdot| \theta)$. 
Second, $F_R(\cdot| \theta)$, is likely to depend on $\theta$, so the score in Proposition \ref{RSDpscore} need not have coarser support than does $\theta$. This is in spite of the fact many applicants with different values of $\theta$ share the same $MID_{\theta s}$. 
Finally, although controlling for $p_{s}(\theta)$
eliminates confounding from type, assignments are a function of tie-breakers as well as type. Confounding from non-lottery tie-breakers remains even after conditioning on $p_{s}(\theta)$.

These challenges are met here by focusing on assignment probabilities for applicants with tie-breaker realizations close to key cutoffs. Specifically, for each $\tau_s$, define an interval, $(\tau_s - \delta, \tau_s + \delta]$, where parameter $\delta$ is a bandwidth analogous to that used for nonparametric RD estimation.  A local propensity score treats the qualification status of applicants inside this interval as randomly assigned.  This assumption is justified by the fact that, given continuous differentiability of tie-breaker distributions, non-lottery tie-breakers have a limiting uniform distribution as the bandwidth shrinks to zero.   

\setcounter{figure}{0}
\begin{figure}[!t]
    \centering
    \caption{Assignment Probabilities under Serial Dictatorship}
    \includegraphics[width=\textwidth]{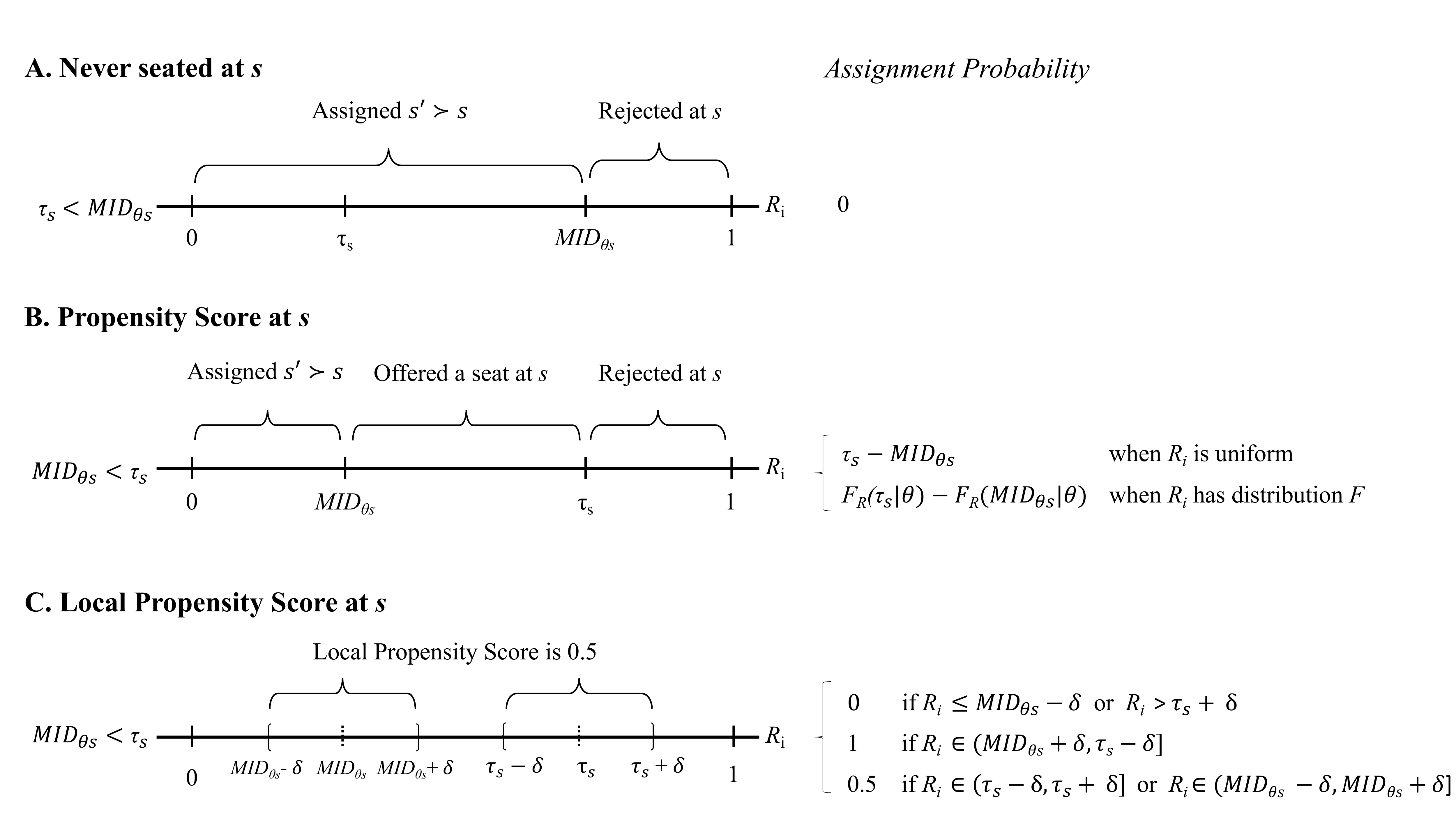}  \label{fig:SDsketch}
    \vspace{-0.5cm}
\floatfoot{\footnotesize{\textit{Notes:} This figure illustrates the assignment probability at school $s$ under serial dictatorship. $R_i$ is the tie-breaker. $MID_{\theta s}$ is the most forgiving cutoff at schools preferred to $s$ and $\tau_s$ is the cutoff at $s$.}}
\end{figure}


The following Proposition uses this fact to characterize the local serial dictatorship propensity score:

\begin{proposition}[Local Serial Dictatorship Propensity Score]\label{local_pscore} 
Suppose seats in a large market are assigned by serial dictatorship. Also, let $W_i$ be any applicant characteristic other than type that is unchanged by school assignment.\footnote{Let $W_i=\Sigma_s D_{i}(s)W_{i}(s)$, where $W_{i}(s)$ is the potential value of $W_i$ revealed when $D_{i}(s)=1$. We say $W_i$ is unchanged by school assignment when $W_{i}(s)=W_{i}(s')$ for all $s \neq s'$. Examples include demographic characteristics and potential outcomes.} Finally, assume $\tau_s\neq \tau_{s'}$ for all $s\neq s'$ unless both are 1. 
Then, 
$$E[D_{i}(s)| \theta_i=\theta, W_i=w]=  0 \text{ if } \tau_s < MID_{\theta s}.$$ 
Otherwise, 
\begin{align*}\label{simplescore}
&\lim_{\delta\rightarrow 0}E[D_{i}(s)| \theta_i=\theta, W_i=w, R_i \le MID_{\theta s}-\delta ] =\lim_{\delta\rightarrow 0}E[D_{i}(s)| \theta_i=\theta, W_i=w, R_i>\tau_s+\delta) ] = 0, \\
&\lim_{\delta\rightarrow 0}E[D_{i}(s)| \theta_i=\theta, W_i=w, R_i \in (MID_{\theta s}+\delta, \tau_s-\delta] ] = 1, \\
&\lim_{\delta\rightarrow 0}E[D_{i}(s)| \theta_i=\theta, W_i=w, R_i \in (MID_{\theta s}-\delta,MID_{\theta s}+\delta] ] \\
& \quad \ = \lim_{\delta\rightarrow 0}E[D_{i}(s)| \theta_i=\theta, W_i=w, R_i \in (\tau_s-\delta,\tau_s+\delta]] = 0.5.
\end{align*}
\end{proposition}
\noindent This follows from a more general result for DA presented in the next section.

Proposition \ref{local_pscore} describes a key conditional independence result: the limiting local probability of seat assignment in serial dictatorship takes on only three values and is unrelated to applicant characteristics.    
Note that the cases enumerated in the proposition (when $\tau_s > MID_{\theta s}$) partition the tie-breaker line as sketched in Figure \ref{fig:SDsketch}.  
Applicants with tie-breaker values above the cutoff at $s$ are disqualified at $s$ and so cannot be seated there, while applicants with tie-breaker values below $MID_{\theta s}$ are qualified at a school they prefer to $s$ and so will be seated elsewhere.  Applicants with tie-breakers strictly between $MID_{\theta s}$ and $\tau_s$ are surely assigned to $s$. Finally, type $\theta$ applicants with tie-breakers near either $MID_{\theta s}$ or the cutoff at $s$ are seated with probability approximately equal to one-half.   Nearness in this case means inside the interval defined by bandwidth $\delta$.

The driving force behind Proposition \ref{local_pscore} is the assumption that the tie-breaker distribution is continuously differentiable.  In a shrinking window, the tie-breaker density therefore approaches that of a uniform distribution, so the limiting qualification rate is one-half (See \cite{mdrd2:17} or \cite{bugni2018testing} for formal proof of this claim). The assumption of a continuously differentiable tie-breaker distribution is analogous to the continuous running variable assumption invoked in \cite{lee2008randomized} and to a local smoothness assumption in \cite{dong2018alternative}. Continuity of tie-breaker distributions implies a weaker smoothness condition asserting continuity at cutoffs of the conditional expectation functions of potential outcomes given running variables. We favor the stronger continuity assumption because the implied local random assignment provides a scaffold for construction of assignment probabilities in more complicated matching scenarios.\footnote{The connection between continuity of running variable distributions and conditional expectation functions is noted by \cite{dong2018alternative} and \cite{arai2019testing}. Antecedents for the local random assignment idea include an unpublished appendix to \cite{frolich2007regression} and an unpublished draft of \cite{frandsen2017party}, which shows something similar for an asymmetric bandwidth. See also \cite{cattaneo2015randomization} and \cite{frolich/huber2019}.} 

\section{The Local DA Propensity Score}\label{sec:multiscore}

Many school districts assign seats using a version of \textit{student-proposing DA}, which can be described like this:
\begin{quote}
	Each applicant proposes to his or her most preferred school. Each school ranks these
	proposals, first by priority then by tie-breaker within priority groups,
	\textit{provisionally} admitting the highest-ranked applicants in this order up to
	its capacity. Other applicants are rejected.
	
	Each rejected applicant proposes to his or her next most preferred school. Each school
	ranks these new proposals \textit{together with applicants admitted provisionally in the previous round}, first by priority and then by
	tie-breaker. From this pool, the school again provisionally admits those ranked highest
	up to capacity, rejecting the rest.
	
	The algorithm terminates when there are no new proposals (some applicants may remain unassigned).
\end{quote}

Different schools may use different tie-breakers. For example, the NYC high school match includes a diverse set of screened schools \citep{apr:05,abdulkadiroglu/pathak/roth:09}.  These schools order applicants using school-specific tie-breakers that are derived from interviews, auditions, or GPA in earlier grades, as well as test scores. The NYC match also includes many unscreened schools, referred to here as lottery schools, that use a uniformly distributed lottery number as tie-breaker. Lottery numbers are distributed independently of type and potential outcomes, but non-lottery tie-breakers like entrance exam scores almost certainly depend on these variables.  


\subsection{Key Assumptions and Main Theorem} 
%

We adopt the convention that tie-breaker indices are ordered such that lottery tie-breakers come first.  That is, 
$v\in\{1,...,U\}$, where $U\leq V$, indexes $U$ lottery tie-breakers. 
Each lottery tie-breaker, $R_{iv}$ for $v=\{1,...U\}$, is uniformly distributed over $[0,1]$.
Non-lottery tie-breakers are indexed by $v \in \{U+1,...,V\}$.
The assumptions employed with general tie-breaking are summarized as follows: 
\begin{assumption}\label{assumption_smooth} 
~\begin{enumerate}[(i)] \item For any tie-breaker indexed by $v\in\{1,...,V\}$ and applicants $i \neq j$, tie-breakers $R_{iv}$ and $R_{jv}$ are independent, though not necessarily identically distributed.  
\item The unconditional joint distribution of non-lottery tie-breakers $\{R_{iv};  v=U+1,..., V\}$ for applicant $i$ is continuously differentiable with positive density over $[0,1]$.
\end{enumerate}
\end{assumption}
Let $v(s)$ be a function that returns the index of the tie-breaker used at school $s$.  By definition, $s\in S_{v(s)}$.  To combine applicants' priority status and tie-breaking variables into a single number for each school, we define  \textit{applicant
position} at school $s$ as:
\[
\pi_{is} = \rho_{is}+R_{iv(s)}.
\]
Since the difference between any two priorities is at least 1 and tie-breaking variables are between 0 and 1, applicant
order by position at $s$ is lexicographic, first by priority then by tie-breaker. As noted in the discussion of serial dictatorship, we distinguish between tie-breakers and priorities because the latter are fixed, while the former are random variables. 

We also generalize cutoffs to incorporate priorities; these \textit{DA cutoffs} are denoted $\xi_s$.  For any school $s$ that ends up filled to capacity, $\xi_s$ is given by $\max_i\{\pi_{is}| D_i(s)=1\}$. Otherwise, we set
$\xi_s=K+1$ to indicate that $s$ has slack (recall that $K$ is the lowest possible priority).

DA assigns a seat at school $s$ to any applicant $i$ ranking $s$ who has
\begin{equation}
\pi_{is} \leq \xi_s \text{ and } \pi_{ib} > \xi_b \text{ for all } b \succ_i s.\label{eq:core}
\end{equation}
This is a consequence of the fact that the student-proposing DA is stable.\footnote{In particular, if an applicant is seated at $s$ but prefers $b$, she must be qualified at $s$ and not have been assigned to $b$.  Moreover, since DA-generated assignments at $b$ are made in order of position, applicants not assigned to $b$ must be disqualified there.}
In large markets, $\xi_s$ is fixed as tie breakers are drawn and re-drawn.  DA-induced school assignment rates are therefore determined by the distribution of stochastic tie-breakers evaluated at fixed school cutoffs. Condition (\ref{eq:core}) nests our characterization of seat assignment under serial dictatorship since we can set $\rho_{is}=0$ for all applicants and use a single tie-breaker to determine position.  Statement (\ref{eq:core}) then says that $R_i\leq\tau_s$ and
$R_{i}>MID_{\theta s}$ for applicants with $\theta_i=\theta$. 


The DA propensity score is the probability of the event described by \eqref{eq:core}. This probability is determined in part by \textit{marginal priority} at school $s$, denoted $\rho_s$ and defined as $\text{int}(\xi_s)$, the integer part of the DA cutoff.  Conditional on rejection by all preferred schools, applicants to $s$ are assigned $s$ with certainty if $\rho_{is}<\rho_{s}$, that is, if they clear marginal priority. Applicants with $\rho_{is}>\rho_{s}$ have no chance of finding a seat at
$s$. Applicants for whom $\rho_{is} = \rho_{s}$ are marginal: these applicants are seated at $s$ when their tie-breaker values fall below tie-breaker cutoff $\tau_{s}$. This quantity can therefore be written as the decimal part of the DA cutoff: 
$$\tau_s=\xi_{s}-\rho_{s}.$$
Applicants with marginal priority have $\rho_{is} = \rho_{s}$, so
$\pi_{is} \leq \xi_{s} \Leftrightarrow R_{iv(s)} \leq \tau_s.$


	
	


In addition to marginal priority, the local DA propensity score is conditioned on applicant position relative to screened school cutoffs.  To describe this conditioning, define a set of variables, $t_{is}(\delta)$, as follows:
\begin{equation*} 
\begin{split}
t_{is}(\delta)= 
&
\left\{
\begin{array}
[c]{ll}
n & \text{ if }\rho_{\theta s}>\rho_{s} \text{ or, if } v(s) > U, \: \rho_{\theta s}=\rho_{s} \text{ and } R_{iv(s)} > \tau_{s} + \delta \\
a & \text{ if }\rho_{\theta s}<\rho_{s} \text{ or, if } v(s) > U, \: \rho_{\theta s}=\rho_{s} \text{ and } R_{iv(s)} \leq \tau_{s} - \delta \\
c & \text{ if }\rho_{\theta s}=\rho_{s}\text{ and, if } v(s) > U, \, R_{iv(s)} \in (\tau_{s} - \delta,\tau_{s} + \delta], \\
\end{array}
\right.
\end{split}
\end{equation*}
where the mnemonic value labels $n,a,c$ stand for \textit{never seated, always seated, and conditionally seated}. 
It's convenient to collect these variables in a vector,
\[
T_{i}(\delta)=[t_{i 1}(\delta), ..., t_{i s}(\delta), ..., t_{i S}(\delta)].
\]

Elements of $T_i(\delta)$ for unscreened schools are a function only of the partition of types determined by marginal priority.  For screened schools, however, $T_i(\delta)$ also encodes the relationship between tie-breakers and cutoffs. Never-seated applicants to $s$ cannot be seated there, either because they fail to clear marginal priority at $s$ or because they're too far above the cutoff when $s$ is screened.  Always-seated applicants to $s$ are assigned $s$ for sure when they can't do better, either because they clear marginal priority at $s$ or because they're well below the cutoff at $s$ when $s$ is screened.  Finally, conditionally-seated applicants to $s$ are randomized marginal priority applicants.  Randomization is by lottery number when $s$ is a lottery school or by non-lottery tie-breaker within the bandwidth when $s$ is screened. 

With this machinery in hand, the \textit{local DA propensity score} is defined as follows:
\[
\psi_s(\theta,T)=\lim_{\delta\rightarrow 0} E[D_i(s)|\theta_{i}=\theta,T_i(\delta)=T],
\]
for $T=[t_{1}, ..., t_{s}, ..., t_{S}]$ where $t_{s}\in\{n,a,c\}$ for each $s$. This describes assignment probabilities as a function of type and cutoff proximity at each school.  As in Proposition \ref{local_pscore}, formal characterization of $\psi_s(\theta,T)$ requires cutoffs be distinct:

\begin{assumption}\label{assumption_differentiability}
	$\tau_s\neq \tau_{s'}$ for all $s\neq s'$ unless both are 1.
\end{assumption}

The formula characterizing $\psi_s(\theta,T)$ builds on an extension of the $MID$ idea to a general tie-breaking regime. First, the set of schools $\theta$ prefers to $s$, $B_{\theta s}$, is partitioned by tie-breakers by defining $B^{v}_{\theta s}\equiv\{b \in S_v  \hspace*{.1cm} | \hspace*{.1cm}  b \succ_{\theta}s\}$ for each $v$. We then have:
\begin{align*}
MID^{v}_{\theta s}=\left\{
\begin{array}
[c]{ll}%
0 &\text{ if } \rho_{\theta b}>\rho_{b} \text{ for all } b \in  B^{v}_{\theta s} \text{ or if } B_{\theta s}^v = \emptyset\\
1 &\text{ if } \rho_{\theta b}<\rho_{b} \text{ for some } b \in B^{v}_{\theta s} \\
\max\{ \tau_{b}  \hspace*{.1cm} | \hspace*{.1cm} b\in B^{v}_{\theta s} \text{ and } \rho_{\theta b}=\rho_{b} \} & \text{ otherwise. }
\end{array}
\right.
\end{align*}
\noindent  $MID^{v}_{\theta s}$ quantifies the extent to which qualification at schools using tie-breaker $v(s)$ and that type $\theta$ applicants prefer to $s$ truncates the tie-breaker distribution among those contending for seats at s.

Next, define:
\[
m_{s}(\theta,T) = |\{v>U: MID^{v}_{\theta s}=\tau_b \text{ and } t_b=c \text{ for some } b\in B^{v}_{\theta s} \}|.
\]
This quantity counts the number of RD-style experiments created by the screened schools that type $\theta$ prefers to $s$. 

The last preliminary to a formulation of local DA assignment scores uses $MID^{v}_{\theta s}$ and $m_{s}(\theta,T)$
to compute disqualification rates at all schools preferred to $s$.  We break this into two pieces: variation generated by screened schools and variation generated by lottery schools. As the bandwidth shrinks, the limiting disqualification probability at screened schools in $B_{\theta s}$ converges to
\begin{equation}
\sigma_s(\theta,T) = 0.5^{m_{s}(\theta,T)}.\label{eq:sigma} 
\end{equation}
The disqualification probability at lottery schools in $B_{\theta s}$ is 
\begin{equation}\label{eq:lotto_risk}
\lambda_s(\theta)=\prod_{v=1}^{U}(1-MID^{v}_{\theta s}),
\end{equation}
without regard to bandwidth.

To recap: the local DA score for type $\theta$ applicants is determined in part by the screened schools $\theta$ prefers to $s$. Relevant screened schools are those determining $MID^{v}_{\theta s}$, and at which applicants are close to tie-breaker cutoffs. The variable $m_{s}(\theta,T)$ counts the number of tie-breakers involved in such close encounters. Applicants drawing screened school tie-breakers close to $\tau_b$ for some $b\in B^{v}_{\theta s}$ face qualification rates of $0.5$ for each tie-breaker $v$. Since screened school disqualification is locally independent over tie-breakers, the term $\sigma_s(\theta,T)$ computes the probability of not being assigned a screened school preferred to $s$. Likewise, since the qualification rate at preferred lottery schools is $MID^{v}_{\theta s}$, the term $\lambda_s(\theta)$ computes the probability of not being assigned a lottery school preferred to $s$.

The following theorem combines these in a formula for the local DA propensity score:

\begin{theorem}[Local DA Propensity Score with General Tie-breaking]\label{theorem:local}
	Suppose seats in a large market are assigned by DA with tie-breakers indexed by $v$, and suppose Assumptions \ref{assumption_smooth} and \ref{assumption_differentiability} hold. 	For all schools $s, \theta,$ $T$ and $w$, we have 
	$$\psi_s(\theta,T)=\lim_{\delta \rightarrow 0} E[D_i(s)| \theta_i=\theta,T_i(\delta)=T, W_i=w]=0,$$
	if (a) $t_{s}=n$; or (b) $t_b=a \ \text{ for some } b \in B_{\theta s}$. Otherwise,
	\begin{equation}\label{equation_main_theorem2}
	\psi_s(\theta,T)=\left\{
	\begin{array}
	[c]{ll}
	\sigma_s(\theta,T)\lambda_s(\theta) & \text{ if } t_{s}=a\\
	\sigma_s(\theta,T)\lambda_s(\theta)\max\left\{ 0,\dfrac{\tau_{s}-MID^{v(s)}_{\theta s}}{1-MID^{v(s)}_{\theta s}}\right\} & \text{ if } t_{s}=c\text{ and }v(s)\leq U\\
	\sigma_s(\theta,T)\lambda_s(\theta) \times 0.5 & \text{ if } t_{s}=c\text{ and }v(s)> U.\\
	\end{array}
	\right.
	\end{equation}
\end{theorem}

\bigskip

Theorem \ref{theorem:local} starts with a scenario where applicants to $s$ are either disqualified there or assigned to a preferred school for sure.\footnote{See the appendix for proof of the Theorem, along with other theoretical results, including derivation of a non-limit form of the DA propensity score.}  In this case, we need not worry about whether $s$ is a screened or lottery school.  In other scenarios where applicants are surely qualified at $s$, the probability of assignment to $s$ is determined entirely by disqualification rates at preferred screened schools and by truncation of lottery tie-breaker distributions at preferred lottery schools.  These sources of assignment risk combine to produce the first line of \eqref{equation_main_theorem2}.   The conditional assignment probability at any lottery $s$, described on the second line of  \eqref{equation_main_theorem2}, is determined by the disqualification rate at preferred schools and the qualification rate at $s$, where the latter is given by $\tau_{s}-MID^{v(s)}_{\theta s}$ (to see this, note that $\lambda_s(\theta)$ includes the term $1-MID^{v(s)}$ in the product over lottery tie-breakers).  Similarly, the conditional assignment probability at any screened $s$, on the third line of  \eqref{equation_main_theorem2}, is determined by the disqualification rate at preferred schools and the qualification rate at $s$, where the latter is given by $0.5$.  

The Theorem covers the non-lottery tie-breaking serial dictatorship scenario in the previous section. With a single non-lottery tie-breaker, $\lambda_s(\theta)=1$.  When $t_s=n$ or $t_b=a$ for some $b \in B_{\theta s}$, the local propensity score at $s$
is zero. Otherwise, suppose $t_b=n$ for all $b \in B_{\theta s}$, so that $m_s(\theta,T)=0$.   If $t_s=a$,
then the local propensity score is $1$.  If $t_s=c$, then the local propensity score is $0.5$.  Suppose, instead, that $MID_{\theta s} = \tau_b$ for some $b \in B_{\theta s}$, so
that $m_s(\theta,T)=1$.  In this case, $t_s\neq c$ because cutoffs are distinct.  If $t_s=a$, then the local propensity score is $0.5$. 
Online Appendix \ref{appendix_illustration} uses an example to illustrate the Theorem in other scenarios.  

\subsection{Score Estimation}\label{section:convergence_result}

Theorem \ref{theorem:local} characterizes the theoretical probability of school assignment in a large market with a continuum of applicants. In reality, of course, the number of applicants is finite and propensity scores must be estimated.  We show here that, in an asymptotic sequence that increases market size with a shrinking bandwidth, a sample analog of the local DA score described by Theorem \ref{theorem:local} converges uniformly to the corresponding local score for a finite market. Our empirical application establishes the relevance of this asymptotic result by showing that applicant characteristics are balanced by assignment status conditional on estimates of the local DA propensity score.

The asymptotic sequence for the estimated local DA score works as follows: randomly sample $N$ applicants from a continuum economy. 
The applicant sample (of size $N$) includes information on each applicant's type and the vector of large-market school capacities, $q_s$, which give the proportion of $N$ seats that can be seated at $s$.  We observe realized tie-breaker values for each applicant, but not the underlying distribution of non-lottery tie-breakers. The set of finitely many schools is unchanged along this sequence. 

Fix the number of seats at school $s$ in a sampled finite market to be the integer part of $Nq_s$ and run DA with these applicants and schools.  We consider the limiting behavior of an estimator computed using the estimated $\hat{MID}^v_{\theta_i s}$, $\hat \tau_s$, and marginal priorities generated by this single realization.  Also, given a bandwidth $\delta_N>0$, we compute $t_{is}(\delta_N)$ for each $i$ and $s$, collecting these in vector $T_i(\delta_N)$. 
These statistics then determine:
$$\hat{m}_{s}(\theta_i, T_i(\delta_N))=|\{v>U: \hat{MID}^{v}_{\theta_i s}= \hat \tau_b \text{ and }  t_{ib}(\delta_N)=c \text{ for some } b\in B^{v}_{\theta_i s} \}|.$$


Our local DA score estimator, denoted $\hat \psi_{s}(\theta_i,T_i(\delta_N))$, is constructed by plugging these ingredients into the formula in Theorem \ref{theorem:local}. That is, if (a) $\hat t_{is}(\delta_N)=n$; or (b) $\hat t_{ib}(\delta_N)=a \ \text{ for some } b \in B_{\theta_i s}$, then $\hat \psi_{s}(\theta_i,T_i(\delta_N))=0.$ 
Otherwise,
	\begin{align}
	 \hat \psi_{s}(\theta_i,T_i(\delta_N)) = \qquad \qquad \qquad \qquad \qquad \qquad \qquad \qquad \qquad  \qquad \qquad \qquad \qquad \qquad & \nonumber \\
	\left\{
	\begin{array}
	[c]{ll}
	 \hat \sigma_s(\theta_i,T_i(\delta_N))\hat \lambda_s(\theta_i) &\text{if } t_{is}(\delta_N)=a\\
	 \hat \sigma_s(\theta_i,T_i(\delta_N))\hat \lambda_s(\theta_i)\max\left\{ 0,\frac{\hat \tau_{s}-\hat{MID}^{v(s)}_{\theta_i s}}{1-\hat{MID}^{v(s)}_{\theta_i s}}\right\} &\text{if }   t_{is}(\delta_N)=c\text{ and }v(s)\leq U\\
	 \hat \sigma_s(\theta_i,T_i(\delta_N)) \hat \lambda_s(\theta_i) \times 0.5 &\text{if }   t_{is}(\delta_N)=c\text{ and }v(s)> U,\\
	\end{array}
	\right.
	\end{align}
\noindent where 
$$\hat \sigma_s(\theta_i,T_i(\delta_N)) = 0.5^{\hat{m}_{s}(\theta_i, T_i(\delta_N))} $$
and
$$\hat \lambda_s(\theta_i)=\prod_{v=1}^{U}(1-\hat{MID}^{v}_{\theta_i s}).$$

As a theoretical benchmark for the large-sample performance of $\hat \psi_{s}$, consider the true local DA score for a finite market of size $N$.  This is 
\begin{align}
\psi_{Ns}(\theta, T) = \lim_{\delta\rightarrow 0} E_N[D_i(s)| \theta_i=\theta, T_i(\delta)=T],\label{eq:finitepsi}
\end{align}
\noindent where $E_N$ is the expectation induced by the joint tie-breaker distribution for applicants in the finite market. This quantity is defined by fixing the distribution of types and the vector of proportional school capacities, as well as market size.  $\psi_{Ns}(\theta, T)$ is then the limit of the average of $D_i(s)$ across infinitely many tie-breaker draws in ever-narrowing bandwidths for this finite market.  Because tie-breaker distributions are assumed to have continuous density in the neighborhood of any cutoff, the finite-market local propensity score is well-defined for any positive $\delta$.  

We're interested in the gap between the estimator $\hat \psi_{s}(\theta,T(\delta_N))$ and the true local score $\psi_{Ns}(\theta, T)$ as $N$ grows and $\delta_N$ shrinks. 
We show below that $\hat \psi_{s}(\theta,T(\delta_N))$ converges uniformly to 
$\psi_{Ns}(\theta, T)$ in our asymptotic sequence.

This result uses a regularity condition:

\begin{assumption}\label{assumption_types} (Rich support)
In the population continuum market, for every school $s$ and every priority $\rho$ held by a positive mass of applicants who rank $s$, the proportion of applicants $i$ with $\rho_{is}=\rho$ who rank $s$ first is also positive. 
\end{assumption}


Uniform convergence of $\hat \psi_{s}(\theta,T(\delta_N))$ is formalized below:

\begin{theorem}[Consistency of the Estimated Local DA Propensity Score]
\label{theorem:main}
In the asymptotic sequence described above, and maintaining Assumptions \ref{assumption_smooth}-\ref{assumption_types}, the estimated local DA propensity score $\hat \psi_{s}(\theta,T(\delta_N))$ is a consistent
estimator of $\psi_{Ns}(\theta, T)$ 
in the following sense:
For any $\delta_N$ such that $\delta_N\rightarrow 0$, $N\delta_N\rightarrow \infty,$ and $T(\delta_n) \rightarrow T$,
\[
\sup_{\theta, s, T}|\hat \psi_{s}(\theta,T(\delta_N))-\psi_{Ns}(\theta, T)| \overset{p}{\longrightarrow}0,
\]
as $N\rightarrow\infty$. 
\end{theorem}

\noindent This result (proved in the online appendix) justifies conditioning on an estimated local propensity score to eliminate omitted variables bias in school attendance effect estimates. 

\subsection{Treatment Effect Estimation}

Theorems \ref{theorem:local} and \ref{theorem:main} provide a foundation for causal inference. In combination with an exclusion restriction discussed below, these results imply that a dummy variable indicating Grade A assignments is asymptotically independent of potential outcomes (represented by the residuals in a equation \eqref{RF}), conditional on an estimate of the Grade A local propensity score.  Let $S_A$ denote the set of Grade A schools. Because DA generates a single offer, the local propensity score for Grade A assignment can be computed as:
$$\hat \psi_{A}(\theta_i,T_i(\delta_N))=\sum_{s\in S_A} \hat \psi_{s}(\theta_i,T_i(\delta_N)).$$
In other words, the local score for Grade A assignment is the sum of the scores for all Grade A schools in the match.

These considerations lead to a 2SLS estimator with second and first stage equations that can be written in stylized form as:
\begin{align}
Y_i & = \beta C_{i} + \sum_{x} \alpha_{2}(x)d_{i}(x) + g_2(\mathcal{R}_i; \delta_N) + \eta_i\label{implement2}
\end{align}
\begin{align}
C_{i} &=  \gamma D_{Ai} + \sum_{x} \alpha_{1}(x)d_{i}(x) + g_1(\mathcal{R}_i; \delta_N) + \nu_{i},\label{implement1}
\end{align}
where $d_i(x)=1\{\hat \psi_{A}(\theta_i,T_i(\delta_N))=x\}$ and
the set of parameters denoted $\alpha_2(x)$ and $\alpha_1(x)$ provide saturated control for the local propensity score.  
As detailed in the next section, functions $g_2(\mathcal{R}_i; \delta_N)$ and $g_1(\mathcal{R}_i; \delta_N)$ implement local linear control for screened school tie-breakers for applicants to these schools with $\hat{t}_{is}(\delta_N)=c$.  Linking this with the empirical strategy sketched at the outset, equation \eqref{implement2} is a version of of equation \eqref{RF} that sets 
$$f_2(\theta_i, \mathcal{R}_i, \delta) = \sum_{x} \alpha_{2}(x)d_{i}(x) + g_2(\mathcal{R}_i; \delta_N).$$
Likewise, equation \eqref{implement1} is a version of equation \eqref{FS} with $f_1(\theta_i, \mathcal{R}_i, \delta)$ defined similarly. 

Our implementation of score-controlled instrumental variables is inspired by the \cite{calonico_et_al:2019_RD} analysis of RD designs with covariates. Using a mix of simulation evidence and theoretical reasoning, \cite{calonico_et_al:2019_RD} argues that additive control for covariates in a local linear regression model requires fewer assumptions and is likely to have better finite-sample behavior than more elaborate procedures. The covariates of interest to us are a full set of dummies for values in the support of the Grade A local propensity score. We'd like to control for these while also benefiting from the good performance of local linear regression estimators of conditional mean functions near cutoffs.\footnote{\cite{calonico_et_al:2019_RD} discuss both sharp and fuzzy RD designs.  The conclusions for sharp design carry over to the fuzzy case in which cutoff clearance is used as an instrument. Equations \eqref{implement2} and \eqref{implement1} are said to be stylized because they omit a number of implementation details supplied in the following section.}

Note that saturated regression-conditioning on the local propensity score eliminates applicants with score values of zero or one.  This is apparent from an analogy with a fixed-effects panel model.  In panel data with multiple annual observations on individuals, estimation with individual fixed effects is equivalent to estimation after subtracting person means from regressors.  Here, the ``fixed effects'' are coefficients on dummies for each possible score value.  When the score value is 0 or 1 for applicants of a given type, assignment status is constant and observations on applicants of this type drop out. We therefore say an applicant \textit{has Grade A risk} when $\hat \psi_{A}(\theta_i,T_i(\delta_N)) \in (0, 1)$.  The sample with risk contributes to parameter estimation in models with saturated score control. 

Propensity score conditioning facilitates control for applicant type in the sample with risk. In practice, local propensity score conditioning yields considerable dimension reduction compared to full-type conditioning, as we would hope. The 2014 NYC high school match, for example, involved 52,124 applicants of 47,074 distinct types.  Of these, 42,461 types listed a Grade A school on their application to the high school match.  By contrast, the local propensity score for Grade A school assignment takes on only 2,054 values.  

\section{A Brief Report on NYC Report Cards}\label{sec:NYC}

\subsection{Doing DA in the Big Apple}
Since the 2003-04 school year, the NYC Department of Education (DOE) has used DA to assign rising ninth graders to high schools.  Many high schools in the match host multiple programs, each with their own admissions protocols.   Applicants are matched to programs rather than schools.  Each
applicant for a ninth grade seat can rank up to twelve programs.  All traditional public high schools participate in the match, but
charter schools and NYC's specialized exam high schools have separate admissions
procedures.\footnote{Some special needs students are also matched separately. The centralized NYC
high school match is detailed in \cite{apr:05,abdulkadiroglu/pathak/roth:09}. \cite{abdulkadiroglu/angrist/pathak:14} describe NYC exam school admissions.}

The NYC match is structured like the general DA match described in Section \ref{sec:multiscore}: lottery programs use a common uniformly distributed lottery number, while screened programs use a variety of non-lottery tie-breaking variables. Screened tie-breakers are mostly distinct, with one for each school or program, though some screened programs share a tie-breaker.  In any case, our theoretical framework accommodates all of NYC's many tie-breaking protocols.\footnote{Screened tie-breakers are reported as an integer variable encoding the underlying tie-breaker order such as a test score or portfolio summary score. We scale these so as to lie in $(0, 1]$ by computing $[R_{iv}-\min_j{R_{jv}}+1]/[\max_jR_{jv}-\min_jR_{jv}+1]$ for each tie-breaker $v$.  This transformation produces a positive cutoff at $s$ when only one applicant is seated at $s$ and a cutoff of 1 when all applicants who rank $s$ are seated there.}

Our analysis uses Theorems \ref{theorem:local} and \ref{theorem:main} to compute propensity scores for programs rather than schools since programs are the unit of assignment. For our purposes, a lottery school is a school hosting any lottery program. Other schools are defined as screened.\footnote{Some NYC high schools sort applicants on a coarse screening tie-breaker that allows ties, breaking these ties using the common lottery number. Schools of this type are treated as lottery schools, with priority groups defined by values of the screened tie-breaker. Seats for NYC's ed-opt programs are allocated to two groups, one of which screens applicants using a single non-lottery tie-breaker and the other using the common lottery number.  The online appendix explains how ed-opt programs are handled by our analysis.}

In 2007, the NYC DOE launched a school accountability system that graded schools from A to F. This mirrors similar accountability systems in Florida and other states. NYC's school grades were determined by achievement levels and, especially, achievement growth, as well as by survey- and attendance-based features of the school environment.  Growth looked at credit accumulation, Regents test completion and pass rates; performance measures were derived mostly from four- and six-year graduation rates.  Some schools were ungraded. Figure \ref{fig:sch_report} reproduces a sample letter-graded school progress report.\footnote{\cite{walcott:12} details the NYC grading methodology used in this period. Note that the computation of the grade of a school for a particular year uses only information from past years, so that there is no feedback between school grades and the school's current outcomes.}

\begin{figure}[!t]
\centering
\caption{Sample NYC School Report Card}
  \includegraphics[width=.85\linewidth]{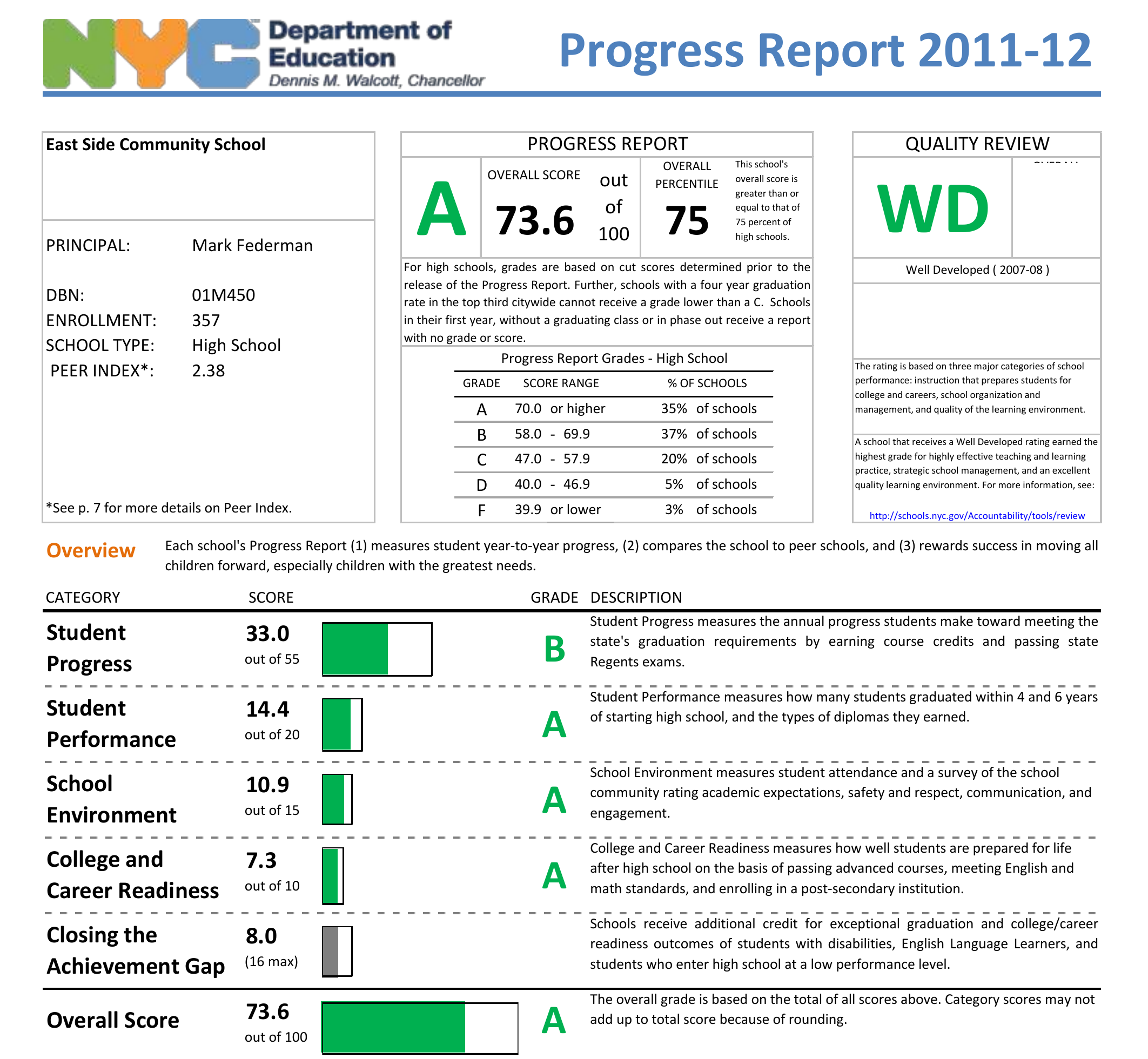} \label{fig:sch_report}
\floatfoot{\footnotesize{\textit{Notes:} This figure shows the 2011/12 progress report for East Side Community School. Source: $www.crpe.org$}}
\end{figure}

The 2007 grading system was controversial.  Proponents applauded the integration of multiple measures of school quality while opponents objected to the high-stakes consequences of low school grades, such as school closure or consolidation. \cite{rockoff/turner:10} provide a partial validation of the system by showing that low grades seem to have sparked school improvement. In 2014, the DOE replaced the 2007 scheme with school quality measures that place less weight on test scores and more on curriculum characteristics and subjective assessments of teaching quality.  The relative merits of the old and new systems continue to be debated.

The results reported here use application data from the 2011-12, 2012-13, and 2013-14 school years (students in these application cohorts enrolled in the following school years). Our sample includes first-time applicants seeking 9th grade seats, who submitted preferences over programs in the main round of the NYC high school match. We obtained data on school capacities and priorities, lottery numbers, and screened school tie-breakers, information that allows us to replicate the match. Details related to match replication appear in the online appendix.\footnote{Our analysis assigns report card grades to a cohort's schools based on the report cards published in the previous year. For the 2011/12 application cohort, for instance, we used the grades published in 2010/11. On the other hand, applicant SAT scores from tests taken before 9th grade are dropped.}

\begin{center}
	\includegraphics[width = .95  \textwidth]{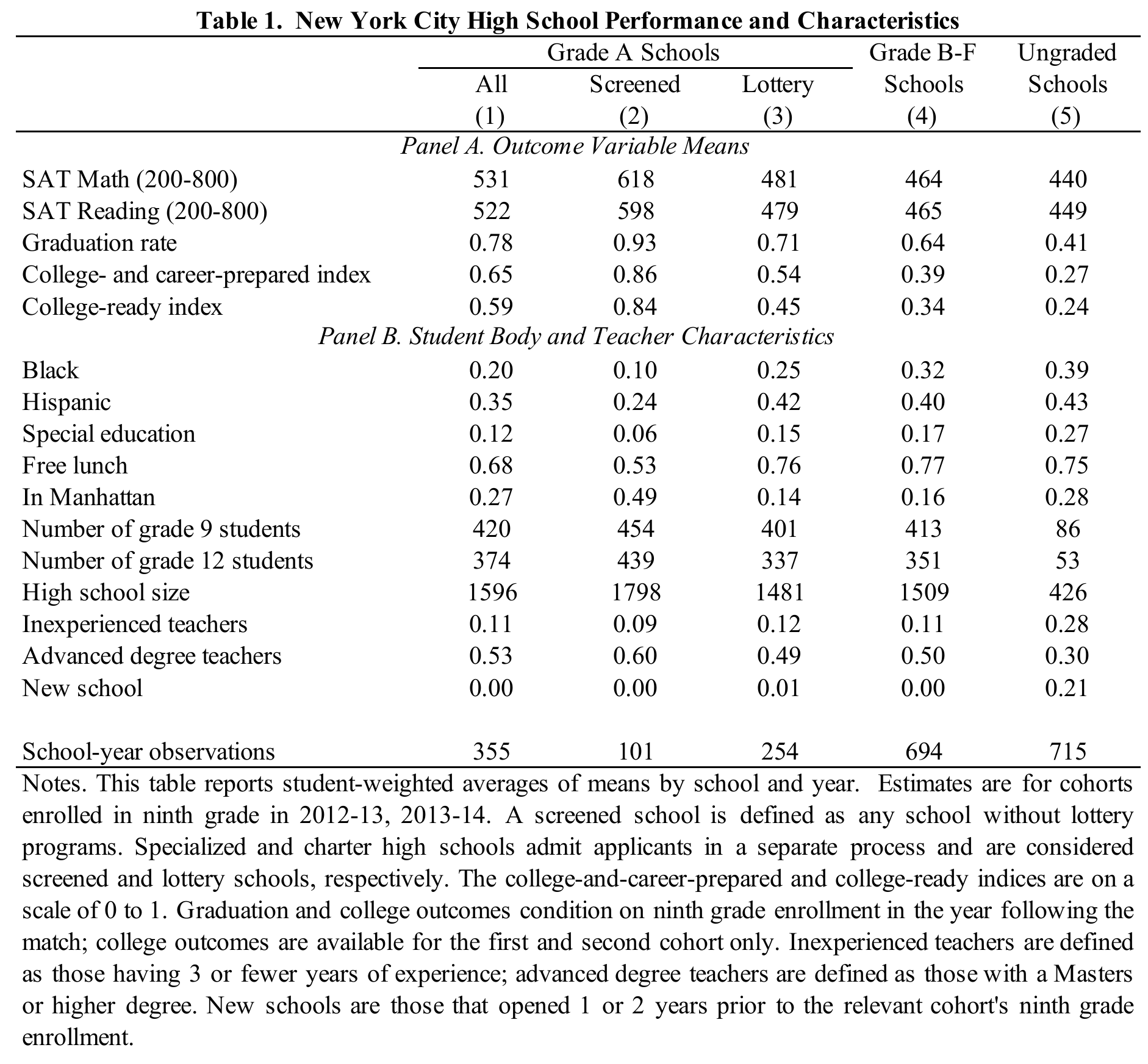} \refstepcounter{tablenums} \label{tab:schchars}
\end{center}

Students at Grade A schools have higher average SAT scores and higher graduation rates than do students at other schools.  Differences in graduation rates across schools feature in popular accounts of socioeconomic differences in school access (see, e.g., \cite{harris/fessenden:17} and \cite{disare:17}). Grade A students are also more likely than students attending other schools to be deemed ``college- and career-prepared'' or ``college-ready.''\footnote{These composite variables are determined as a function of Regents and AP scores, course grades, vocational or arts certification, and college admission tests.}  These and other school characteristics are documented in Table \ref{tab:schchars}, which reports statistics separately by school grade and admissions regime. Achievement gaps between screened and lottery Grade A schools are especially large, likely reflecting selection bias induced by test-based screening.

Screened Grade A schools have a majority white and Asian student body, the only group of schools described in the table to do so (the table reports shares black and Hispanic).  These schools are also over-represented in Manhattan, a borough that includes most of New York's wealthiest neighborhoods (though average family income is higher on Staten Island).  Teacher experience is similar across school types, while screened Grade A schools have somewhat more teachers with advanced degrees.

The first two columns of Table \ref{tab:stuchars} describe the roughly 180,000 ninth graders enrolled in the 2012-13, 2013-14, and 2014-15 school years. Students enrolled in a Grade A school, including those enrolled in the Grade A schools assigned outside the match, are less likely to be black or Hispanic and have higher baseline scores than the general population of 9th graders.  The 153,000 eighth graders who applied for ninth grade seats are described in column 3 of the table.  Roughly 130,000 listed a Grade A school for which seats are assigned in the match on their application form and a little over a third of these were assigned to a Grade A school.\footnote{The difference between total 9th grade enrollment and the number of match participants is accounted for by special education students outside the main match, direct-to-charter enrollment, and a few schools that straddle 9th grade.}  Applicants in the match have baseline scores (from tests taken in 6th grade) above the overall district mean (baseline scores are standardized to the population of test-takers).  As can be seen by comparing columns 3 and 4 in Table \ref{tab:stuchars}, however, the average characteristics of Grade A applicants are mostly similar to those of the entire applicant population.   

The statistics in column 5 of Table \ref{tab:stuchars} show that applicants \textit{enrolled} in a Grade A school (among schools participating in the match) are somewhat less likely to be black and have higher baseline scores than the total applicant pool.  These gaps likely reflect systematic differences in offer rates by race at screened Grade A schools. Column 5 of Table \ref{tab:stuchars} also shows that most of those attending a Grade A school were assigned there, and that most Grade A students ranked a Grade A school first.  Grade A students are about twice as likely to go to a lottery school as to a screened school.  Interestingly, enthusiasm for Grade A schools is far from universal: just under half of all applicants in the match ranked a Grade A school first. 

\begin{center}
	\includegraphics[width = .9\textwidth]{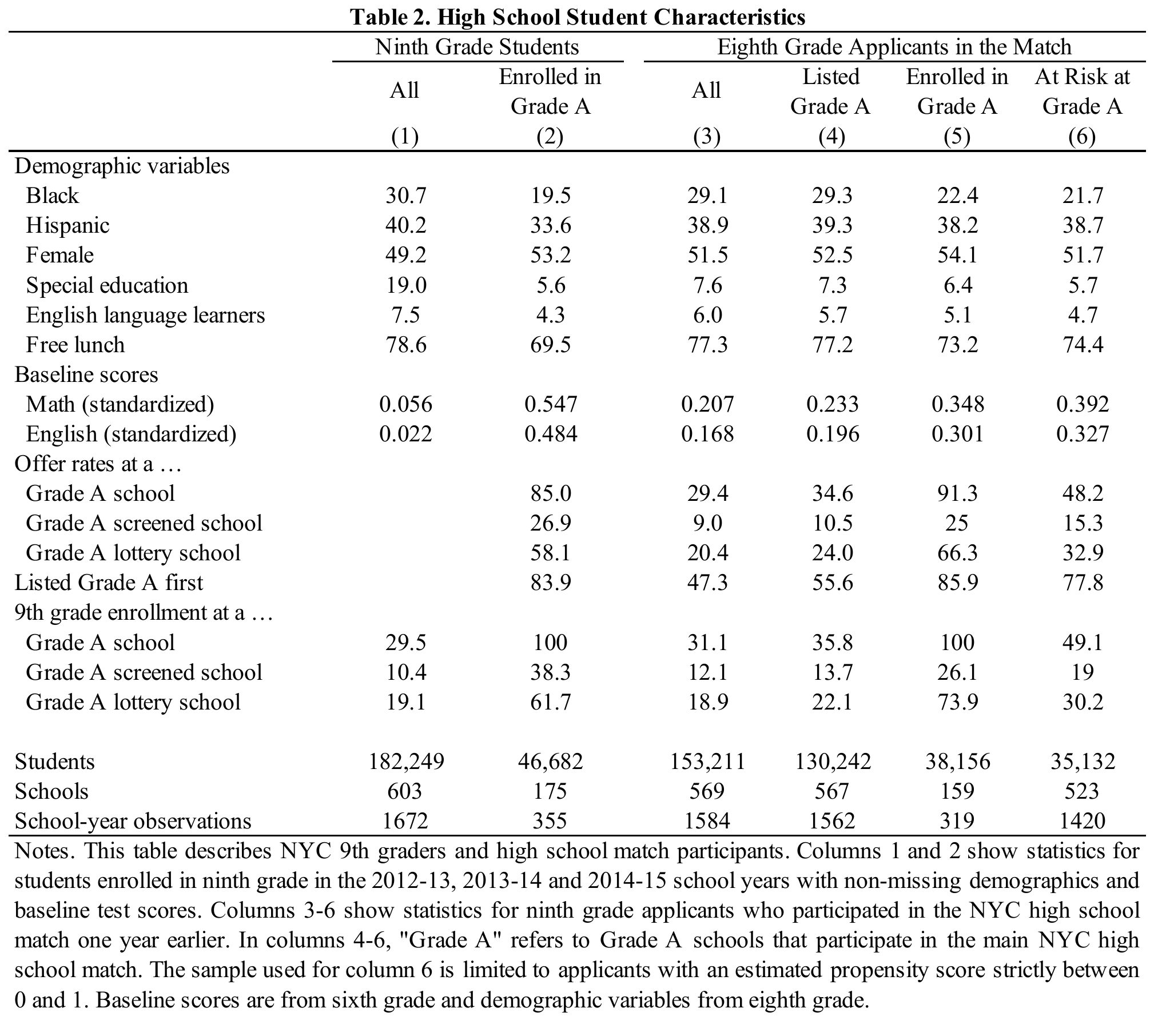} \refstepcounter{tablenums} \label{tab:stuchars}
\end{center}

\subsection{Balance and 2SLS Estimates} 

Because NYC has a single lottery tie-breaker, the disqualification probability at lottery schools in $B_{\theta s}$ described by equation \eqref{eq:lotto_risk} simplifies to
$$\lambda_s(\theta)=(1-MID^{1}_{\theta s}),$$ 
where $MID^{1}_{\theta s}$ is most informative disqualification at schools using the common lottery tie-breaker, $R_{1i}$.
The local DA score described by equation \eqref{equation_main_theorem2} therefore also simplifies, in this case to:
	\begin{equation}\label{equation_main_theorem2_nyc_simplification}
	\psi_s(\theta,T)=\left\{
	\begin{array}
	[c]{ll}
	\sigma_s(\theta,T)(1-MID^{1}_{\theta s}) & \text{ if } t_{s}=a,\\
	\sigma_s(\theta,T)\max\left\{ 0, \tau_{s}-MID^{1}_{\theta s}\right\} & \text{ if } t_{s}=c\text{ and }v(s)=1,\\
	\sigma_s(\theta,T)(1-MID^{1}_{\theta s}) \times 0.5 & \text{ if } t_{s}=c\text{ and }v(s)>1.\\
	\end{array}
	\right.
	\end{equation}
 
Estimates of the local DA score based on \eqref{equation_main_theorem2_nyc_simplification} reveal that roughly 35,000 applicants have Grade A risk, that is, an estimated local DA score value strictly between 0 and 1.  As can be seen in column 6 of Table \ref{tab:stuchars}, applicants with Grade A risk have mean baseline scores and demographic characteristics much like those of the sample enrolled at a Grade A school.  The ratio of screened to lottery enrollment among those with Grade A risk is also similar to the corresponding ratio in the sample of enrolled students (compare 32.9/15.3 in the former group to 66.3/25.0 in the latter). Online Appendix Figure \ref{figappdx:ptsofsupport} plots the distribution of Grade A assignment probabilities for applicants with risk.  The modal probability is $0.5$, reflecting the fact that roughly 25\% of those with Grade A risk rank a single Grade A school and that this school is screened.

The balancing property of local propensity score conditioning is evaluated using score-controlled differences in covariate means for applicants who do and don't receive Grade A assignments. Score-controlled differences by Grade A assignment status are estimated in a model that includes a dummy indicating assignments at ungraded schools as well as a dummy for Grade A assignments, controlling for the propensity scores for both.  We account for ungraded school attendance to ensure that estimated Grade A effects compare schools with high and low grades, omitting the ungraded.\footnote{Ungraded schools were mostly new when grades were assigned or had data insufficient to determine a grade.}
Specifically, let $D_{Ai}$ denote Grade A assignments as before, and let $D_{0i}$ indicate assignments at ungraded schools. Assignment risk for each type of school is controlled using sets of dummies denoted $d_{Ai}(x)$ and $d_{0i}(x)$, respectively, for score values indexed by $x$.  

The covariates of interest here, denoted by $W_i$, are those that are unchanged by school assignment and should therefore be mean-independent of $D_{Ai}$ in the absence of selection bias. The balance test results reported in Table \ref{tab:balance} are estimates of 
parameter $\gamma_A$ in regressions of $W_i$ on $D_{Ai}$ of the form:
\begin{align}
W_i =  \gamma_{A} D_{Ai} + \gamma_0 D_{0i} + \sum_{x} \alpha_{A}(x)d_{Ai}(x) + \sum_{x} \alpha_{0}(x)d_{0i}(x) + g(\mathcal{R}_i; \delta_N) + \nu_i. \label{eq:covbalance}
\end{align}
Local piecewise linear control for screened tie-breakers is parameterized as:
\begin{align}
g(\mathcal{R}_i; \delta_N) =  \sum_{s \in S \backslash S_0} \omega_{1s} a_{is} + k_{is}[\omega_{2s} + \omega_{3s}( R_{iv(s)} - \tau_s )
+ \omega_{4s}( R_{iv(s)}- \tau_s )\textbf{1}( R_{iv(s)} > \tau_s )], \label{rv-control}
\end{align}
where $S\backslash S_0$ is the set of screened programs, $a_{is}$ indicates whether applicant $i$
applied to screened program $s$, and $k_{is} = 1[\hat{t}_{is}(\delta_N)=c]$. 
The sample used to estimate \eqref{eq:covbalance} is limited to applicants with Grade A risk.  



Parameters in \eqref{eq:covbalance} and \eqref{rv-control} vary by application cohort (three cohorts are stacked in the estimation sample). Bandwidths are estimated two ways, as suggested by \cite{imbens/kalyanaraman:12} (IK) using a uniform kernel, and using methods and software described in \cite{rdrobust2017} (CCFT). These bandwidths are computed separately for each program (the notation ignores this), for the set of applicants in the relevant marginal priority group.\footnote{The IK bandwidths used here are identical to those yielded by the IK implementation referenced in \cite{armstrong2018optimal} and distributed via the \textsf{\href{https://github.com/kolesarm/RDHonest}{RDhonest}} package. Bandwidths are computed separately for each outcome variable; we use the smallest of these for each program.  The bandwidth for screened programs is set to zero when there are fewer than five in-bandwidth observations on one or the other side of the relevant cutoff. The control function $g(\mathcal{R}_i; \delta_N)$ is unweighted and can therefore be said to use a uniform kernel. We also explored bandwidths designed to produce balance as in \cite{locrand2016}. These results proved to be sensitive to implementation details such as the p-value used to establish balance.} 

As can be seen in column 2 of Table \ref{tab:balance}, which reports raw differences in means by Grade A assignment status, applicants assigned to a Grade A school are much more likely to have ranked a Grade A school first, and ranked more Grade A schools highly than did other applicants.  These applicants are also more likely
to rank a Screened Grade A school first and among their top three.  Minority and free-lunch-eligible applicants are less likely to be assigned to a Grade A school, while those assigned to a Grade A school have much higher baselines scores, with gaps of $0.3-0.4$ in favor of those assigned.  These raw differences notwithstanding, our theoretical results suggest that estimates of $\gamma_A$ in equation \eqref{eq:covbalance} should be close to zero.

This is borne out by the estimates reported in column 4 of the table, which shows small, mostly insignificant differences in covariates by assignment status when estimated using using \cite{imbens/kalyanaraman:12} bandwidths. The estimated covariate gaps in column 6, computed using \cite{rdrobust2017} bandwidths, are similar. These estimates establish the empirical relevance of both the large-market model of DA and the local DA propensity score derived from it.\footnote{Our balance assessment relies on linear models to estimate mean differences rather than comparisons of distributions. The focus on means is justified because the IV reduced form relationships we aspire to validate are themselves regressions. Recall that in a regression context, reduced form causal effects are unbiased provided omitted variables are mean-independent of the instrument, $D_{Ai}$. Since treatment variable $D_{Ai}$ is a dummy, the regression of omitted control variables on it is given by the difference in conditional control variable means computed with $D_{Ai}$ switched on and off.}

\begin{center}
	\includegraphics[width = .92\textwidth]{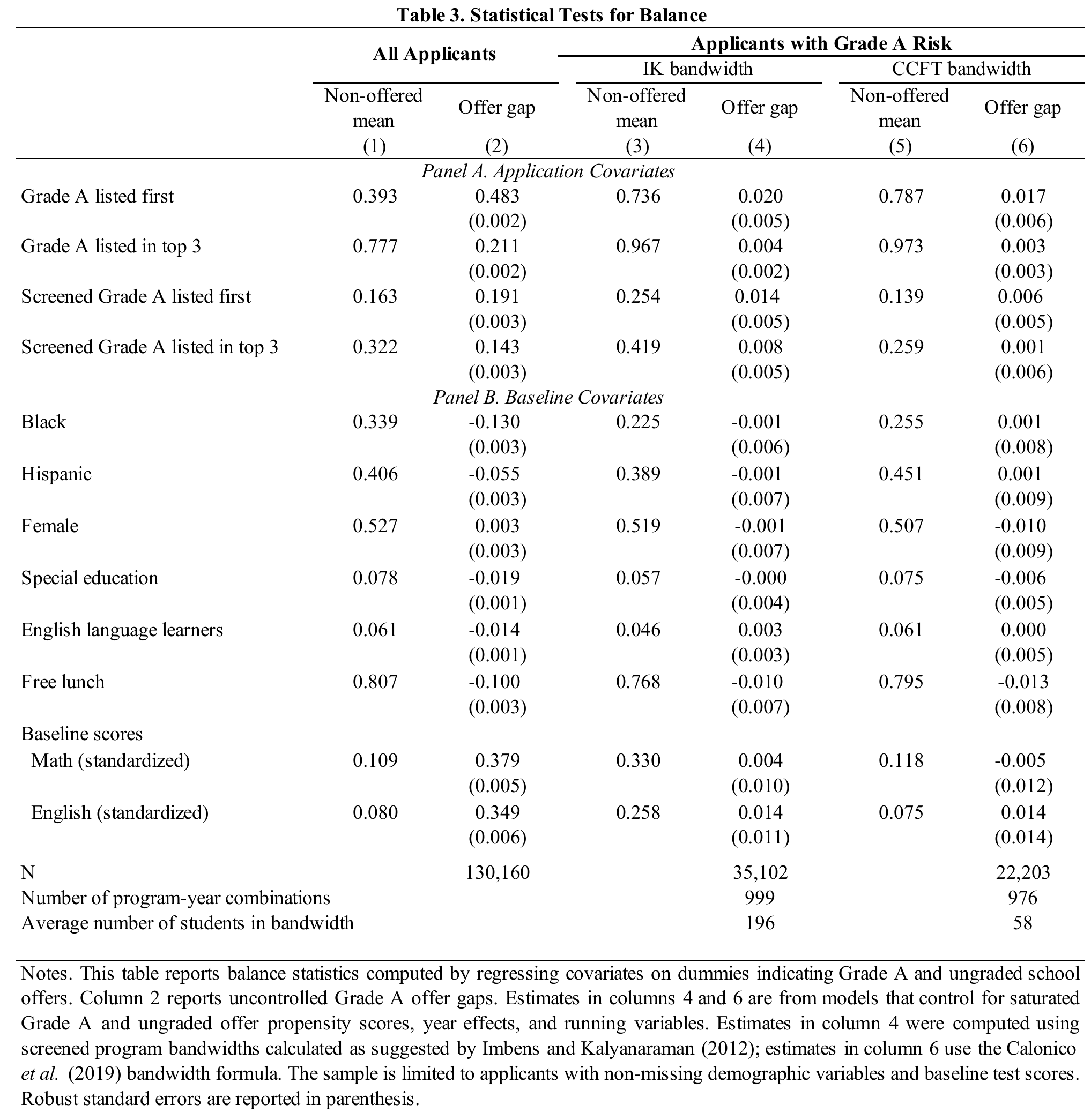} \refstepcounter{tablenums} \label{tab:balance}
\end{center}

Causal effects of Grade A attendance are estimated by 2SLS using assignment dummies as instruments for years of exposure to schools of a particular type, as suggested by equations \eqref{RF} and \eqref{FS}. As in the setup used to establish covariate balance, however, the 2SLS estimating equations include two endogenous variables, $C_{Ai}$ for Grade A exposure and $C_{0i}$ measuring exposure to an ungraded school. Exposure is measured in years for SAT outcomes; otherwise, $C_{Ai}$ and $C_{0i}$ are enrollment dummies. As in equation \eqref{eq:covbalance}, local propensity score controls consist of saturated models for Grade A and ungraded propensity scores, with local linear control for screened tie-breakers as described by equation \eqref{rv-control}.   These equations also control for baseline math and English scores, free lunch, special education, and English language learner dummies, and gender and race dummies (estimates without these controls are similar, though less precise).\footnote{Replacing $W_i$ on the left hand side of \eqref{eq:covbalance} with outcome variable $Y_i$, equations \eqref{eq:covbalance} and \eqref{rv-control} describe the reduced form for our 2SLS estimator. In an application with lottery tie-breaking, \cite{mdrd1:17} compare score-controlled 2SLS estimates with semiparametric instrumental variables estimates based on \cite{abadie:03}. The former are considerably more precise than the latter.}   

OLS estimates of Grade A effects, reported as a benchmark in the second column of Table \ref{tab:2slsmain}, indicate that Grade A attendance is associated with higher SAT scores and graduation rates, as well as increased college and career readiness. The OLS estimates in Table \ref{tab:2slsmain} are from models that omit local propensity score controls, computed in a sample that includes all participants in the high school match without regard to assignment probability. OLS estimates of the SAT gains associated with Grade A enrollment are around 6-7 points.   Estimated graduation gains are similarly modest at   2.4 points, but effects on college and career readiness are substantial, running 7-10 points on a base rate around 40.  

The first stage effects of Grade A assignments on Grade A enrollment, shown in columns 4 and 6 of Panel A in Table \ref{tab:2slsmain}, show that Grade A offers boost Grade A enrollment by about 1.8 years between the time of application and SAT test-taking.  Grade A assignments boost the likelihood of any Grade A enrollment by about 67 percentage points. This can be compared with Grade A enrollment rates of 16-19 percent among those not assigned a Grade A seat in the match.\footnote{The gap between assignment and enrollment arises from several sources. Applicants remaining in the public system may attend charter or non-match exam schools.  Applicants  may also reject a match-based assignment, turning instead to an ad hoc administrative assignment process later in the year.} 

In contrast with the OLS estimates in column 2, the 2SLS estimates shown in columns 4 and 6 of Table \ref{tab:2slsmain} suggest that most of the SAT gains associated with Grade A attendance reflect selection bias. Computed with either bandwidth, 2SLS estimates of SAT math gains are around 2 points, though still (marginally) significant. 2SLS estimates of SAT reading effects are even smaller and not significantly different from zero, though estimated with similar precision. At the same time, the 2SLS estimate for graduation status shows a statistically significant gain of 3-4 percentage points, exceeding the corresponding OLS estimate. The estimated standard error of $0.009$ associated with the graduation estimate in column 4 seems especially noteworthy, as this suggests that our research design has the power to uncover even modest improvements in high school completion rates.\footnote{Estimates reported in Online Appendix Table \ref{tabappdx:attrition} show little difference in follow-up rates between applicants who are and aren't offered a Grade A seat. The 2SLS estimates in Table \ref{tab:2slsmain} are therefore unlikely to be compromised by differential attrition.}

\begin{center}
	\includegraphics[width = .88\textwidth]{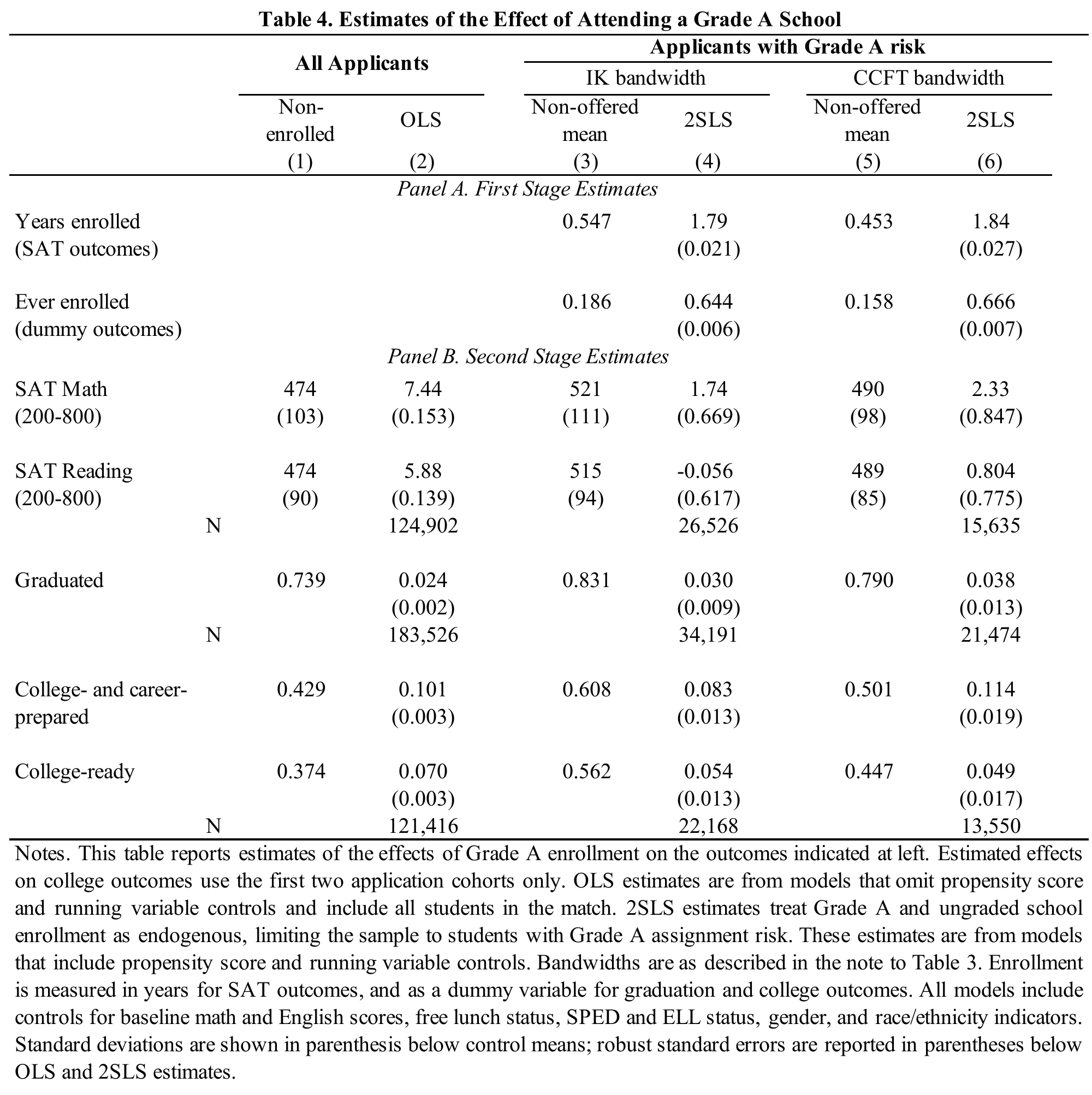} \refstepcounter{tablenums} \label{tab:2slsmain}
\end{center}
 
The strongest Grade A effects appear in estimates of effects on college and career preparedness and college readiness. This may in part reflect the fact that Grade A schools are especially likely to offer advanced courses, the availability of which contributes to the college- and career-related composite outcome variables (the online appendix details the construction of these variables). 2SLS estimates of effects on these outcomes are  mostly close to the corresponding OLS estimates (three out of four are smaller). Here too, switching bandwidth matters little for magnitudes. Throughout Table \ref{tab:2slsmain}, however, 2SLS estimates computed with an IK bandwidth are more precise than those computed using CCFT.

\subsection{Screened vs. Lottery Grade A Effects} 

In New York, education policy discussions often focus on access to academically selective screened schools such as Townsend Harris in Queens, a school consistently ranked among the top American high schools by \textit{U.S. News and World Report}.  Public interest in screened schools motivates an analysis that distinguishes screened from lottery Grade A effects.  The possibility of different effects within the Grade A sector also raises concerns related to the exclusion restriction underpinning a causal interpretation of 2SLS estimates. In the context of our causal model of Grade A effects, the exclusion restriction fails when the offer of a Grade A seat moves applicants between schools of different quality within the Grade A sector. We therefore explore multi-sector models that distinguish causal effects of attendance at different sorts of Grade A schools, focusing on differences by admissions regime since this is widely believed to matter for school quality.

The multi-sector estimates reported in Table \ref{tab:2slsmulti} are from models that include separate endogenous variables for screened and lottery Grade A schools, along with a third endogenous variable for the ungraded sector.  Instruments in this just-identified set-up are two dummies indicating each sort of Grade A offer, as well as a dummy indicating the offer of a seat at an ungraded school.  2SLS models include separate saturated local propensity score controls for screened Grade A offer risk, unscreened Grade A offer risk, and ungraded offer risk. These multi-sector estimates are computed in a sample limited to applicants at risk of assignment to either a screened or lottery Grade A school. In view of the relative precision of estimates using IK bandwidth, multi-sector estimates using CCFT bandwidths are omitted.

OLS estimates again provide an interesting benchmark.  As can be seen in the first two columns of Table \ref{tab:2slsmulti}, screened Grade A students appear to reap a large SAT advantage even after controlling for baseline achievement and other covariates. In particular, OLS estimates of Grade A effects for schools in the screened sector are on the order of 14-18 points.  At the same time, Grade A lottery schools appear to generate achievement gains of only about 2 points.  Yet the corresponding 2SLS estimates, reported in columns 3 and 4 of the table, suggest the achievement gains yielded by enrollment in both sorts of Grade A schools are equally modest. The 2SLS estimates here run less than 2 points for math scores, with smaller (not significant) negative estimates for reading. The sole statistically significant SAT effect is that for the lottery Grade A school impact on math scores.

The remaining 2SLS estimates in the table likewise show similar screened-school and lottery-school effects. With one marginal exception, p-values in the table reveal estimates for the two sectors to be statistically indistinguishable. As in Table \ref{tab:2slsmain}, the 2SLS estimates in Table \ref{tab:2slsmulti} suggest that screened and lottery Grade A schools boost graduation rates by about 3 points. Effects on college and career preparedness are larger for lottery schools than for screened, but this impact ordering is reversed for effects on college readiness. On the whole, Table \ref{tab:2slsmulti} leads us to conclude that OLS estimates showing a large screened Grade A advantage are driven by selection bias.  

\begin{center}
	\includegraphics[width = 0.85\textwidth]{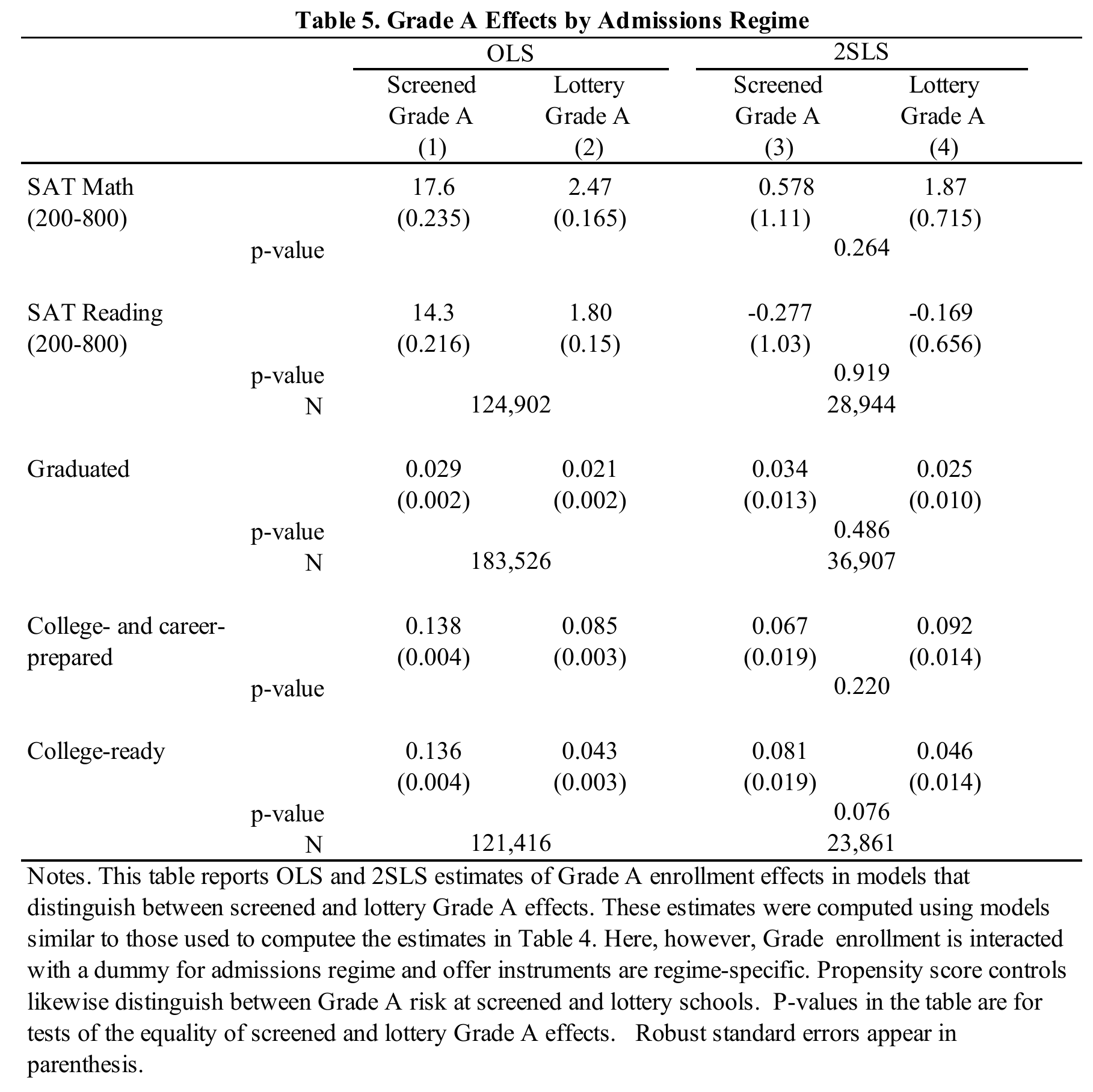} \refstepcounter{tablenums} \label{tab:2slsmulti}
\end{center}

\section{Summary and Next Steps} \label{sec:conclusions}

Centralized student assignment opens new opportunities for the measurement of school quality. The research potential of matching markets is enhanced here by marrying the conditional random assignment generated by lottery tie-breaking with RD-style variation at screened schools. The key to this intermingled empirical framework is a local propensity score that controls for differential assignment rates in DA matches with general tie-breakers.  This new tool allows us to exploit all sources of quasi-experimental variation arising from any mechanism in the DA class. 

Our analysis of NYC school report cards suggests Grade A schools boost SAT math scores and high school graduation rates by a few points. OLS estimates, by contrast, show considerably larger effects of Grade A attendance on test scores.  Grade A screened schools enroll some of the city's highest achievers, but large OLS estimates of achievement gains from attendance at these schools appear to be an artifact of selection bias. 
Concerns about access to such schools (expressed, for example, in \cite{harris/fessenden:17}) may therefore be overblown.
On the other hand, Grade A attendance increases measures of college and career preparedness. 
These results may reflect the greater availability of advanced courses in Grade A schools, a feature that should be replicable at other schools.

In principle, Grade A assignments may act to move applicants between schools within the Grade A sector as well as to boost overall Grade A enrollment. Offer-induced movement between screened and lottery Grade A schools may violate the exclusion restriction that underpins our 2SLS results if schools within the Grade A sector vary in quality.  It's therefore worth asking whether screened and lottery schools should indeed be treated as having the same effect.   Perhaps surprisingly, our analysis supports the idea that screened and lottery Grade A schools can be pooled and treated as having a common average causal effect.  

Our provisional agenda for further research prioritizes an investigation of econometric implementation strategies for DA-founded research designs.  This work is likely to build on the asymptotic framework in \cite{bugni2018testing}
and the study of RD designs with multiple tie-breakers in \cite{papay2011extending},  \cite{zajonc2012regression}, \cite{WSC:2013_RD} and \cite{Cattaneo_et_al:2019:RD_multi}. 
It may be possible to extend the reasoning behind doubly robust nonparametric estimators, such as discussed by \cite{rothe2019properties} and \cite{rothe2020flexible}, to our setting.

Statistical inference in Section \ref{sec:NYC} relies on conventional large sample reasoning of the sort widely applied in empirical RD applications.
It seems natural to consider permutation or randomization inference along the lines suggested by \cite{cattaneo2015randomization, cattaneo2017comparing}, and \cite{canay2017approximate}, along with optimal inference and estimation strategies such as those introduced by \cite{armstrong2018optimal} and \cite{imbens2017optimized}. Also on the agenda, \cite{narita2017non} suggests a path toward generalization of the large-market model of DA assignment risk. Finally, we look forward to a more detailed investigation of the consequences of heterogeneous treatment effects for identification strategies of the sort considered here.


\newpage


\onehalfspacing
\bibliographystyle{bib/ecta}
\bibliography{bib/mdrd2}

\begin{thebibliography}{64}
\newcommand{\enquote}[1]{``#1''}
\expandafter\ifx\csname natexlab\endcsname\relax\def\natexlab#1{#1}\fi

\bibitem[\protect\citeauthoryear{Abadie}{Abadie}{2003}]{abadie:03}
\textsc{Abadie, A.} (2003): \enquote{{Semiparametric instrumental variables
  estimation of treatment response models},} \emph{Journal of Econometrics},
  113(2), 231--263.

\bibitem[\protect\citeauthoryear{Abdulkadiro\u{g}lu, Angrist, Narita, and
  Pathak}{Abdulkadiro\u{g}lu et~al.}{2017a}]{mdrd1:17}
\textsc{Abdulkadiro\u{g}lu, A., J.~D. Angrist, Y.~Narita, and P.~A. Pathak}
  (2017a): \enquote{{Research Design Meets Market Design: Using Centralized
  Assignment for Impact Evaluation},} \emph{Econometrica}, 85(5), 1373--1432.

\bibitem[\protect\citeauthoryear{Abdulkadiro\u{g}lu, Angrist, Narita, and
  Pathak}{Abdulkadiro\u{g}lu et~al.}{2017b}]{mdrd2:17}
---\hspace{-.1pt}---\hspace{-.1pt}--- (2017b): \enquote{Impact Evaluation in
  Matching Markets with General Tie-breaking,} {NBER Working Paper No. 24172}.

\bibitem[\protect\citeauthoryear{Abdulkadiro\u{g}lu, Angrist, Narita, and
  Pathak}{Abdulkadiro\u{g}lu
  et~al.}{2019}]{abdulkadiroglu/angrist/narita/pathak:19}
---\hspace{-.1pt}---\hspace{-.1pt}--- (2019): \enquote{Breaking Ties:
  Regression Discontinuity Design Meets Market Design,} Cowles Foundation
  Discussion Paper 2170.

\bibitem[\protect\citeauthoryear{Abdulkadiro\u{g}lu, Angrist, and
  Pathak}{Abdulkadiro\u{g}lu et~al.}{2014}]{abdulkadiroglu/angrist/pathak:14}
\textsc{Abdulkadiro\u{g}lu, A., J.~D. Angrist, and P.~A. Pathak} (2014):
  \enquote{{The Elite Illusion: Achievement Effects at Boston and New York Exam
  Schools},} \emph{{E}conometrica}, 82(1), 137--196.

\bibitem[\protect\citeauthoryear{Abdulkadiro\u{g}lu, Pathak, and
  Roth}{Abdulkadiro\u{g}lu et~al.}{2005}]{apr:05}
\textsc{Abdulkadiro\u{g}lu, A., P.~A. Pathak, and A.~E. Roth} (2005):
  \enquote{{The New York City High School Match},} \emph{{A}merican {E}conomic
  {R}eview, Papers and Proceedings}, 95, 364--367.

\bibitem[\protect\citeauthoryear{Abdulkadiro\u{g}lu, Pathak, and
  Roth}{Abdulkadiro\u{g}lu et~al.}{2009}]{abdulkadiroglu/pathak/roth:09}
---\hspace{-.1pt}---\hspace{-.1pt}--- (2009): \enquote{{Strategy-Proofness
  versus Efficiency in Matching with Indifferences: Redesigning the New York
  City High School Match},} \emph{{A}merican {E}conomic {R}eview}, 99(5),
  1954--1978.

\bibitem[\protect\citeauthoryear{Abdulkadiro\u{g}lu and
  S\"{o}nmez}{Abdulkadiro\u{g}lu and
  S\"{o}nmez}{2003}]{abdulkadiroglu/sonmez:03}
\textsc{Abdulkadiro\u{g}lu, A. and T.~S\"{o}nmez} (2003): \enquote{{School
  Choice: A Mechanism Design Approach},} \emph{{A}merican {E}conomic {R}eview},
  93, 729--747.

\bibitem[\protect\citeauthoryear{Abebe, Fafchamps, Koelle, and Quinn}{Abebe
  et~al.}{2019}]{abebe2019learning}
\textsc{Abebe, G., M.~Fafchamps, M.~Koelle, and S.~Quinn} (2019):
  \enquote{{Learning Management Through Matching: A Field Experiment Using
  Mechanism Design},} {NBER Working Paper No. 26035}.

\bibitem[\protect\citeauthoryear{Ajayi}{Ajayi}{2014}]{ajayi:13}
\textsc{Ajayi, K.} (2014): \enquote{{Does School Quality Improve Student
  Performance? New Evidence from Ghana},} {IED Discussion Paper No. 260}.

\bibitem[\protect\citeauthoryear{Arai, Hsu, Kitagawa, Mourifi{\'e}, and
  Wan}{Arai et~al.}{2019}]{arai2019testing}
\textsc{Arai, Y., Y.-C. Hsu, T.~Kitagawa, I.~Mourifi{\'e}, and Y.~Wan} (2019):
  \enquote{{Testing Identifying Assumptions in Fuzzy Regression Discontinuity
  Designs},} {Cemmap Working Paper CWP10/19}.

\bibitem[\protect\citeauthoryear{Armstrong and Koles{\'a}r}{Armstrong and
  Koles{\'a}r}{2018}]{armstrong2018optimal}
\textsc{Armstrong, T.~B. and M.~Koles{\'a}r} (2018): \enquote{{Optimal
  Inference in a Class of Regression Models},} \emph{Econometrica}, 86,
  655--683.

\bibitem[\protect\citeauthoryear{Azevedo and Leshno}{Azevedo and
  Leshno}{2016}]{azevedo/leshno:14}
\textsc{Azevedo, E. and J.~Leshno} (2016): \enquote{{A Supply and Demand
  Framework for Two-Sided Matching Markets},} \emph{Journal of Political
  Economy}, 124(5), 1235--1268.

\bibitem[\protect\citeauthoryear{Barrow, Sartain, and de~la Torre}{Barrow
  et~al.}{2016}]{barrow2016role}
\textsc{Barrow, L., L.~Sartain, and M.~de~la Torre} (2016): \enquote{{The Role
  of Selective High Schools in Equalizing Educational Outcomes: Heterogeneous
  Effects by Neighborhood Socioeconomic Status},} {FRB of Chicago Working Paper
  No. 2016-17}.

\bibitem[\protect\citeauthoryear{Bergman}{Bergman}{2018}]{bergman2018risks}
\textsc{Bergman, P.} (2018): \enquote{{The Risks and Benefits of School
  Integration for Participating Students: Evidence from a Randomized
  Desegregation Program},} {IZA Discussion Paper}.

\bibitem[\protect\citeauthoryear{Beuermann, Jackson, and Sierra}{Beuermann
  et~al.}{2016}]{beuermann2015privately}
\textsc{Beuermann, D., C.~K. Jackson, and R.~Sierra} (2016):
  \enquote{{Privately Managed Public Secondary Schools and Academic Achievement
  in Trinidad and Tobago: Evidence from Rule-Based Student Assignments},} {IDB
  Working Paper Series No. 637}.

\bibitem[\protect\citeauthoryear{Brody}{Brody}{2019}]{brody:19}
\textsc{Brody, L.} (2019): \enquote{{Inside the Effort to Diversity Middle
  School in New York},} \emph{Wall Street Journal}, {May 18}.

\bibitem[\protect\citeauthoryear{Bugni and Canay}{Bugni and
  Canay}{2018}]{bugni2018testing}
\textsc{Bugni, F.~A. and I.~A. Canay} (2018): \enquote{Testing Continuity of a
  Density via g-order statistics in the Regression Discontinuity Design,}
  {Cemmap Working Paper CWP20/18}.

\bibitem[\protect\citeauthoryear{Calonico, Cattaneo, Farrell, and
  Titiunik}{Calonico et~al.}{2017}]{rdrobust2017}
\textsc{Calonico, S., M.~D. Cattaneo, M.~H. Farrell, and R.~Titiunik} (2017):
  \enquote{Rdrobust: Software for Regression-discontinuity Designs,} \emph{{The
  Stata Journal}}, 17, 372--404.

\bibitem[\protect\citeauthoryear{Calonico, Cattaneo, Farrell, and
  Titiunik}{Calonico et~al.}{2019}]{calonico_et_al:2019_RD}
---\hspace{-.1pt}---\hspace{-.1pt}--- (2019): \enquote{{Regression
  Discontinuity Designs Using Covariates},} \emph{The Review of Economics and
  Statistics}, 101, 442--451.

\bibitem[\protect\citeauthoryear{Canay and Kamat}{Canay and
  Kamat}{2017}]{canay2017approximate}
\textsc{Canay, I.~A. and V.~Kamat} (2017): \enquote{{Approximate Permutation
  Tests and Induced Order Statistics in the Regression Discontinuity Design},}
  \emph{Review of Economic Studies}, 85, 1577--1608.

\bibitem[\protect\citeauthoryear{Cattaneo, Frandsen, and Titiunik}{Cattaneo
  et~al.}{2015}]{cattaneo2015randomization}
\textsc{Cattaneo, M.~D., B.~R. Frandsen, and R.~Titiunik} (2015):
  \enquote{{Randomization Inference in the Regression Discontinuity Design: An
  Application to Party Advantages in the US Senate},} \emph{Journal of Causal
  Inference}, 3(1), 1--24.

\bibitem[\protect\citeauthoryear{Cattaneo, Titiunik, and Vazquez-Bare}{Cattaneo
  et~al.}{2017}]{cattaneo2017comparing}
\textsc{Cattaneo, M.~D., R.~Titiunik, and G.~Vazquez-Bare} (2017):
  \enquote{{Comparing Inference Approaches for RD Designs: A Reexamination of
  the Effect of Head Start on Child Mortality},} \emph{Journal of Policy
  Analysis and Management}, 36(3), 643--681.

\bibitem[\protect\citeauthoryear{Cattaneo, Titiunik, and Vazquez-Bare}{Cattaneo
  et~al.}{2019}]{Cattaneo_et_al:2019:RD_multi}
---\hspace{-.1pt}---\hspace{-.1pt}--- (2019): \enquote{{Analysis of Regression
  Discontinuity Designs with Multiple Cutoffs or Multiple Scores},} .

\bibitem[\protect\citeauthoryear{Cattaneo, Titiunik, Vazquez-Bare, and
  Keele}{Cattaneo et~al.}{2016{\natexlab{a}}}]{cattaneo2016interpreting}
\textsc{Cattaneo, M.~D., R.~Titiunik, G.~Vazquez-Bare, and L.~Keele}
  (2016{\natexlab{a}}): \enquote{{{I}nterpreting {R}egression {D}iscontinuity
  {D}esigns with {M}ultiple {C}utoffs},} \emph{Journal of Politics}, 78(4),
  1229--1248.

\bibitem[\protect\citeauthoryear{Cattaneo, Vazquez-Bare, and Titiunik}{Cattaneo
  et~al.}{2016{\natexlab{b}}}]{locrand2016}
\textsc{Cattaneo, M.~D., G.~Vazquez-Bare, and R.~Titiunik}
  (2016{\natexlab{b}}): \enquote{Inference in regression discontinuity designs
  under local randomization,} \emph{Stata Journal}, 16, 331--367(37).

\bibitem[\protect\citeauthoryear{Chen and Kesten}{Chen and
  Kesten}{2017}]{chen/kesten:17}
\textsc{Chen, Y. and O.~Kesten} (2017): \enquote{{Chinese College Admissions
  and School Choice Reforms: A Theoretical Analysis},} \emph{Journal of
  Political Economy}, 125, 99--139.

\bibitem[\protect\citeauthoryear{Disare}{Disare}{2017}]{disare:17}
\textsc{Disare, M.} (2017): \enquote{{City to Eliminate High School Admissions
  Method that Favored Families with Time and Resources},} \emph{Chalkbeat},
  {June 6}.

\bibitem[\protect\citeauthoryear{Dobbie and Fryer}{Dobbie and
  Fryer}{2014}]{dobbie/fryer:11}
\textsc{Dobbie, W. and R.~G. Fryer} (2014): \enquote{{Exam High Schools and
  Academic Achievement: Evidence from New York City},} \emph{American Economic
  Journal: Applied Economics}, 6(3), 58--75.

\bibitem[\protect\citeauthoryear{Dong}{Dong}{2018}]{dong2018alternative}
\textsc{Dong, Y.} (2018): \enquote{{Alternative Assumptions to Identify LATE in
  Fuzzy Regression Discontinuity Designs},} \emph{Oxford Bulletin of Economics
  and Statistics}, 80, 1020--1027.

\bibitem[\protect\citeauthoryear{Dur, Pathak, Song, and S\"{o}nmez}{Dur
  et~al.}{2018}]{dur/pathak/song/sonmez:18}
\textsc{Dur, U., P.~A. Pathak, F.~Song, and T.~S\"{o}nmez} (2018):
  \enquote{{Deduction Dilemmas: The Taiwan Assignment Mechanism},} {NBER
  Working Paper No. 25024}.

\bibitem[\protect\citeauthoryear{Ergin and S\"{o}nmez}{Ergin and
  S\"{o}nmez}{2006}]{ergin/sonmez:05}
\textsc{Ergin, H. and T.~S\"{o}nmez} (2006): \enquote{{Games of School Choice
  under the Boston Mechanism},} \emph{Journal of Public Economics}, 90,
  215--237.

\bibitem[\protect\citeauthoryear{Fort, Ichino, and Zanella}{Fort
  et~al.}{2020}]{fort2016cognitive}
\textsc{Fort, M., A.~Ichino, and G.~Zanella} (2020): \enquote{{Cognitive and
  Non-Cognitive Costs of Daycare 0-2 for Children in Advantaged Families},}
  \emph{{J}ournal of {P}olitical {E}conomy}, 128.

\bibitem[\protect\citeauthoryear{Frandsen}{Frandsen}{2017}]{frandsen2017party}
\textsc{Frandsen, B.~R.} (2017): \enquote{{Party Bias in Union Representation
  Elections: Testing for Manipulation in the Regression Discontinuity Design
  When the Running Variable is Discrete},} in \emph{Regression Discontinuity
  Designs: Theory and Applications}, Emerald Publishing Limited, 281--315.

\bibitem[\protect\citeauthoryear{Frolich}{Frolich}{2007}]{frolich2007regression}
\textsc{Frolich, M.} (2007): \enquote{{Regression Discontinuity Design with
  Covariates (Unpublished Appendix)},} {IZA Discussion Paper No. 3024}.

\bibitem[\protect\citeauthoryear{Frolich and Huber}{Frolich and
  Huber}{2019}]{frolich/huber2019}
\textsc{Frolich, M. and M.~Huber} (2019): \enquote{{Including Covariates in the
  Regression Discontinuity Design},} \emph{Journal of Business and Economic
  Statistics}, 37, 736--748.

\bibitem[\protect\citeauthoryear{Hahn, Todd, and Van~der Klaauw}{Hahn
  et~al.}{2001}]{hahn2001identification}
\textsc{Hahn, J., P.~Todd, and W.~Van~der Klaauw} (2001):
  \enquote{{Identification and Estimation of Treatment Effects with a
  Regression-Discontinuity Design},} \emph{Econometrica}, 69(1), 201--209.

\bibitem[\protect\citeauthoryear{Harris and Fessenden}{Harris and
  Fessenden}{2017}]{harris/fessenden:17}
\textsc{Harris, E. and F.~Fessenden} (2017): \enquote{{The Broken Promises of
  Choice in New York City Schools},} \emph{New York Times}, {May 5}.

\bibitem[\protect\citeauthoryear{Hastings, Neilson, and Zimmerman}{Hastings
  et~al.}{2013}]{hastings/neilson/zimmerman:13}
\textsc{Hastings, J., C.~Neilson, and S.~D. Zimmerman} (2013): \enquote{{Are
  Some Degrees Worth More than Others? Evidence from College Admission Cutoffs
  in Chile},} {NBER Working Paper No. 19241}.

\bibitem[\protect\citeauthoryear{Imbens and Kalyanaraman}{Imbens and
  Kalyanaraman}{2012}]{imbens/kalyanaraman:12}
\textsc{Imbens, G.~W. and K.~Kalyanaraman} (2012): \enquote{{Optimal Bandwidth
  Choice for the Regression Discontinuity Estimator},} \emph{Review of Economic
  Studies}, 79(3), 933--959.

\bibitem[\protect\citeauthoryear{Imbens and Wager}{Imbens and
  Wager}{2019}]{imbens2017optimized}
\textsc{Imbens, G.~W. and S.~Wager} (2019): \enquote{{Optimized Regression
  Discontinuity Designs},} \emph{Review of Economics and Statistics}, 101,
  264--278.

\bibitem[\protect\citeauthoryear{Jackson}{Jackson}{2010}]{jackson:10}
\textsc{Jackson, K.} (2010): \enquote{{Do Students Benefit from Attending
  Better Schools? Evidence from Rule-based Student Assignments in Trinidad and
  Tobago},} \emph{Economic Journal}, 120(549), 1399--1429.

\bibitem[\protect\citeauthoryear{Jackson}{Jackson}{2012}]{jackson:12}
---\hspace{-.1pt}---\hspace{-.1pt}--- (2012): \enquote{{Single-sex Schools,
  Student Achievement, and Course Selection: Evidence from Rule-based Student
  Assignments in Trinidad and Tobago},} \emph{Journal of Public Economics},
  96(1-2), 173--187.

\bibitem[\protect\citeauthoryear{Kirkeboen, Leuven, and Mogstad}{Kirkeboen
  et~al.}{2016}]{kirkeboen/leuven/mogstad:15}
\textsc{Kirkeboen, L., E.~Leuven, and M.~Mogstad} (2016): \enquote{{Field of
  Study, Earnings, and Self-Selection},} \emph{Quarterly Journal of Economics},
  131, 1057--1111.

\bibitem[\protect\citeauthoryear{Lee}{Lee}{2008}]{lee2008randomized}
\textsc{Lee, D.~S.} (2008): \enquote{{Randomized Experiments from Non-Random
  Selection in US House Elections},} \emph{Journal of Econometrics}, 142,
  675--697.

\bibitem[\protect\citeauthoryear{Lucas and Mbiti}{Lucas and
  Mbiti}{2014}]{lucas/mbiti:14}
\textsc{Lucas, A. and I.~Mbiti} (2014): \enquote{{Effects of School Quality on
  Student Achievement: Discontinuity Evidence from Kenya},} \emph{American
  Economic Journal: Applied Economics}, 6(3), 234--263.

\bibitem[\protect\citeauthoryear{Narita}{Narita}{2020}]{narita2017non}
\textsc{Narita, Y.} (2020): \enquote{{A Theory of Quasi-Experimental Evaluation
  of School Quality},} \emph{Management Science}.

\bibitem[\protect\citeauthoryear{Papay, Willett, and Murnane}{Papay
  et~al.}{2011}]{papay2011extending}
\textsc{Papay, J.~P., J.~B. Willett, and R.~J. Murnane} (2011):
  \enquote{{Extending the Regression-Discontinuity Approach to Multiple
  Assignment Variables},} \emph{Journal of Econometrics}, 161(2), 203--207.

\bibitem[\protect\citeauthoryear{Pathak and S\"{o}nmez}{Pathak and
  S\"{o}nmez}{2013}]{pathak/sonmez:13}
\textsc{Pathak, P.~A. and T.~S\"{o}nmez} (2013): \enquote{{School Admissions
  Reform in Chicago and England: Comparing Mechanisms by their Vulnerability to
  Manipulation},} \emph{American Economic Review}, 103(1), 80--106.

\bibitem[\protect\citeauthoryear{P{\'e}rez~Vincent and Ubfal}{P{\'e}rez~Vincent
  and Ubfal}{2019}]{argentina}
\textsc{P{\'e}rez~Vincent, S. and D.~Ubfal} (2019): \enquote{{Using Centralized
  Assignment to Evaluate Entrepreneurship and Life-Skills Training Programs in
  Argentina},} {Working Paper}.

\bibitem[\protect\citeauthoryear{Pop-Eleches and Urquiola}{Pop-Eleches and
  Urquiola}{2013}]{pop-eleches/urquiola:13}
\textsc{Pop-Eleches, C. and M.~Urquiola} (2013): \enquote{{Going to a Better
  School: Effects and Behavioral Responses},} \emph{{A}merican {E}conomic
  {R}eview}, 103(4), 1289--1324.

\bibitem[\protect\citeauthoryear{Rockoff and Turner}{Rockoff and
  Turner}{2011}]{rockoff/turner:10}
\textsc{Rockoff, J. and L.~Turner} (2011): \enquote{{Short Run Impacts of
  Accountability of School Quality},} \emph{American Economic Journal: Economic
  Policy}, 2(4), 119--147.

\bibitem[\protect\citeauthoryear{Rosenbaum and Rubin}{Rosenbaum and
  Rubin}{1983}]{rosenbaum/rubin:83}
\textsc{Rosenbaum, P.~R. and D.~B. Rubin} (1983): \enquote{{The Central Role of
  the Propensity Score in Observational Studies for Causal Effects},}
  \emph{Biometrica}, 70, 41--55.

\bibitem[\protect\citeauthoryear{Rothe}{Rothe}{2020}]{rothe2020flexible}
\textsc{Rothe, C.} (2020): \enquote{{Flexible Covariate Adjustments in
  Randomized Experiments},} .

\bibitem[\protect\citeauthoryear{Rothe and Firpo}{Rothe and
  Firpo}{2019}]{rothe2019properties}
\textsc{Rothe, C. and S.~Firpo} (2019): \enquote{Properties of doubly robust
  estimators when nuisance functions are estimated nonparametrically,}
  \emph{Econometric Theory}, 35, 1048--1087.

\bibitem[\protect\citeauthoryear{Sekhon and Titiunik}{Sekhon and
  Titiunik}{2017}]{sekhon2017interpreting}
\textsc{Sekhon, J.~S. and R.~Titiunik} (2017): \enquote{{On Interpreting the
  Regression Discontinuity Design as a Local Experiment},} in \emph{Regression
  Discontinuity Designs: Theory and Applications}, Emerald Publishing Limited,
  1--28.

\bibitem[\protect\citeauthoryear{van~der Vaart}{van~der Vaart}{2000}]{vdv:00}
\textsc{van~der Vaart, A.~W.} (2000): \emph{{Asymptotic Statistics}}, Cambridge
  University Press.

\bibitem[\protect\citeauthoryear{Veiga}{Veiga}{2018}]{veiga:18}
\textsc{Veiga, C.} (2018): \enquote{{Brooklyn Middle Schools Eliminate
  `Screening' as New York City Expands Integration Efforts},} \emph{Chalkbeat},
  {September 20}.

\bibitem[\protect\citeauthoryear{Walcott}{Walcott}{2012}]{walcott:12}
\textsc{Walcott, D.} (2012): \enquote{{NYC Department of Education: Progress
  Reports for New York City Public Schools},} .

\bibitem[\protect\citeauthoryear{Wellner}{Wellner}{1981}]{wellner1981glivenko}
\textsc{Wellner, J.~A.} (1981): \enquote{{A Glivenko-Cantelli Theorem for
  Empirical Measures of Independent but Non-Identically Distributed Random
  Variables},} \emph{Stochastic Processes and Their Applications}, 11(3),
  309--312.

\bibitem[\protect\citeauthoryear{Wong, Steiner, and Cook}{Wong
  et~al.}{2013{\natexlab{a}}}]{wong2013analyzing}
\textsc{Wong, V.~C., P.~M. Steiner, and T.~D. Cook} (2013{\natexlab{a}}):
  \enquote{{Analyzing Regression-Discontinuity Designs with Multiple Assignment
  Variables: A Comparative Study of Four Estimation Methods},} \emph{Journal of
  Educational and Behavioral Statistics}, 38, 107--141.

\bibitem[\protect\citeauthoryear{Wong, Steiner, and Cook}{Wong
  et~al.}{2013{\natexlab{b}}}]{WSC:2013_RD}
---\hspace{-.1pt}---\hspace{-.1pt}--- (2013{\natexlab{b}}): \enquote{{Analyzing
  Regression-Discontinuity Designs With Multiple Assignment Variables: A
  Comparative Study of Four Estimation Methods},} \emph{Journal of Educational
  and Behavioral Statistics}, 38, 107--141.

\bibitem[\protect\citeauthoryear{Zajonc}{Zajonc}{2012}]{zajonc2012regression}
\textsc{Zajonc, T.} (2012): \enquote{{Regression Discontinuity Design with
  Multiple Forcing Variables},} \emph{Essays on Causal Inference for Public
  Policy}, 45--81.

\bibitem[\protect\citeauthoryear{Zimmerman}{Zimmerman}{2019}]{zimmerman2016making}
\textsc{Zimmerman, S.~D.} (2019): \enquote{{Elite Colleges and Upward Mobility
  to Top Jobs and Top Incomes},} \emph{American Economic Review}, 109, 1--47.

\end{thebibliography}

\newpage

\appendix 
\begin{center}
\begin{LARGE}
\textbf{Appendix}
\end{LARGE}
\end{center}

\section{Proof of Theorem \ref{theorem:local}}

Let $F^i_v(r)$ denote the cumulative distribution function (CDF) of $R_{iv}$ evaluated at $r$ and define
\begin{equation}\label{conditionalF}
F_{v}(r|\theta) = E[F^i_{v}(r)|\theta_i=\theta].
\end{equation}
This is the fraction of type $\theta$ applicants with tie-breaker $v$ below $r$ (set to zero when type $\theta$ ranks no schools using tie-breaker $v$). 
We may condition on additional events. 

Recall that the joint distribution of tie-breakers for applicant $i$ is assumed to be continuously differentiable with positive density. This assumption has the following implication: 
The conditional distribution of tie-breaker $v$, $F_v(r| e),$ is continuously differentiable, with $F'_v(r| e)>0$ at any $r=\tau_1, ..., \tau_S$. 
		Here, the conditioning event $e$ is any event of the form that $\theta_i=\theta, R_{iu}>r_u\text{ for }u=1, ..., v-1,$ and $T_i(\delta)=T$.





Take any large market with the general tie-breaking structure in Section \ref{sec:multiscore}. 
For each $\delta>0$ and each tie-breaker $v=U+1, ..., V+1$, let $e(v)$ be short-hand notation for ``$\theta_i=\theta, R_{iu}>MID^u_{\theta s}\text{ for }u=1, ..., v-1, T_i(\delta)=T,$ and $W_i=w$." 
Similarly, $e(1)$ is short-hand notation for ``$\theta_i=\theta, T_i(\delta)=T,$ and $W_i=w$."  
Let
$\psi_s(\theta,T, \delta, w)\equiv E[D_i(s)| e(1)]$ be the assignment probability for an applicant with $\theta_i=\theta, T_i(\delta)=T,$ and characteristics $W_i=w$. 
Our proofs use a lemma that describes this assignment probability.
To state the lemma, for $v>U$, let
\begin{small}
$$\Phi_\delta(v)\equiv 
\begin{cases}
\dfrac{F_v(MID^v_{\theta s}+\delta| e(v))-F_v(MID^v_{\theta s}-\delta| e(v))}{F_v(MID^v_{\theta s}| e(v))-F_v(MID^v_{\theta s}-\delta| e(v))}\text{ if }t_b(\delta)=c\text{ for some }b\in B^v_{\theta s}\\
1\ \ \ \ \ \ \ \ \ \ \ \ \ \ \ \ \ \ \ \ \ \ \ \ \ \ \ \ \ \ \ \ \ \ \ \ \ \ \ \ \ \ \ \ \ \ \ \ \ \ \ \ \ \ \ \ \text{ otherwise.}
\end{cases}$$
We use this object to define
$\Phi_\delta\equiv\prod_{v=1}^{U}(1-MID^{v}_{\theta s})\prod^V_{v=U+1}\Phi_\delta(v).$ 
Finally, let
$$\Phi'_\delta\equiv
\begin{cases}
\max\left\{0, \dfrac{F_{v(s)}(\tau_{s}| e(V+1))-F_{v(s)}(\tau_s-\delta| e(V+1))}{F_{v(s)}(\tau_{s}+\delta| e(V+1))- F_{v(s)}(\tau_s-\delta| e(V+1))}\right\}\text{ if }v(s)>U\\
\max\left\{0, \dfrac{\tau_{s}-MID^{v(s)}_{\theta s}}{1-MID^{v(s)}_{\theta s}}\right\} \ \ \ \ \ \ \ \ \ \ \ \ \ \ \ \ \ \ \ \ \ \ \ \ \ \ \ \ \ \ \ \ \ \ \ \ \ \ \ \ \ \ \ \ \ \text{ if }v(s)\leq U.
\end{cases}$$
\end{small}

\begin{lemma}\label{pscoreContinuum}
In the general tie-breaking setting of Section \ref{sec:multiscore}, 
for any fixed $\delta>0$ such that $\delta < \min_{\theta,s,v} |\tau_s - MID_{\theta,s}^{v}|$, we have:
\[
\psi_s(\theta,T, \delta, w)
=\left\{\begin{array}
[c]{ll}
0 &\text{ if }t_s(\delta)=n\text{ or }t_b(\delta)=a\text{ for some }b\in B_{\theta s},
\\
\Phi_\delta &\text{ otherwise and }t_s(\delta)=a,
\\
\Phi_\delta\times\Phi'_\delta &\text{ otherwise and }t_s(\delta)=c.
\end{array}
\right.  
\]
\end{lemma}

\noindent \textit{Proof of Lemma \ref{pscoreContinuum}.} 
We start verifying the first line in $\psi_s(\theta,T, \delta, w)$.
Applicants who don't rank $s$ have $\psi_s(\theta,T, \delta, w)=0$.
Among those who rank $s$, those of $t_s(\delta)=n$ have $\rho_{\theta s}>\rho_s \text{ or, if } v(s)\neq 0, \: \rho_{\theta s}=\rho_s \text{ and } R_{iv(s)} > \tau_{s} + \delta$. 
If $\rho_{\theta s}>\rho_s$, then $\psi_s(\theta,T, \delta, w)=0$.
Even if $\rho_{\theta s}\leq\rho_s$, as long as $\: \rho_{\theta s}=\rho_s \text{ and } R_{iv(s)} > \tau_{s} + \delta$, student $i$ never clears the cutoff at school $s$ so $\psi_s(\theta,T, \delta, w)=0$.

To show the remaining cases, take as given that it is not the case that $t_s(\delta)=n\text{ or }t_b(\delta)=a\text{ for some }b\in B_{\theta s}$. 
Applicants with $t_b(\delta)\neq a$ for all $b\in B_{\theta s}$ and $t_s(\delta)=a$ or $c$ may be assigned $b\in B_{\theta s},$ where $\rho_{\theta b} = \rho_{b}$.
Since the (aggregate) distribution of tie-breaking variables for type $\theta$ students is $\hat F_v(\cdot| \theta)=F_v(\cdot| \theta)$, conditional on $T_i(\delta)=T$, the proportion of type $\theta$ applicants not  assigned any $b\in B_{\theta s}$ where $\rho_{\theta b} = \rho_{b}$ is $\Phi_\delta=\prod_{v=1}^{U}(1-MID^{v}_{\theta s})\prod^V_{v=U+1}\Phi_\delta(v)$ since each $\Phi_\delta(v)$ is the probability of not being assigned to any $b\in B^v_{\theta s}$.
To see why $\Phi_\delta(v)$ is the probability of not being assigned to any $b\in B^v_{\theta s}$, note that if $t_b(\delta)\neq c\text{ for all }b\in B^v_{\theta s}$, then $t_b(\delta)=n$ for all $b\in B^v_{\theta s}$ so that applicants are never assigned to any $b\in B^v_{\theta s}$. 
Otherwise, i.e., if $t_b(\delta)=c\text{ for some }b\in B^v_{\theta s}$, then applicants are assigned to $s$ if and only if their values of tie-breaker $v$ clear the cutoff of the school that produces $MID^{v}_{\theta s}$, where applicants have $t_s(\delta)=c$. 
This event happens with probability 

$$\dfrac{F_v(MID^v_{\theta s}| e(v))-F_v(MID^v_{\theta s}-\delta| e(v))}{F_v(MID^v_{\theta s}+\delta| e(v))-F_v(MID^v_{\theta s}-\delta| e(v))},$$

\noindent implying that $\Phi_\delta(v)$ is the probability of not being assigned to any $b\in B^v_{\theta s}$. 

Given this fact, to see the second line, note that every applicant of type $t_s(\delta)=a$ who is not assigned a higher choice
is assigned $s$ for sure because $\rho_{\theta s}<\rho_{s}$ or $\rho_{\theta s}+R_{iv(s)}<\xi_s$. 
Therefore, we have%
\[
\psi_s(\theta,T, \delta, w)=\Phi_\delta.
\]

Finally, consider applicants with $t_s(\delta)=c$.
The fraction of those who are
not assigned a higher choice is $\Phi_\delta$, as explained above.
Also, for tie-breaker $v(s)$, the tie-breaker values of these applicants are larger (worse) than $MID^{v(s)}_{\theta s}$.
If $\tau_{s}<MID^{v(s)}_{\theta s},$ then no such applicant is assigned $s.$ If
$\tau_{s}\geq MID^{v(s)}_{\theta s},$ then the fraction of applicants who are
assigned $s$ conditional on $\tau_{s}\geq MID^{v(s)}_{\theta s}$ is given by
$$\max\left\{0, \frac{F_{v(s)}(\tau_{s}| e(V+1))-\max\{F_{v(s)}(MID^{v(s)}_{\theta s}| e(V+1)), F_{v(s)}(\tau_s-\delta| e(V+1))\}}{F_{v(s)}(\tau_{s}+\delta| e(V+1))-\max\{F_{v(s)}(MID^{v(s)}_{\theta s}| e(V+1)), F_{v(s)}(\tau_s-\delta| e(V+1))\}}\right\}\text{ if }v(s)>U$$
and
$$\max\left\{0, \dfrac{\tau_{s}-MID^{v(s)}_{\theta s}}{1-MID^{v(s)}_{\theta s}}\right\}\text{ if }v(s)\leq U.$$

\noindent If $MID^{v(s)}_{\theta s} < \tau_s$, then $\delta < \min_{\theta,s,v} |\tau_s - MID_{\theta,s}^{v}|$ implies $MID^{v(s)}_{\theta s} < \tau_s-\delta$.  This in turn implies
$$\max\{F_{v(s)}(MID^{v(s)}_{\theta s}| e(V+1)), F_{v(s)}(\tau_s-\delta| e(V+1))\}= F_{v(s)}(\tau_s-\delta| e(V+1)).$$
If $MID^{v(s)}_{\theta s} > \tau_s$, then $\delta < \min_{\theta,s,v} |\tau_s - MID_{\theta,s}^{v}|$ implies $MID^{v(s)}_{\theta s} > \tau_s+\delta$. 
By the definition of $e(V+1)$, $ R_{iu}>MID^u_{\theta s}\text{ for }u=1, ..., V$. Therefore, there is no applicant with $R_{iv(s)}>MID^{v(s)}_{\theta s}$ and $R_{iv(s)} \in [\tau_s-\delta,\tau_s+\delta]$.

Hence, conditional on $t_s(\delta)=c$ and not being assigned a choice preferred to $s,$ the probability of
being assigned $s$ is given by $\Phi'_\delta$.
Therefore, for students with $t_s(\delta)=c$, we have $\psi_s(\theta,T, \delta, w)=\Phi_\delta \times\Phi'_\delta.$\qed

\begin{lemma}\label{lemma2}
In the general tie-breaking setting of Section \ref{sec:multiscore}, 
for all $s$, $\theta$, and sufficiently small $\delta>0$, we have:
\begin{equation}
\psi_s(\theta,T, \delta, w)=\left\{
\begin{array}
[c]{ll}


0 & \text{ if }t_s(0)=n\text{ or }t_b(0)=a\text{ for some }b\in B_{\theta s},\\







\Phi^*_\delta &\text{ otherwise and }t_s(0)=a,\\


\Phi^*_\delta\times&\dfrac{F_{v(s)}(\tau_{s}| e(V+1))-F_{v(s)}(\tau_s-\delta| e(V+1))}{F_{v(s)}(\tau_s+\delta| e(V+1))-F_{v(s)}(\tau_s-\delta| e(V+1))}\\

\ &\text{ otherwise and }t_s(0)=c\text{ and }v(s)>U. \\

\Phi^*_\delta\times&\max\{ 0,\dfrac{\tau_s-MID^{v(s)}_{\theta s}}{1-MID^{v(s)}_{\theta s}}\}\\

\ &\text{ otherwise and }t_s(0)=c\text{ and }v(s)\leq U. \\

\end{array}
\right.
\end{equation}

\noindent where
$$\Phi^*_\delta(v)\equiv
\begin{cases}
\dfrac{F_v(MID^v_{\theta s}+\delta| e(v))-F_v(MID^v_{\theta s}| e(v))}{F_v(MID^v_{\theta s}+\delta| e(v))-F_v(MID^v_{\theta s}-\delta| e(v))}\text{ if }MID^{v}_{\theta s}=\tau_b \text{ and } t_b=c \text{ for some } b\in B^{v}_{\theta s},\\
1\ \ \ \ \ \ \ \ \ \ \ \ \ \ \ \ \ \ \ \ \ \ \ \ \ \ \ \ \ \ \ \ \ \ \ \ \ \ \ \ \ \ \ \ \ \ \ \ \ \ \ \ \text{ otherwise}
\end{cases}$$
and

$$\Phi^*_\delta\equiv \prod_{v=1}^{U}(1-MID^{v}_{\theta s})\prod^V_{v=U+1}\Phi^*_\delta(v).$$
\end{lemma}

\noindent \textit{Proof of Lemma \ref{lemma2}}. 
The first line follows from Lemma \ref{pscoreContinuum} and the fact that $t_s(0)=n\text{ or }t_b(0)=a\text{ for some }b\in B_{\theta s}$ imply $t_s(\delta)=n\text{ or }t_b(\delta)=a\text{ for some }b\in B_{\theta s}$ for sufficiently small $\delta>0$.

For the remaining lines, first note that conditional on $t_s(0)\neq n\text{ and }t_b(0)\neq a\text{ for all }b\in B_{\theta s}$, we have $\Phi^*_\delta(v)=\Phi_\delta(v)$ and so $\Phi^*_\delta=\Phi_\delta$ holds for small enough $\delta$. 
$\Phi^*_\delta$ therefore is the probability of not being assigned to a school preferred to $s$ in the last three cases.

The second line is then by the fact that $t_s(0)=a$ implies $t_s(\delta)=a$ for small enough $\delta>0$. 
The third line is by the fact that for small enough $\delta>0$, 
\begin{align*}
\Phi'_\delta&=\max\bigg\{0, \frac{F_{v(s)}(\tau_{s}| e(V+1))-F_{v(s)}(\tau_s-\delta| e(V+1))}{F_{v(s)}(\tau_s+\delta| e(V+1))-F_{v(s)}(\tau_{s}-\delta| e(V+1))}\bigg\}\\
&=\frac{F_{v(s)}(\tau_{s}| e(V+1))-F_{v(s)}(\tau_s-\delta| e(V+1))}{F_{v(s)}(\tau_s+\delta| e(V+1))-F_{v(s)}(\tau_{s}-\delta| e(V+1))},
\end{align*}

\noindent where we invoke Assumption \ref{assumption_differentiability}, which implies $MID^v_{\theta s}\neq \tau_s$. 
The last line directly follows from Lemma \ref{pscoreContinuum}. \qed \\

We use Lemma \ref{lemma2} to derive Theorem \ref{theorem:local}. 
We characterize $\lim_{\delta\rightarrow 0}\psi_s(\theta,T, \delta, w)$ and show that it coincides with $\psi_{s}(\theta, T)$ in the main text. 
In the first case in Lemma \ref{lemma2}, $\psi_s(\theta,T, \delta, w)$ is constant (0) for any small enough $\delta$.
The constant value is also $\lim_{\delta\rightarrow 0}\psi_s(\theta,T, \delta, w)$ in this case.

To characterize $\lim_{\delta\rightarrow 0}\psi_s(\theta,T, \delta, w)$ in the remaining cases, note that by the differentiability of $F_v(\cdot| e(v))$ (recall the continuous differentiability of $F^i_v(r| e)$), L'Hopital's rule implies:
\[
\lim_{\delta\rightarrow 0}\dfrac{F_{v(s)}(\tau_{s}| e(V+1))-F_{v(s)}(\tau_s-\delta| e(V+1))}{F_{v(s)}(\tau_s+\delta| e(V+1))-F_{v(s)}(\tau_s-\delta| e(V+1))}=\dfrac{F'_{v(s)}(\tau_s| e(V+1))}{2F'_{v(s)}(\tau_s| e(V+1))}=0.5\\
\]
and
\[
\lim_{\delta\rightarrow 0}\dfrac{F_v(MID^v_{\theta s}+\delta| e(v))-F_v(MID^v_{\theta s}| e(v))}{F_v(MID^v_{\theta s}+\delta| e(v))-F_v(MID^v_{\theta s}-\delta| e(v))}=\dfrac{F'_v(MID^v_{\theta s}| e(v))}{2F'_v(MID^v_{\theta s}| e(v))}=0.5.\\
\]

\noindent This implies $\lim_{\delta\rightarrow 0}\Phi^*_\delta(v)=0.5^{1\{MID^{v}_{\theta s}=\tau_b \text{ and } t_b=c \text{ for some } b\in B^{v}_{\theta s}\}}$ since $1\{MID^{v}_{\theta s}=\tau_b \text{ and } t_b=c \text{ for some } b\in B^{v}_{\theta s}\}$ does not depend on $\delta$. 
Therefore $$\lim_{\delta\rightarrow 0}\Phi^*_\delta=\prod_{v=1}^{U}(1-MID^{v}_{\theta s})0.5^{m_{s}(\theta,T)}$$
where $m_{s}(\theta,T) = |\{v>U: MID^{v}_{\theta s}=\tau_b \text{ and } t_b=c \text{ for some } b\in B^{v}_{\theta s}\}|$. 

Combining these limiting facts with the fact  that the limit of a product of functions equals the product of the limits of the functions, we obtain the following: 
$\lim_{\delta\rightarrow 0}\psi_s(\theta,T, \delta, w)=0$ if (a) $t_{s}=n$ or (b) $t_b=a \ \text{ for some } b \in B_{\theta s}$. Otherwise,\\
		\begin{equation*}
	\psi_s(\theta,T)=\left\{
	\begin{array}
	[c]{ll}
	\sigma_s(\theta,T)\lambda_s(\theta) & \text{ if } t_{s}=a\\
	\sigma_s(\theta,T)\lambda_s(\theta)\max\left\{ 0,\frac{\tau_{s}-MID^{v(s)}_{\theta s}}{1-MID^{v(s)}_{\theta s}}\right\} & \text{ if } t_{s}=c\text{ and }v(s)\leq U\\
	0.5\sigma_s(\theta,T)\lambda_s(\theta) & \text{ if } t_{s}=c\text{ and }v(s)> U.\\
	\end{array}
	\right.
	\end{equation*}
		This expression coincides with $\psi_s(\theta,T)$, completing the proof of Theorem \ref{theorem:local}. 

\newpage 

\begin{center}
\begin{LARGE}
\textbf{Online Appendices}
\end{LARGE}
\end{center}

\makeatletter 
\renewcommand{\thefigure}{\thesection\@arabic\c@figure}
\makeatother

\section{Understanding Theorem \ref{theorem:local}}\label{appendix_illustration}

Figure \ref{figappdx:Thm1} illustrates Theorem \ref{theorem:local} for an applicant who ranks screened schools 1, 3, 5 and 6 and lottery schools 2 and 4, where school $k$ is applicant's $k$-th choice. The line next to each school represents applicant position (priority plus tie-breaker) for each school. Schools with the same colored lines have the same tie-breaker. Schools 1 and 5 use screened tie-breaker $2$. Schools 2 and 4 use lottery tie-breaker 1. Schools 3 and 6 use screened tie-breaker $3$. 

Since school 1 has only one priority, positions run from 1 to 2.  School \textit{2} has two priority groups, so positions run from 1 to 3. Figure \ref{figappdx:Thm1} indicates the applicants position $\pi$ by an arrow. At screened schools, the brackets around the DA cutoff $\xi$ represent the $\delta$-neighborhood around the cutoff. 

\setcounter{figure}{0}

\begin{landscape}
\begin{figure}[!t] 
    \centering
    \caption*{Figure B1: Illustrating Theorem 1} \refstepcounter{figure} \label{figappdx:Thm1} 
    \includegraphics[scale=.47]{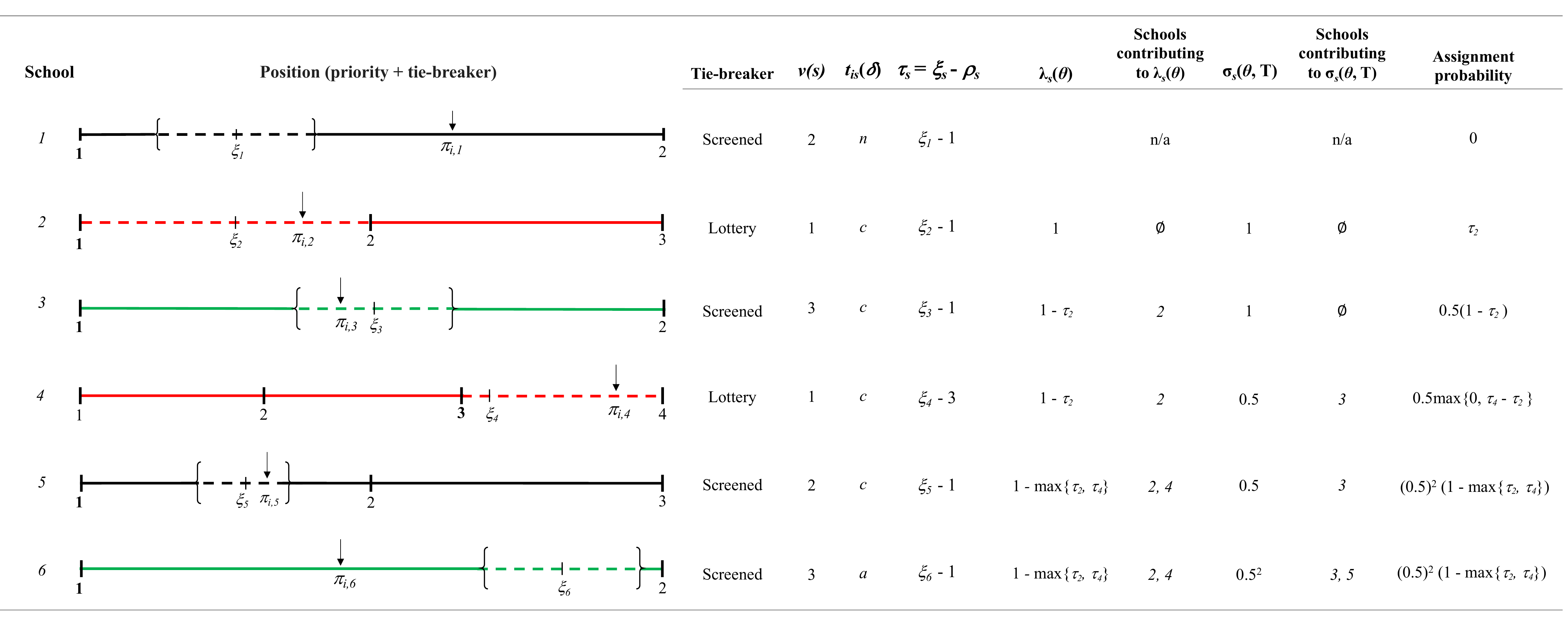}
    \vspace{0.5cm}
    \floatfoot{\footnotesize{Notes: This figure illustrates Theorem 1 for one applicant listing six schools. The applicant has marginal priority (shown in bold) at each. Dashes mark intervals in which offer risk is strictly between 0 and 1. The set of applicants subject to random assignment includes everyone with marginal priority at lottery schools and applicants with tie-breakers inside the relevant bandwidth at screened schools. Same-color tie-breakers are shared. Schools \textit{1, 3, 5,} and \textit{6} are screened, while \textit{2} and \textit{4} have lottery tie-breakers. The applicant's preferences are \textit{1} $ \succ_i $ \textit{2} $ \succ_i $  \textit{3} $ \succ_i $ \textit{4} $ \succ_i $ \textit{5} $ \succ_i $ \textit{6}. Arrows mark $\pi_{is} = \rho_{is} + R_{iv(s)}$, the applicant's position at each school $s$. Lower $\pi_{is}$ is better. Integers indicate priorities $\rho_s$, and tick marks indicate the DA cutoff, $\xi_s = \rho_s + \tau_s$. Note that $t_\mathit{6} = a$, so this applicant is sure to be seated somewhere. The assignment probability therefore sums to 1: if $\tau_{\mathit{2}} \geq \tau_{\mathit{4}},$ the probability of any assignment is $\tau_{\mathit{2}} + 0.5 \times (1- \tau_{\mathit{2}}) + 0 + 2 \times 0.5^2 \times (1 - \tau_{\mathit{2}}) = 1$; if $\tau_{\mathit{2}}  <  \tau_{\mathit{4}},$ this probability is $\tau_{\mathit{2}} + 0.5 \times (1-\tau_{\mathit{2}})+0.5\times (\tau_\mathit{4}-\tau_{\mathit{2}})+ 2 \times 0.5^2 \times (1- \tau_{\mathit{4}})=1$.}}
    
\end{figure}
\end{landscape}

The applicant is never seated at school \textit{1} since his position is to the right of the $\delta$-neighborhood, conditionally seated at schools \textit{2} and \textit{4} since his priority is equal to the marginal priority at each school, conditionally seated at schools \textit{3} and \textit{5} since his position is within the $\delta$-neighborhood at each school, and always seated at school \textit{6} since his position is to the left of the $\delta$-neighborhood. 

The columns next to the lines record tie-breaker cutoff, $\tau$, disqualification probability at lottery schools, $\lambda$, schools contributing to $\lambda$, the disqualification probability at screened schools, $\sigma$, schools contributing to $\sigma$, and assignment probability.

The local score at each school is computed as follows: 
\begin{itemize}
    \item[School \textit{1}:] The local score at school \textit{1} is zero because $t_{i\mathit{1}}(\delta)=n$.
    \item[School \textit{2}:] MID at school \textit{2} is zero because this applicant ranks no other lottery school higher. Hence, the second line of \eqref{equation_main_theorem2} applies and probability is given by the tie-breaker cutoff at school \textit{2}, which is $\tau_{\mathit{2}}$.  
    \item[School \textit{3}:] Since $t_{i\mathit{3}}(\delta)=c$, the third line of \eqref{equation_main_theorem2} applies. The local score at school \textit{3} is the probability of not being assigned to school \textit{2}, that is, $1-\tau_{\mathit{2}}$, times $0.5$.  This last term is the probability associated with being local to the cutoff at school \textit{3}.  
    \item[School \textit{4}:] MID at school \textit{4} is determined by the tie-breaker cutoff at school \textit{2}. When MID exceeds the tie-breaker cutoff at school \textit{4}, then school \textit{4} assignment probability is zero. Otherwise, since $t_{i\mathit{3}}(\delta)=c$
    and school \textit{4} is a lottery school, the second line of \eqref{equation_main_theorem2} applies.   
    The probability is therefore $0.5$ times the difference between the cutoff at school \textit{4} and MID.
    \item[School \textit{5}:] MID at school \textit{5} is determined by the larger of the tie-breaker cutoffs at school \textit{2} and school \textit{4}. Since $t_{i\mathit{5}}(\delta)=c$, the third line of \eqref{equation_main_theorem2} applies, and the probability is determined by $(0.5)^2$ times $\lambda$, the disqualification probability at lottery schools.
    \item[School \textit{6}:] Finally, since $t_{i\mathit{6}}(\delta)=a$, the first line of \eqref{equation_main_theorem2} applies and the local score becomes $(0.5)^2$ times $\lambda$.
    \end{itemize}

Since $t_{i\mathit{6}}(\delta)=a$, the probabilities sum to 1.  If $\tau_{\mathit{2}} \geq \tau_{\mathit{4}}$, the probability of any assignment is $\tau_{\mathit{2}}+0.5 \times (1-\tau_{\mathit{2}})+2\times(0.5)^2\times(1-\tau_{\mathit{2}})=1$.
If $\tau_{\mathit{2}}<\tau_{\mathit{4}}$, the probability is $\tau_{\mathit{2}}+0.5 \times (1-\tau_{\mathit{4}})+0.5 \times (\tau_{\mathit{4}} - \tau_{\mathit{2}})+2 \times 0.5^2 \times (1-\tau_{\mathit{4}})=1$.

\section{Additional Results and Proofs}

\subsection{The DA Propensity Score}\label{globalrisk}

This appendix derives the DA propensity score defined as the probability of assignment conditional on type for all applicants, without regard to cutoff proximity.  The serial dictatorship propensity score discussed in Section \ref{sec:SDscoredetail} is a special case of this.

$MID^{v}_{\theta s}$ and priority status determine DA propensity score with general tie-breakers. For this proposition, we assume that tie-breakers $R_{iv}$ and $R_{iv'}$ are independent for $v \neq v'$.


\begin{proposition}[The DA Propensity Score with General Tie-breaking]\label{theorem:multi} 
	Consider DA with multiple tie-breakers indexed by $v$, 
	distributed independently of one another according to $F_v(r|\theta)$. For all $s$ and $\theta$ in this match,
	\begin{equation*}
	\begin{split}
	p_{s}(\theta)=
	\left\{
	\begin{array}
	[c]{ll}
	0 & \text{ if }\rho_{\theta s}>\rho_s\\
	\prod_{v}(1-F_v(MID^{v}_{\theta s}|\theta)) & \text{ if }\rho_{\theta s}<\rho_s \\
	\prod_{v\neq v(s)}(1-F_v(MID^{v}_{\theta s}|\theta))\\
		\times\max\left\{ 0,F_{v(s)}(\tau_s|\theta)-F_{v(s)}(MID^{{v(s)}}_{\theta s}|\theta)\right\}    & \text{ if }\rho_{\theta s}=\rho_s \\
	\end{array}
	\right.
	\end{split}
	\end{equation*}
	where $F_{v(s)}(\tau_s|\theta)=\tau_s$ and $F_{v(s)}(MID^{{v(s)}}_{\theta s}|\theta) = MID^{v(s)}_{\theta s}$ when $v(s) \in \{1, ..., U\}$.
\end{proposition}

Proposition \ref{theorem:multi}, which generalizes an earlier multiple lottery tie-breaker result in  \cite{mdrd1:17}, covers three sorts of applicants. First, applicants with less-than-marginal priority at $s$ have no chance of being seated there.  The second line of the theorem reflects the likelihood of qualification at schools preferred to $s$ among applicants surely seated at $s$ when they can't do better.   Since tie-breakers are assumed independent, the probability of \text{not} doing better than $s$ is described by a product over tie-breakers, $\prod_{v}(1-F_v(MID^{v}_{\theta s}|\theta))$. If type $\theta$ is sure to do better than $s$, then $MID^{v}_{\theta s}=1$ and the probability at $s$ is zero.

Finally, the probability for applicants with $\rho_{\theta_i s}=\rho_s$ multiplies the term
$$\prod_{v\neq v(s)}(1-F_v(MID^{v}_{\theta s}|\theta))$$
by
$$\max\left\{ 0,F_{v(s)}(\tau_s|\theta)-F_{v(s)}(MID^{{v(s)}}_{\theta s}|\theta)\right\}.$$
The first of these is the probability of failing to improve on $s$ by virtue of being seated at schools using a tie-breaker \textit{other} than $v(s)$.  
The second parallels assignment probability in single-tie-breaker serial dictatorship: to be seated at $s$, applicants in $\rho_{\theta_i s}=\rho_s$ must have $R_{iv(s)}$ between $MID^{{v(s)}}_{\theta s}$ and $\tau_s$. 

Proposition \ref{theorem:multi} allows for single tie-breaking, lottery tie-breaking, or a mix of non-lottery and lottery tie-breakers as in the NYC high school match.
With a single tie-breaker, the propensity score formula simplifies, omitting product terms over $v$:  
\begin{corollary}[\cite{mdrd1:17}] \label{proposition:SingleR} 
	Consider DA using a single tie-breaker, $R_i$, distributed according to $F_R(r|\theta)$ for type $\theta$. For all $s$ and $\theta$ in this market, we have:
	\begin{equation*}
	p_{s}(\theta)=\left\{
	\begin{array}
	[c]{ll}
	0 & \text{ if }\rho_{\theta s}>\rho_s,\\
	1-F_R(MID_{\theta s}|\theta) & \text{ if }\rho_{\theta s}<\rho_s, \\
	(1-F_R(MID_{\theta s}|\theta))\times\max\left\{ 0,\dfrac{F_R(\tau_s|\theta)-F_R(MID_{\theta s}|\theta)}{1-F_R(MID_{\theta s}|\theta)}\right\}  & \text{ if }\rho_{\theta s}=\rho_s,\\
	\end{array}
	\right.
	\end{equation*}
	where $p_{s}(\theta)=0$ when $MID_{\theta s}=1$ and $\rho_{\theta s}=\rho_s$, and $MID_{\theta s}$ is as defined in Section \ref{sec:SDsection}, applied to a single tie-breaker. 
\end{corollary}

Common lottery tie-breaking for all schools further simplifies the DA propensity score.  When $v(s)=1$ for all $s$, $F_R(MID_{\theta s})=MID_{\theta s}$ and $F_R(\tau_s|\theta)=\tau_s$, as in the Denver match analyzed by \cite{mdrd1:17}.
In this case, the DA propensity score is a function only of $MID_{\theta s}$ and the classification of applicants into being never, always, and conditionally seated.   This contrasts with the scores in Propositions \ref{proposition:SingleR} and \ref{theorem:multi}, which depend on the unknown and unrestricted conditional distributions of tie-breakers given type ($F_R(\tau_{s}| \theta)$ and $F_R(MID_{\theta s}| \theta)$ with a single tie-breaker; $F_v(\tau_{s}| \theta)$ and $F_v(MID_{\theta s}| \theta)$ with general tie-breakers).  We therefore turn again to the local propensity score to isolate assignment variation that is independent of type and potential outcomes.

\subsubsection*{Proof of Proposition \ref{theorem:multi}} 

We prove Proposition \ref{theorem:multi} using a strategy to that used in the proof of Theorem 1 in \cite{mdrd1:17}.
Note first that admissions cutoffs $\mathbf{\xi}$ in a large market do not depend on the realized tie-breakers $r_{iv}$'s:
DA in the large market depends on the $r_{iv}$'s only through $G(I_0)$, defined as the fraction of applicants in set $I_0 = \{i\in I  \hspace*{.1cm} | \hspace*{.1cm} \theta_i\in \Theta_0, r_{iv}\leq r_v\text{ for all }v\}$ with various choices of $\Theta_0$ and $r_v$.
In particular, $G(I_0)$ doesn't depend on tie-breaker realizations in the large market.  For the empirical CDF of each tie-breaker conditional on each type, $\hat F_v(\cdot| \theta)$, the
Glivenko-Cantelli theorem for independent but non-identically distributed random variables implies
$\hat F_v(\cdot| \theta)=F_v(\cdot| \theta)$ for any $v$ and $\theta$ \citep{wellner1981glivenko}.
Since cutoffs $\mathbf{\xi}$ are constant, marginal priority $\rho_{s}$ is also constant for every school $s$.

Now, consider the propensity score for school $s.$ First, applicants who don't rank $s$ have $p_{s}(\theta)=0$. 
If $\rho_{\theta s}>\rho_s,$ then $\rho_{\theta s}>\rho_{s}.$ Therefore, $$p_{s}(\theta)=0 \text{ if }\rho_{\theta s}>\rho_s\text{ or $\theta$ does not rank $s$}.$$

Second, if $\rho_{\theta s}\leq\rho_s$, then the type $\theta$ applicant may be
assigned a preferred school $\tilde{s}\in B_{\theta s},$ where $\rho_{\theta \tilde{s}} = \rho_{\tilde{s}}$. 
For each tie-breaker $v$, the proportion of type $\theta$ applicants assigned some $\tilde{s}\in B^v_{\theta s}$ where $\rho_{\theta \tilde{s}} = \rho_{\tilde{s}}$
is $F_v(MID^{v}_{\theta s}|\theta)$. 
This means that for each $v$, the probability of not being assigned any $\tilde{s}\in
B^v_{\theta s}$ where $\rho_{\theta \tilde{s}} = \rho_{\tilde{s}}$ is $1-F_v(MID^{v}_{\theta s}|\theta).$ 
Since tie-breakers are assumed to be distributed independently of one another, the probability of not being assigned any $\tilde{s}\in
B_{\theta s}$ where $\rho_{\theta \tilde{s}} = \rho_{\tilde{s}}$ for a type $\theta$ applicant is $\Pi_v (1-F_v(MID^{v}_{\theta s}|\theta)).$ 
Every
applicant of type $\rho_{\theta s}<\rho_s$ who is not assigned a preferred choice
is assigned $s$ because $\rho_{\theta s}<\rho_{s}.$ So%
\[
p_{s}(\theta)=\Pi_v (1-F_v(MID^{v}_{\theta s}|\theta)) \text{ if }\rho_{\theta s}<\rho_s.
\]

Finally, consider applicants of type $\rho_{\theta s}=\rho_s$ who are not assigned a
choice preferred to $s$. The fraction of applicants $\rho_{\theta s}=\rho_s$ who are
not assigned a preferred choice is $\Pi_v (1-F_v(MID^{v}_{\theta s}|\theta))$. 
Also, the values of the tie-breaking variable $v(s)$ of these applicants are larger than $MID^{v(s)}_{\theta s}$. 
If $\tau
_{s}<MID^{v(s)}_{\theta s},$ then no such applicant is assigned $s.$ If
$\tau_{s}\geq MID^{v(s)}_{\theta s},$ then the fraction of applicants who are
assigned $s$ within this set is given by $\frac{F_{v(s)}(\tau_{s}|\theta)-F_{v(s)}(MID^{v(s)}_{\theta
s}|\theta)}{1-F_{v(s)}(MID^{v(s)}_{\theta s}|\theta)}.$\ Hence, conditional on $\rho_{\theta s}=\rho_s$ and not being assigned a choice higher than $s,$ the probability of
being assigned $s$ is given by $\max\{0, \frac{F_{v(s)}(\tau_{s}|\theta)-F_{v(s)}(MID^{v(s)}_{\theta
s}|\theta)}{1-F_{v(s)}(MID^{v(s)}_{\theta s}|\theta)}\}.$\ Therefore,
\[
p_{s}(\theta)=\prod_{v\neq v(s)}(1-F_v(MID^{v}_{\theta s}|\theta))\times\max\left\{ 0,F_{v(s)}(\tau_s|\theta)-F_{v(s)}(MID^{{v(s)}}_{\theta s}|\theta)\right\}   \text{ if }\rho_{\theta s}=\rho_s.
\]

\subsection{Proof of Theorem \ref{theorem:main}}
\noindent The proof uses lemmas established below. 
The first lemma shows that the vector of DA cutoffs computed for the sampled market, $\hat{\xi}_N$, converges to the vector of cutoffs in the continuum.

\begin{lemma}
\label{CutoffConvergence} (Cutoff almost sure convergence) $\hat{\mathbf{\xi}}_N\overset{a.s.}{\longrightarrow} \mathbf{\xi}$ where $\mathbf{\xi}$ denotes the vector of continuum market cutoffs.
\end{lemma}

\noindent This result implies that the estimated score converges to the large-market local score as market size grows and bandwidth shrinks.  

 
 \begin{lemma}
\label{estimatedPScoreConvergence} (Estimated local propensity score almost sure convergence)
For all $\theta\in\Theta, s\in S,$ and $T\in \{a, c, n\}^{S}$, we have $\hat \psi_{s}(\theta, T(\delta_N))\overset{a.s.}{\longrightarrow}\psi_{s}(\theta, T)$ as $N\rightarrow \infty$ and $\delta_N\rightarrow 0$.
\end{lemma}
 
The next lemma shows that the true finite market score with a fixed bandwidth, defined as $\psi_{Ns}(\theta, T; \delta_N) \equiv E_N[D_i(s)| \theta_i=\theta, T_i(\delta_N)=T]$, also converges to $\psi_{s}(\theta, T)$ as market size grows and bandwidth shrinks. 

\begin{lemma}
\label{PScoreConvergence} (Bandwidth-specific propensity score almost sure convergence)
For all $\theta\in\Theta, s\in S,$ $T\in \{a, c, n\}^{S}$,  and $\delta_N$ such that $\delta_N\rightarrow 0$ and $N\delta_N\rightarrow \infty$ as $N\rightarrow\infty$, we have $\psi_{Ns}(\theta, T; \delta_N)\overset{p}{\longrightarrow}\psi_{s}(\theta, T)$ as $N\rightarrow \infty$.
\end{lemma}

Finally, the definitions of $\psi_{Ns}(\theta, T; \delta_N)$ and $\psi_{Ns}(\theta, T)$ imply that  $|\psi_{Ns}(\theta, T; \delta_N)-\psi_{Ns}(\theta, T)|\overset{a.s.}{\longrightarrow}0$ as $\delta_N\rightarrow 0$. 
Combining these results shows that for all $\theta\in\Theta, s\in S,$ and $T$, as $N\rightarrow\infty$ and $\delta_N\rightarrow 0$ with $N\delta_N\rightarrow \infty$, we have
\begin{align*}
&|\hat \psi_{s}(\theta, T(\delta_N))-\psi_{Ns}(\theta, T)| \\
=&|\hat \psi_{s}(\theta, T(\delta_N))-\psi_{Ns}(\theta, T; \delta_N)+\psi_{Ns}(\theta, T; \delta_N)-\psi_{Ns}(\theta, T)|\\
\leq&|\hat \psi_{s}(\theta, T(\delta_N))-\psi_{Ns}(\theta, T; \delta_N)|+|\psi_{Ns}(\theta, T; \delta_N)-\psi_{Ns}(\theta, T)|\\
\overset{p}{\longrightarrow}&|\psi_{s}(\theta, T)-\psi_{s}(\theta, T)|+0\\
=&0.
\end{align*}

\noindent This yields the theorem since $\Theta, S$, and $\{n, c, a\}^S$ are finite. 


\subsubsection*{Proof of Lemma \ref{CutoffConvergence}}

The proof of Lemma \ref{CutoffConvergence} is analogous to the proof of Lemma 3 in \cite{mdrd1:17} and available upon request. 
The main difference is that to deal with multiple non-lottery tie-breakers, the proof of Lemma \ref{CutoffConvergence} needs to invoke the continuous differentiability of $F^i_v(r| e)$ and the Glivenko-Cantelli theorem for independent but non-identically distributed random variables \citep{wellner1981glivenko}. 

\subsubsection*{Proof of Lemma \ref{estimatedPScoreConvergence}}

$\hat \psi_{s}(\theta, T(\delta_N))$ is almost everywhere continuous in finite sample cutoffs $\hat{\mathbf{\xi}}_N$, finite sample MIDs ($MID^v_{\theta s}$), and bandwidth $\delta_N$.
Since every $MID^v_{\theta s}$ is almost everywhere continuous in finite sample cutoffs $\hat{\mathbf{\xi}}_N$, $\hat \psi_{s}(\theta, T(\delta_N))$ is almost everywhere continuous in finite sample cutoffs $\hat{\mathbf{\xi}}_N$ and bandwidth $\delta_N$.
Recall $\delta_N\rightarrow 0$ by assumption while $\hat{\mathbf{\xi}}_N\overset{a.s.}{\longrightarrow} \mathbf{\xi}$ by Lemma \ref{CutoffConvergence}.
Therefore, by the continuous mapping theorem, as $N\rightarrow \infty$, $\hat \psi_{s}(\theta, T(\delta_N))$ almost surely converges to $\hat \psi_{s}(\theta, T(\delta_N))$ with $\mathbf{\xi}$ replacing $\hat{\mathbf{\xi}}_N$, which converges to $\psi_{s}(\theta, T)$ as $\delta_N\rightarrow 0$.

\subsubsection*{Proof of Lemma \ref{PScoreConvergence}}

We use the following fact, which is implied by Example 19.29 in \cite{vdv:00}.

\begin{lemma}\label{lemma_local_gc}
Let $X$ be a random variable distributed according to some CDF $F$ over $[0, 1]$.
Let $F(\cdot | X\in [x-\delta, x+\delta])$ be the conditional version of $F$ conditional on $X$ being in a small window $[x-\delta, x+\delta]$ where $x\in [0, 1]$ and $\delta\in (0, 1]$.
Let $X_1, ..., X_N$ be iid draws from $F$.
Let $\hat F_N$ be the empirical CDF of $X_1, ..., X_N$.
Let $\hat F_N(\cdot | X\in[x-\delta, x+\delta])$ be the conditional version of $\hat F_N$ conditional on a subset of draws falling in $[x-\delta, x+\delta]$, i.e., $\{X_i| i=1, ..., n, X_i \in [x-\delta, x+\delta]\}$.
Suppose $(\delta_N)$ is a sequence with $\delta_N\downarrow 0$ and $\delta_N\times N \rightarrow \infty$.
Then $\hat F_N(\cdot | X\in[x-\delta_N, x+\delta_N])$ uniformly converges to $F(\cdot | X\in [x-\delta_N, x+\delta_N])$, i.e.,
$$\sup_{x'\in [0, 1]} |\hat F_N(x'| X\in[x-\delta_N, x+\delta_N])-F(x'| X\in[x-\delta_N, x+\delta_N])|\rightarrow_{p}0\text{ as }N\rightarrow\infty\text{ and }\delta_N\rightarrow 0.$$
\end{lemma}

\begin{proof}[Proof of Lemma \ref{lemma_local_gc}]
We first prove the statement for $x\in(0,1)$.
Let $P$ be the probability measure of $X$ and $\hat P_N$ be the empirical measure of $X_1,...,X_N$.
Note that
\begin{align*}
	&\sup_{x'\in [0, 1]} |\hat F_N(x'| X\in[x-\delta_N, x+\delta_N])-F(x'| X\in[x-\delta_N, x+\delta_N])|\\
	=&\sup_{t\in [-1, 1]} |\hat F_N(x+t\delta_N| X\in[x-\delta_N, x+\delta_N])-F(x+t\delta_N| X\in[x-\delta_N, x+\delta_N])|\\
	=&\sup_{t\in [-1, 1]} |\frac{\hat P_N[x-\delta_N,x+t\delta_N]}{\hat P_N[x-\delta_N, x+\delta_N]}-\frac{P_{X}[x-\delta_N,x+t\delta_N]}{P_{X}[x-\delta_N, x+\delta_N]}|\\
	=&\frac{1}{\hat P_N[x-\delta_N, x+\delta_N]P_{X}[x-\delta_N, x+\delta_N]}\\
	&\times\sup_{t\in [-1, 1]} |\hat P_N[x-\delta_N,x+t\delta_N]P_{X}[x-\delta_N,x+\delta_N]-\hat P_N[x-\delta_N,x+\delta_N]P_{X}[x-\delta_N,x+t\delta_N]|\\
	=&\frac{1}{\hat P_N[x-\delta_N, x+\delta_N]P_{X}[x-\delta_N, x+\delta_N]}\\
	&\times\sup_{t\in [-1, 1]} |\hat P_N[x-\delta_N,x+t\delta_N] (P_{X}[x-\delta_N,x+\delta_N]-\hat P_N[x-\delta_N,x+\delta_N])\\
	&~~~~+\hat P_N[x-\delta_N,x+\delta_N](\hat P_N[x-\delta_N,x+t\delta_N]-P_{X}[x-\delta_N,x+t\delta_N])|\\
	\le &\frac{1}{\hat P_N[x-\delta_N, x+\delta_N]P_{X}[x-\delta_N, x+\delta_N]}\\
	&\times\{\sup_{t\in [-1, 1]} \hat P_N[x-\delta_N,x+t\delta_N] |\hat P_N[x-\delta_N,x+\delta_N]-P_{X}[x-\delta_N,x+\delta_N]|\\
	&~~~~+\sup_{t\in [-1, 1]}\hat P_N[x-\delta_N,x+\delta_N]|\hat P_N[x-\delta_N,x+t\delta_N]-P_{X}[x-\delta_N,x+t\delta_N]|\} \\
	&= \frac{1}{P_{X}[x-\delta_N, x+\delta_N]} \times\{|\hat P_N[x-\delta_N,x+\delta_N]-P_{X}[x-\delta_N,x+\delta_N]| \\ &+\sup_{t\in [-1, 1]} |\hat P_N[x-\delta_N,x+t\delta_N]-P_{X}[x-\delta_N,x+t\delta_N]|\}\\
	=&\frac{A_N}{P_{X}[x-\delta_N, x+\delta_N]},
\end{align*}
where $$A_N=|\hat P_N[x-\delta_N,x+\delta_N]-P_{X}[x-\delta_N,x+\delta_N]|+\sup_{t\in [-1, 1]} |\hat P_N[x-\delta_N,x+t\delta_N]-P_{X}[x-\delta_N,x+t\delta_N]|.$$ The above inequality holds by the triangle inequality and the second last equality holds because $\sup_{t\in [-1, 1]}\hat P_N[x-\delta_N,x+t\delta_N]=\hat P_N[x-\delta_N,x+\delta_N]$.

We show that $A_N/P_{X}[x-\delta_N, x+\delta_N]\stackrel{p}{\longrightarrow} 0$. Example 19.29 in \cite{vdv:00} implies that the sequence of processes $\{\sqrt{n/\delta_N}(\hat P_N[x-\delta_N,x+t\delta_N]-P_{X}[x-\delta_N,x+t\delta_N]):t\in [-1,1]\}
$ converges in distribution to a Gaussian process in the space of bounded functions on $[-1,1]$ as $N\rightarrow\infty$. We denote this Gaussian process by $\{\mathbb{G}_t:t\in [-1,1]\}$.
We then use the continuous mapping theorem to obtain
$$
\sqrt{n/\delta_N}A_N \stackrel{d}{\longrightarrow} |\mathbb{G}_1|+\sup_{t\in [-1, 1]} |\mathbb{G}_t|
$$
as $N\rightarrow \infty$.
Since $\{\mathbb{G}_t:t\in [-1,1]\}$ has bounded sample paths, it follows that $|\mathbb{G}_1|<\infty$ and $\sup_{t\in [-1, 1]} |\mathbb{G}_t|<\infty$ for sure.
By the continuous mapping theorem, under the condition that $N\delta_N\rightarrow \infty$, 
\begin{align*}
	(1/\delta_N) A_N &= (1/\sqrt{N\delta_N})\times \sqrt{n/\delta_N}A_N\\
	&\stackrel{d}{\longrightarrow} 0\times (|\mathbb{G}_1|+\sup_{t\in [-1, 1]} |\mathbb{G}_t|)\\
	&=0.
\end{align*}
This implies that $(1/\delta_N) A_N\stackrel{p}{\longrightarrow} 0$, because for any $\epsilon>0$,
\begin{align*}
	\Pr(|(1/\delta_N)A_N|>\epsilon)&=\Pr((1/\delta_N)A_N<-\epsilon)+\Pr((1/\delta_N)A_N>\epsilon)\\
	&\le\Pr((1/\delta_N)A_N\le-\epsilon)+1-\Pr((1/\delta_N)A_N\le\epsilon)\\
	&\rightarrow \Pr(0\le-\epsilon)+1-\Pr(0\le\epsilon)\\
	&=0,
\end{align*}
where the convergence holds since $(1/\delta_N) A_N\stackrel{d}{\longrightarrow} 0$.
To show that $A_N/P_{X}[x-\delta_N, x+\delta_N]\stackrel{p}{\longrightarrow} 0$, it is therefore enough to show that $\lim_{N\rightarrow \infty}(1/\delta_N)P_{X}[x-\delta_N, x+\delta_N]>0$.
We have
\begin{align*}
	(1/\delta_N)P_{X}[x-\delta_N, x+\delta_N]
	&= (1/\delta_N)(F_{X}(x+\delta_N)-F_{X}(x-\delta_N))\\
	&= (1/\delta_N)(2f(x)\delta_N+o(\delta_N))\\
	&= 2f(x)+o(1)\\
	&\rightarrow 2f(x)\\
	&>0,
\end{align*}
where we use Taylor's theorem for the second equality and the assumption of $f(x)>0$ for the last inequality.

We next prove the statement for $x=0$. Note that
\begin{align*}
	&\sup_{x'\in [0, 1]} |\hat F_N(x'| X\in[-\delta_N, \delta_N])-F(x'| X\in[-\delta_N, \delta_N])|\\
	=&\sup_{t\in [0, 1]} |\hat F_N(t\delta_N| X\in[0, \delta_N])-F(t\delta_N| X\in[0, \delta_N])|\\
	=&\sup_{t\in [0, 1]} |\frac{\hat F_N(t\delta_N)}{\hat F_N(\delta_N)}-\frac{F_{X}(t\delta_N)}{F_{X}(\delta_N)}|\\
	=&\frac{1}{\hat F_N(\delta_N)F_{X}(\delta_N)}\sup_{t\in [0, 1]} |\hat F_N(t\delta_N)F_{X}(\delta_N)-\hat F_N(\delta_N)F_{X}(t\delta_N)|\\
	=&\frac{1}{\hat F_N(\delta_N)F_{X}(\delta_N)}\sup_{t\in [0, 1]} |\hat F_N(t\delta_N) (F_{X}(\delta_N)-\hat F_N(\delta_N))+\hat F_N(\delta_N)(\hat F_N(t\delta_N)-F_{X}(t\delta_N))|\\
	\le &\frac{1}{\hat F_N(\delta_N)F_{X}(\delta_N)}\{\sup_{t\in [0, 1]} \hat F_N(t\delta_N) |\hat F_N(\delta_N)-F_{X}(\delta_N)|+\sup_{t\in [0, 1]}\hat F_N(\delta_N)|\hat F_N(t\delta_N)-F_{X}(t\delta_N)|\}\\
	=&\frac{1}{F_{X}(\delta_N)}\{|\hat F_N(\delta_N)-F_{X}(\delta_N)|+\sup_{t\in [0, 1]} |\hat F_N(t\delta_N)-F_{X}(t\delta_N)|\}
	=\frac{A_N^0}{F_{X}(\delta_N)},
\end{align*}
where $A_N^0=|\hat F_N(\delta_N)-F_{X}(\delta_N)|+\sup_{t\in [0, 1]}|\hat F_N(t\delta_N)-F_{X}(t\delta_N)|$.
By the argument used in the above proof for $x\in (0,1)$, we have $(1/\delta_N) A_N^0\stackrel{p}{\longrightarrow} 0$. It also follows that
\begin{align*}
	(1/\delta_N)F_{X}(\delta_N)
	&= (1/\delta_N)(f(0)\delta_N+o(\delta_N))\\
	&= f(0)+o(1)\\
	&\rightarrow f(0)\\
	&>0.
\end{align*}
Thus, $\frac{A_N^0}{F_{X}(\delta_N)}\stackrel{p}{\longrightarrow} 0$, and hence $\sup_{x'\in [0, 1]} |\hat F_N(x'| X\in[-\delta_N, \delta_N])-F(x'| X\in[-\delta_N, \delta_N])|\stackrel{p}{\longrightarrow} 0$.
The proof for $x=1$ follows from the same argument.
\end{proof}
\noindent

Consider any deterministic sequence of economies $\{g_N\}$ such that $g_N\in \mathcal{G}$ for all $N$ and
$g_N\rightarrow G$ in the $(\mathcal{G},d)$ metric space. 
Let $(\delta_N)$ be an associated sequence of positive numbers (bandwidths) such that $\delta_N\rightarrow 0$ and $N\delta_N\rightarrow \infty$ as $N\rightarrow\infty$. 
Let $\psi_{Ns}(\theta, T; \delta_N)\equiv E_N[D_i(s)| \theta_i=\theta, T_i(\delta_N)=T]$ be the (finite-market, deterministic) bandwidth-specific propensity score for particular $g_N$ and $\delta_N$.

For Lemma \ref{PScoreConvergence}, it is enough to show deterministic convergence of this finite-market score, that is, $\psi_{Ns}(\theta, T; \delta_N)\rightarrow \psi_{s}(\theta, T)$ as $g_N\rightarrow G$ and $\delta_N\rightarrow 0$.
To see this, let $G_N$ be the distribution over $I(\Theta_{0}, r_0, r_{1})$'s induced by randomly drawing $N$ applicants from $G$, where $I(\Theta_{0}, r_0, r_{1})\equiv\{i| \theta_i\in \Theta_{0}, r_0<r_i\leq r_1\}$. 
Note that $G_N$ is random and that $G_N\overset{a.s.}{\rightarrow}G$ by \cite{wellner1981glivenko}'s Glivenko-Cantelli theorem for independent but non-identically distributed random variables.
$G_N\overset{p}{\rightarrow}G$ and $\psi_{Ns}(\theta, T; \delta_N)\rightarrow \psi_{s}(\theta, T)$ allow us to apply the Extended Continuous Mapping Theorem (Theorem 18.11 in \cite{vdv:00}) to obtain $\tilde{\psi}_{Ns}(\theta, T; \delta_N)\overset{p}{\longrightarrow}\psi_{s}(\theta, T)$ where $\tilde{\psi}_{Ns}(\theta, T; \delta_N)$ is the random version of $\psi_{Ns}(\theta, T; \delta_N)$ defined for $G_N$.

For notational simplicity, consider the single-school RD case, where there is only one school $s$ making assignments based on a single non-lottery tie-breaker $v(s)$ (without using any priority). 
A similar argument with additional notation shows the result for DA with general tie-breaking. 

For any $\delta_N>0,$ whenever $T_i(\delta_N)=a$, it is the case that $D_i(s)=1$. 
As a result, 
$$\psi_{Ns}(\theta, a; \delta_N) \equiv E_N[D_i(s)| \theta_i=\theta, T_i(\delta_N)=a]=1\equiv\psi_{s}(\theta, a).$$ 

\noindent Therefore, $\psi_{Ns}(\theta, a; \delta_N)\rightarrow\psi_{s}(\theta, a)$ as $N\rightarrow \infty$.
Similarly, for any $\delta_N>0,$ whenever $T_i(\delta_N)=n$, it is the case that $D_i(s)=0$. 
As a result, 
$$\psi_{Ns}(\theta, n; \delta_N) \equiv E_N[D_i(s)| \theta_i=\theta, T_i(\delta_N)=n]=0\equiv\psi_{s}(\theta, n).$$ 

\noindent Therefore, $\psi_{Ns}(\theta, n; \delta_N)\rightarrow\psi_{s}(\theta, n)$ as $N\rightarrow \infty$. 
Finally, when $T_i(\delta_N)=c$, let $$F_{N, v(s)}(r| \theta)\equiv \dfrac{\sum^N_{i=1}1\{\theta_i=\theta\}F_{v(s)}^i(r)}{\sum^N_{i=1}1\{\theta_i=\theta\}}$$ be the aggregate tie-breaker distribution conditional on each applicant type $\theta$ in the finite market. 
$\mathbf{\tilde{\xi}}_{Ns}$ denotes the random cutoff at school $s$ in a realized economy $g_N$. 
For any $\epsilon$, there exists $N_0$ such that for any $N>N_0$, we have
\begin{align*}
\psi_{Ns}(\theta, c; \delta_N)&\equiv E_N[D_i(s)| \theta_i=\theta, T_i(\delta_N)=c]\\
&=P_N[R_{iv(s)}\leq \mathbf{\tilde{\xi}}_{Ns}| \theta_i=\theta, R_{iv(s)}\in(\mathbf{\tilde{\xi}}_{Ns}-\delta_N, \mathbf{\tilde{\xi}}_{Ns}+\delta_N]] \\
&\in (P[R_{iv(s)}\leq \mathbf{\xi}_{s}| \theta_i=\theta, R_{iv(s)}\in(\mathbf{\xi}_{s}-\delta_N, \mathbf{\xi}_{s}+\delta_N]]-\epsilon/2, \\
&\ \ \ \ \ P[R_{iv(s)}\leq \mathbf{\xi}_{s}| \theta_i=\theta, R_{iv(s)}\in(\mathbf{\xi}_{s}-\delta_N, \mathbf{\xi}_{s}+\delta_N]]+\epsilon/2),
\end{align*}

\noindent where $\mathbf{\xi}_{s}$ is school $s$'s continuum cutoff, $P$ is the probability induced by the tie-breaker distributions in the continuum economy, and the inclusion is by Assumption \ref{assumption_differentiability} and Lemmata \ref{CutoffConvergence} and \ref{lemma_local_gc}.  
Again for any $\epsilon$, there exists $N_0$ such that for any $N>N_0$, we have
\begin{align*}
&(P[R_{iv(s)}\leq \mathbf{\xi}_{s}| \theta_i=\theta, R_{iv(s)}\in(\mathbf{\xi}_{s}-\delta_N, \mathbf{\xi}_{s}+\delta_N]]-\epsilon/2, \\
&\ P[R_{iv(s)}\leq \mathbf{\xi}_{s}| \theta_i=\theta, R_{iv(s)}\in(\mathbf{\xi}_{s}-\delta_N, \mathbf{\xi}_{s}+\delta_N]]+\epsilon/2)\\
&=(\dfrac{F_{v(s)}(\mathbf{\xi}_{s}| \theta)-F_{v(s)}(\mathbf{\xi}_{s}-\delta_N| \theta)}{F_{v(s)}(\mathbf{\xi}_{s}+\delta_N| \theta)-F_{v(s)}(\mathbf{\xi}_{s}-\delta_N| \theta)}-\epsilon/2, \\
&\ \ \ \ \ \dfrac{F_{v(s)}(\mathbf{\xi}_{s}| \theta)-F_{v(s)}(\mathbf{\xi}_{s}-\delta_N| \theta)}{F_{v(s)}(\mathbf{\xi}_{s}+\delta_N| \theta)-F_{v(s)}(\mathbf{\xi}_{s}-\delta_N| \theta)}+\epsilon/2) \\
&=(\dfrac{\{F_{v(s)}(\mathbf{\xi}_{s}| \theta)-F_{v(s)}(\mathbf{\xi}_{s}-\delta_N| \theta)\}/\delta_N}{\{F_{v(s)}(\mathbf{\xi}_{s}+\delta_N| \theta)-F_{v(s)}(\mathbf{\xi}_{s}| \theta)\}/\delta_N+\{F_{v(s)}(\mathbf{\xi}_{s}| \theta)-F_{v(s)}(\mathbf{\xi}_{s}-\delta_N| \theta)\}/\delta_N}-\epsilon/2, \\
&\ \ \ \ \ \dfrac{\{F_{v(s)}(\mathbf{\xi}_{s}| \theta)-F_{v(s)}(\mathbf{\xi}_{s}-\delta_N| \theta)\}/\delta_N}{\{F_{v(s)}(\mathbf{\xi}_{s}+\delta_N| \theta)-F_{v(s)}(\mathbf{\xi}_{s}| \theta)\}/\delta_N+\{F_{v(s)}(\mathbf{\xi}_{s}| \theta)-F_{v(s)}(\mathbf{\xi}_{s}-\delta_N| \theta)\}/\delta_N}+\epsilon/2) \\
&\in (0.5-\epsilon, 0.5+\epsilon)\\
&=(\psi_{s}(\theta, c)-\epsilon, \psi_{s}(\theta, c)+\epsilon), 
\end{align*}

\noindent completing the proof.

\makeatletter 
\renewcommand{\thetablenums}{\thesection\@arabic\c@tablenums}
\makeatother

\newpage
\section{Empirical Appendix} \setcounter{tablenums}{0}

\subsection{Data}

The NYC DOE provided data on students, schools, the rank-order lists submitted by match participants, school assignments, and outcome variables. Applicants and programs are uniquely identified by a number that can be used to merge data sets. Students with a record in assignment files who cannot be matched to other files are omitted. 

\subsubsection{Applicant Data}

We focus on first-time applicants to the NYC public (unspecialized) high school system who live in NYC and attended a public middle school in eighth grade. The NYC high school match is conducted in three rounds. The data used for the present analyses are from the first assignment round, which uses DA and we refer to as \textit{main round}. Applicants who were not assigned after the main round apply to the remaining seats in a subsequent \textit{supplementary round}. Students who remain unassigned in the supplementary round are then assigned on a case-by-case basis in the final \textit{administrative round}.  \\

\noindent \textbf{Assignment, Priorities, and Ranks}

\noindent Data on the assignment system come from the DOE's enrollment office, and report assignments for our two cohorts. The main application data set details applicant program choices, eligibility, priority group and rank, as well as the admission procedure used at the respective program. Lottery numbers and details on assignments at Educational Option (Ed-Opt) programs are provided in separate data sets. \\

\noindent \textbf{Student Characteristics}

\noindent NYC DOE students files record grade, gender, ethnicity, and whether students attended a public middle school. Separate files include (i) student scores on middle school standardized tests, (ii) English language learner and special education status, and (iii) subsidized lunch status. Our baseline middle school scores are from 6th grade math and English exams. If a student re-took a test, the latest result is used. Our demographic characteristics come from the DOE's snapshot for 8th grade.

\subsubsection{School-level Data}

\noindent \textbf{School Letter Grades}

\noindent School grades are drawn from NYC DOE School Report Cards for 2010/11, 2011/12 and 2012/13. For each application cohort, we grade schools based on the report cards published in the school year prior to the application school year: for the 2011/12 application cohort, for instance, schools are assigned grades published in 2010/11, and similarly for the other two cohorts.\\

\noindent \textbf{School Characteristics}

\noindent School characteristics were taken from report card files provided by the DOE. These data provide information on enrollment statistics, racial composition, attendance rates, suspensions, teacher numbers and experience, and graduating class Regents Math and English performance. A unique identifier for each school allows these data to be merged with data from other sources.  The analyses on teacher experience and education reported in Table \ref{tab:stuchars} of this publication are based on the School-Level Master File 1996-2016, a dataset compiled by the Research Alliance for NYC Schools at New York University' Steinhardt School of Culture, Education, and Human Development (www.ranycs.org). All data in the School-Level Master File are publicly available. The Research Alliance takes no responsibility for potential errors in the dataset or the analysis. The opinions expressed in this publication are those of the authors and do not represent the views of the Research Alliance for NYC Schools or the institutions that posted the original publicly available data.\footnote{Research Alliance for New York City Schools (2017). School-Level Master File 1996-2016 [Data file and code book]. Unpublished data}\\

\noindent \textbf{Defining Screened and Lottery Schools}

\noindent We define lottery schools as any school hosting at least one program for which the lottery number is used as the tie-breaker. Screened schools are the remaining schools. 
Some schools allow students to share a screened tie-breaker rank, breaking screening-variable ties with lottery numbers.  Propensity scores for such schools are computed using the lottery tie breaker and schools are considered lottery in any analysis that makes this substantive distinction. Specialized high schools are considered screened schools. The remaining schools, mostly charters that conduct a separate lottery process, are considered lottery schools.

\subsubsection{SAT and Graduation Outcomes}

 \noindent \textbf{SAT Tests}

 \noindent The NYC DOE has data on SAT scores for test-takers from 2006-17. These data originate with the College Board. We use the first test for multiple takers. For applicants tested in the same month, we use the highest score. During our sample period, the SAT has been redesigned. We re-scale scores of SAT exams taken prior to the reform according to the official re-scaling scheme provided by CollegeBoard.\footnote{See \url{https://collegereadiness.collegeboard.org/educators/higher-ed/scoring/concordance} for the conversion scale.}\\

 \noindent \textbf{Graduation}

 \noindent The DOE Graduation file records the discharge status for public school students  enrolled from 2005-17. Because data on graduation results are not yet available for the youngest (2013/14) cohort, graduation results are for the two older cohorts only.\\

 \noindent \textbf{College- and Career-preparedness and College-readiness}

 \noindent The DOE provided us with individual-level indicators for college- and career-preparedness as well as college-readiness for public school students enrolled from 2005-17. Since these data are not yet available for the youngest (2013/14) cohort, the results are for the two older cohorts only. Table \ref{tabappdx:collegeindex} gives an overview on the criteria for the two indicators.

\begin{figure}[!ht]
 \centering
 \includegraphics[scale=0.8]{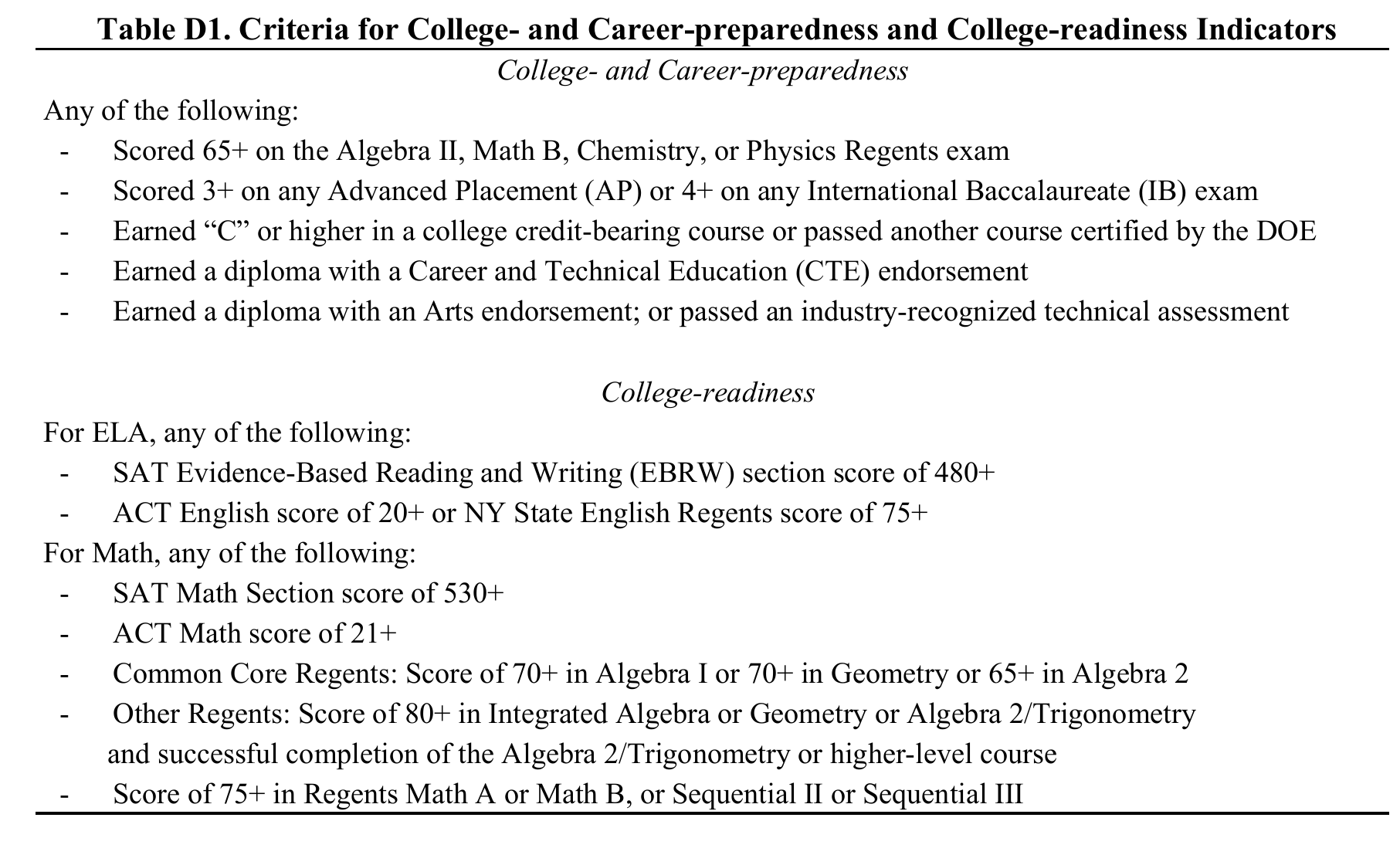} \refstepcounter{tablenums} \label{tabappdx:collegeindex}
\end{figure}

\subsubsection{Replicating the NYC Match}

NYC uses the student-proposing DA algorithm to determine assignments. The three ingredients for this algorithm are: student's ranking of up to 12 programs, program capacities and priorities, and tie-breakers.\\

\noindent \textbf{Program Assignment Rules}

\noindent Programs use a variety of assignment rules. Lottery, Limited Unscreened, and Zoned programs order students first by priority group, and within priority group by lottery number. Screened and Audition programs order students by priority group and then by a non-lottery tie-breaker, referred to as running or rank variable.  We observe these in the form of an an ordering of applicants provided by Screened and Audition programs. Ed-Opt programs use two tie-breakers, which is described into more detail below. Finally, as mentioned above, some schools allow students to share a screened tie-breaker rank, breaking screening-variables ties with lottery numbers.\\

 \noindent \textbf{Program Capacities and Priorities}

\noindent Program capacities must be imputed. We assume program capacity equals the number of assignments extended. Program type determines priorities. The priority group is a number assigned by the NYC DOE depending on addresses, program location, siblings, among other considerations, including, in some cases, whether applicants attended an information session or open house (for Limited Unscreened programs). \\

 \noindent \textbf{Lottery Numbers}

\noindent The lottery numbers are provided by the NYC DOE in a separate data set. Lottery tie-breakers are reported as unique alphanumeric string and scaled to $[0, 1]$. Lottery numbers are missing for some; we assign these applicants a randomly drawn lottery number and use it in our replicated match. It is this replicated match that is used to construct assignment instruments and their associated propensity scores.\\

 \noindent \textbf{Ranks}

\noindent Screened, Audition, and half of the seats at Ed-Opt programs assign students a rank, based on various diverging criteria, such as former test performance. Ranks are reported as an integer reflecting raw tie-breaker order in this group. We scale these so as to lie in $(0, 1]$ by transforming raw tie-breaking realizations $R_{iv}$ into $[R_{iv}-\min_j{R_{jv}}+1]/[\max_jR_{jv}-\min_jR_{jv}+1]$ for each tie-breaker $v$. At some screened programs, the rank numbers of applicants have gaps, i.e. the distribution of running variable values is discontinuous. Potential reasons include i) human error when school principals submit applicant rankings to the NYC DOE, and ii) while running variables are assigned at the program level, applications at Ed-Opt programs are treated as six separate buckets (i.e. distinct application choices), leading to artificial gaps in rank distributions (see discussion of assignment at Ed-Opt programs below).\\

 \noindent \textbf{Assignment at Educational Option programs}

\noindent Ed-Opt programs use two tie-breakers. Applicants are first categorized into high performers, middle performers, and low performers by scores on a seventh grade reading test. Ed-Opt programs aim to have an enrollment distribution of 16\% high performers, 68\% middle performers and 16\% low performers.
Half of Ed-Opt seats are assigned using the lottery tie-breaker. These seats are called ``random.'' The other half uses a rank variable such as those used by other screened programs. These seats are called ``select.''\\
We refer to the resulting six combinations as ``buckets.''  Ed-Opt applicants are treated as applying to all six. A separate data set details which bucket applicants were offered. Buckets have their own priorities and capacities. The latter are imputed based on the observed assignments to buckets. \\
Tables \ref{tabappdx:edoptschoice} and \ref{tabappdx:edoptspriority} show applicants' choice order of and priorities at Ed-Opt buckets, respectively. Both are based on consultations with the NYC DOE and our simulations of the match. \\

 \noindent High performers rank high buckets first, while medium and low performers apply to medium and low buckets first, respectively. 

\begin{figure}[ht]
 \centering
 \includegraphics[scale=0.9]{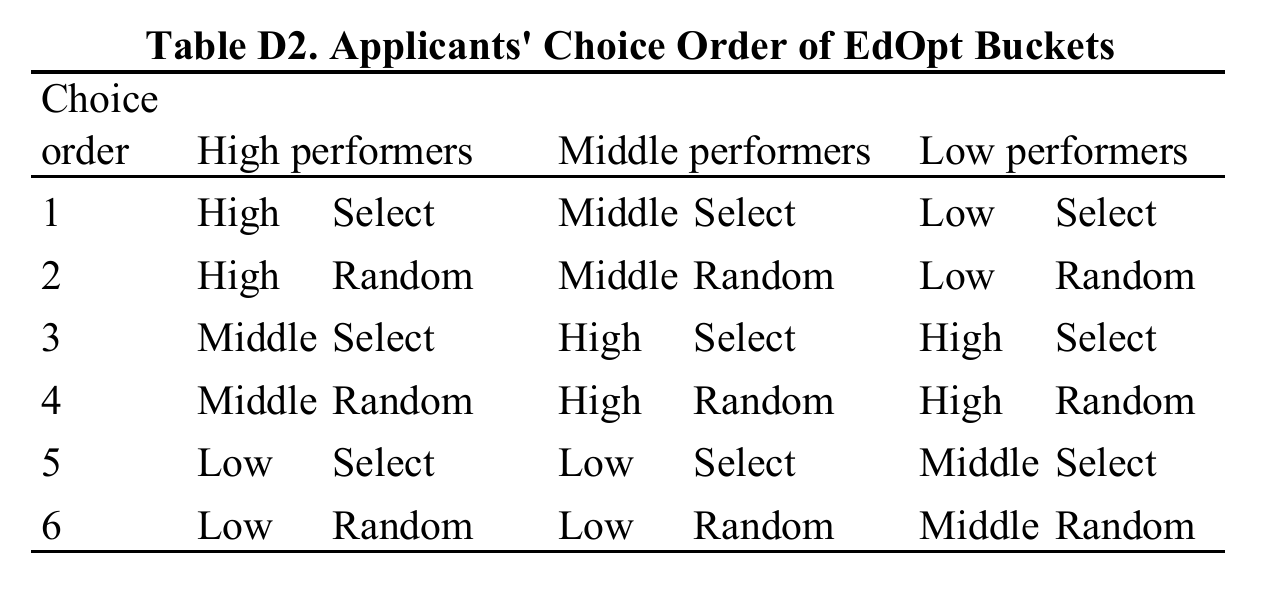} \refstepcounter{tablenums} \label{tabappdx:edoptschoice}
\end{figure}
\noindent 
 High performers have highest priority (priority group 1) at high buckets, while medium and low performers receive highest priority at medium and low buckets, respectively. 

\begin{figure}[H]
 \centering
 \includegraphics[scale=0.9]{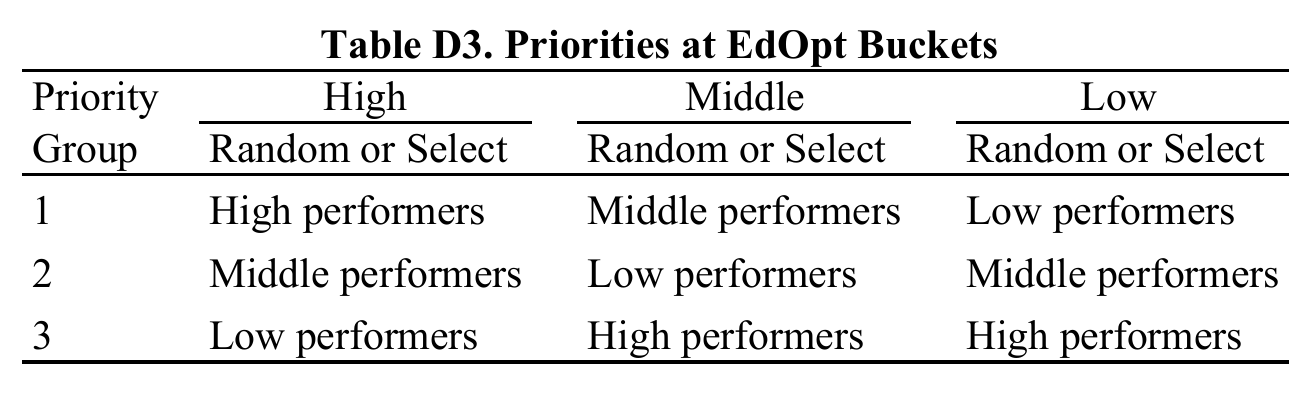} \refstepcounter{tablenums} \label{tabappdx:edoptspriority}
\end{figure}

\noindent \textbf{Miscellaneous Sample Restrictions}
\noindent 
The analysis sample is limited to first-time eighth grade applicants for ninth grade seats. Ineligible applications (as indicated in the main application data set) are dropped. Applicants with special education status compete for a different set of seats and are thus dropped in the analysis.

\noindent Students in the top 2\% of scorers on the seventh grade reading are automatically admitted into any Ed-Opt program they rank first. We gather these assignments in a separate Ed-Opt bucket, thereby leaving the admission process to the other six unaffected.

Table \ref{tabappdx:reprates} records the proportion of applicants for which our match replication was successful.
\linebreak

\begin{center}
\includegraphics[scale=0.9]{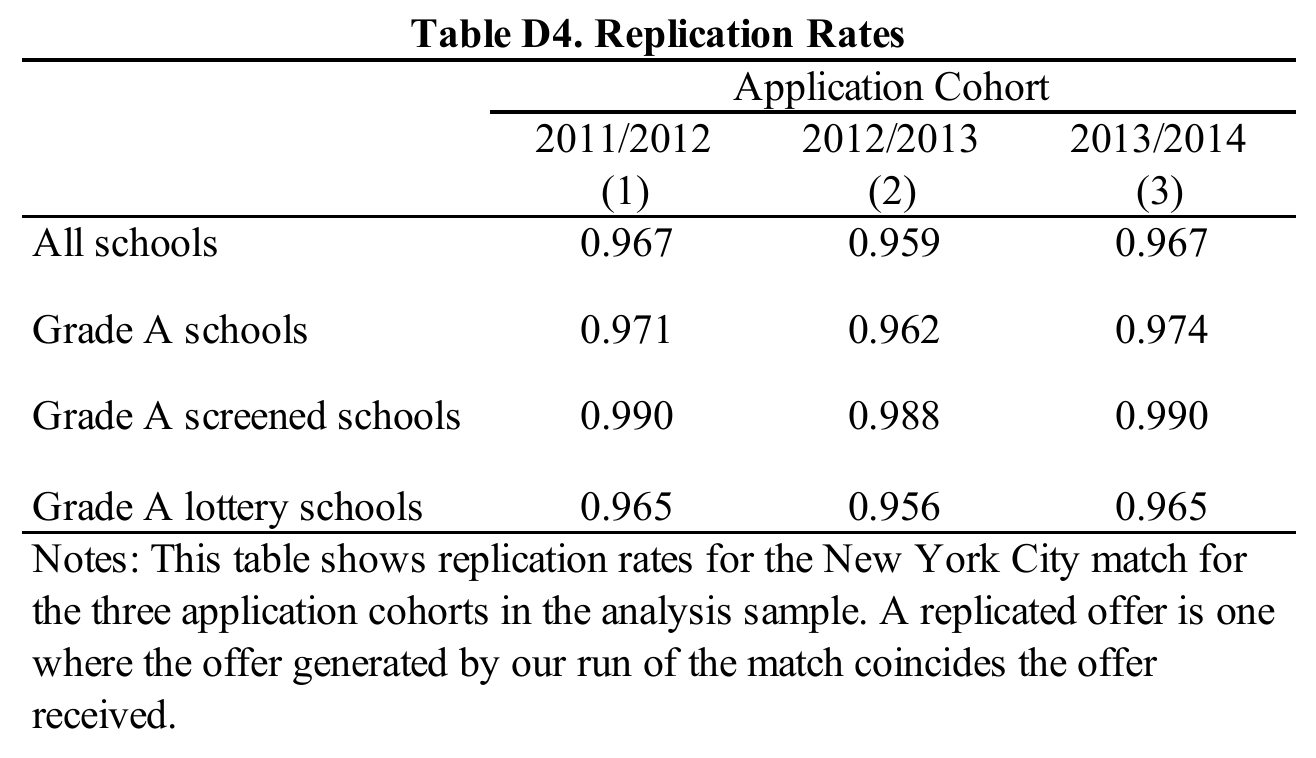} \refstepcounter{tablenums} \label{tabappdx:reprates}
\end{center}

\subsection{Additional Empirical Results}

Grade A risk has a mode at 0.5, but takes on many other values as well. A probability of 0.5 arises when the overall Grade A propensity score is generated by a single Grade A screened school. This can be seen in Figure \ref{figappdx:ptsofsupport}, which tabulates the estimated probability of assignment to a Grade A school for applicants in all cohorts (2012-2014) with a probability strictly between 0 and 1 calculated using the formula in Theorem 1. There are 24,966 students with the estimated assignment probability equal to 1, 86,494 students with the propensity score equal to 0, and 41,647 students with Grade A risk. The propensity score of 0.5 arises when the overall Grade A propensity score is generated by a single Grade A screened school.  

\begin{figure}

	\centering
	\setcounter{figure}{0}
	\caption{Distribution of Grade A Risk}
	 \vspace{0.3cm}
	 \includegraphics[scale=.34]{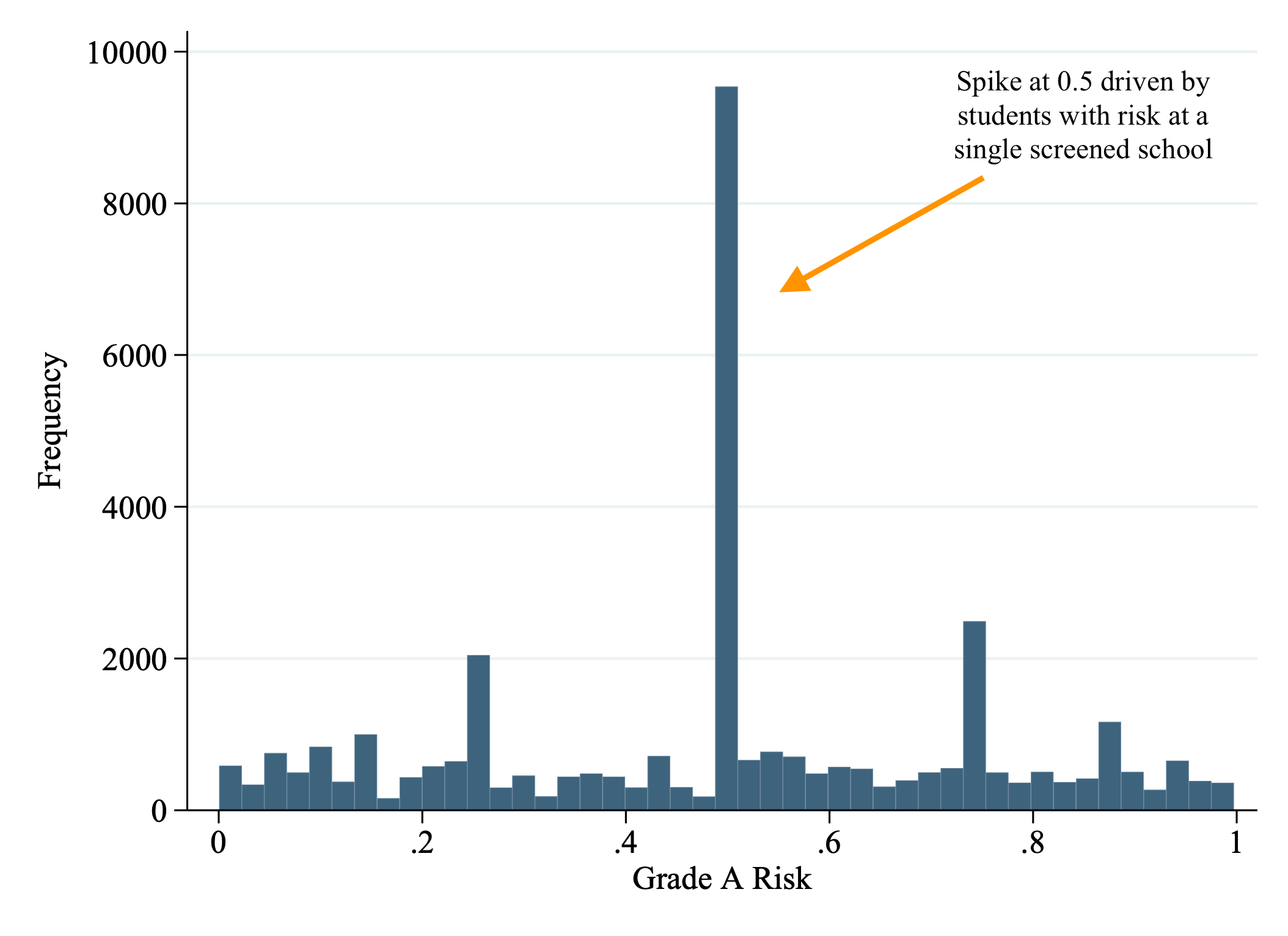}  \label{figappdx:ptsofsupport} 
	\vspace{0.2cm}
    \floatfoot{\small {\textit{Notes:} This figure shows the histogram of the estimated probability of assignment to a Grade A school for at-risk applicants in all sample cohorts (2012-2014), calculated using Theorem 1. The full sample includes 24,966 applicants with a Grade A propensity score equal to 1, 86,494 applicants with propensity score equal to 0, and 35,102 students with Grade A risk. The at-risk sample is used to compute the balance estimates reported in Table \ref{tab:balance}.}}
    
\end{figure}

Table \ref{tabappdx:attrition} reports estimates of the effect of Grade A assignments on attrition, computed by estimating models like those used to gauge balance.  Applicants who receive Grade A school assignments have a slightly higher likelihood of taking the SAT.  Decomposing Grade A schools into screened and lottery schools, applicants who receive lottery Grade A school assignments are 1.6 percent more likely to have SAT scores, while assignments to Grade A screened schools do not correspond to a statistically significant difference in the likelihood of having follow-up SAT scores. This modest difference seems unlikely to bias the 2SLS Grade A estimates reported in Tables \ref{tab:2slsmain} and \ref{tab:2slsmulti}. 

Table \ref{tabappdx:ungraded} reports estimates of the effect of enrollment in an ungraded high school.  These use models like those used to compute the estimates presented in Table \ref{tab:2slsmain}. OLS estimates show a small positive effect of ungraded school attendance on SAT scores and a strong negative effect  on graduation outcomes.  2SLS estimates, by contrast, suggest ungraded school attendance is unrelated to these outcomes. \\

\begin{center} 
	\includegraphics[scale=0.9]{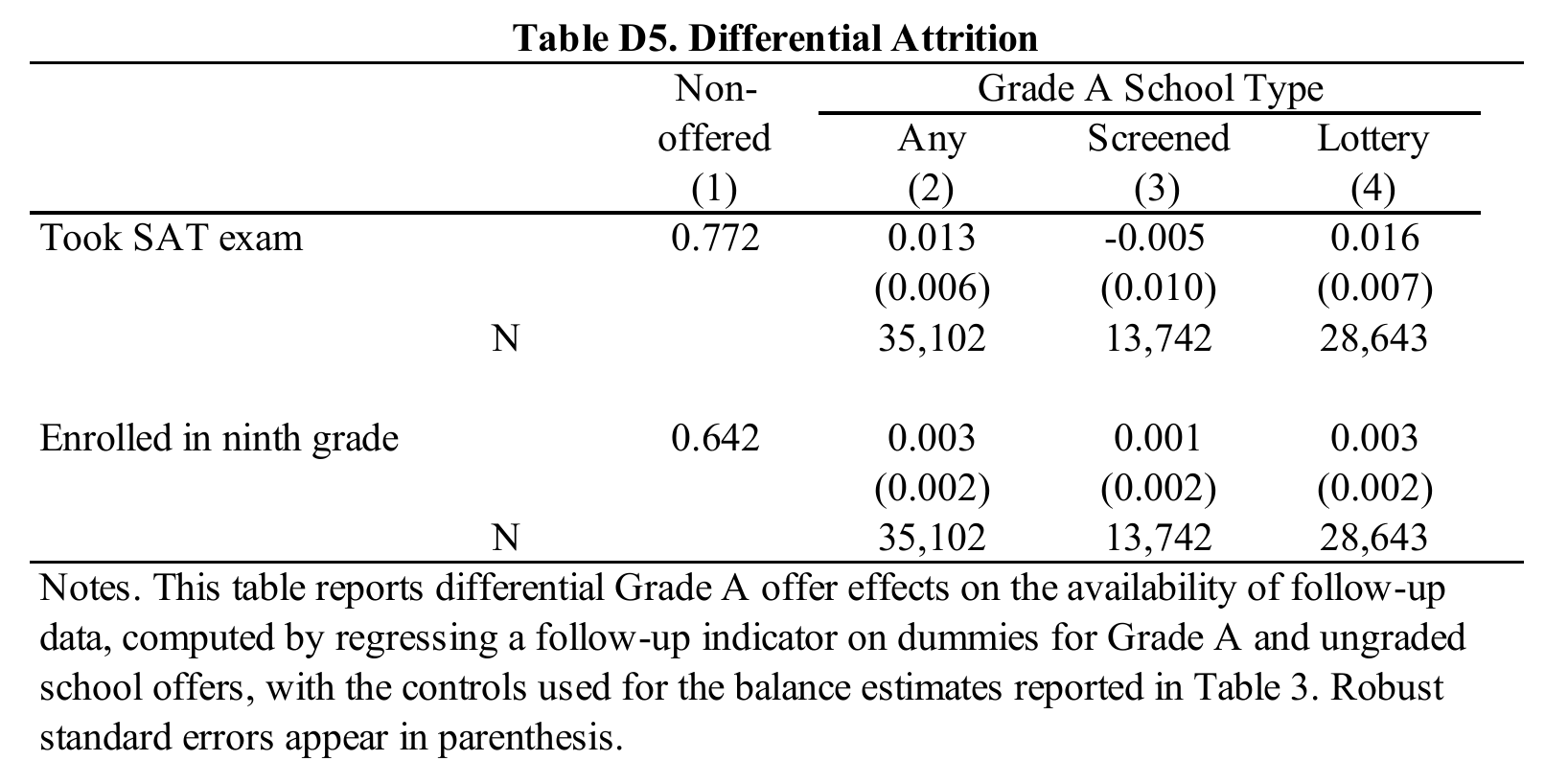} \refstepcounter{tablenums} \label{tabappdx:attrition} \\
\end{center}

\medskip

\begin{center}
	\includegraphics[scale=0.9]{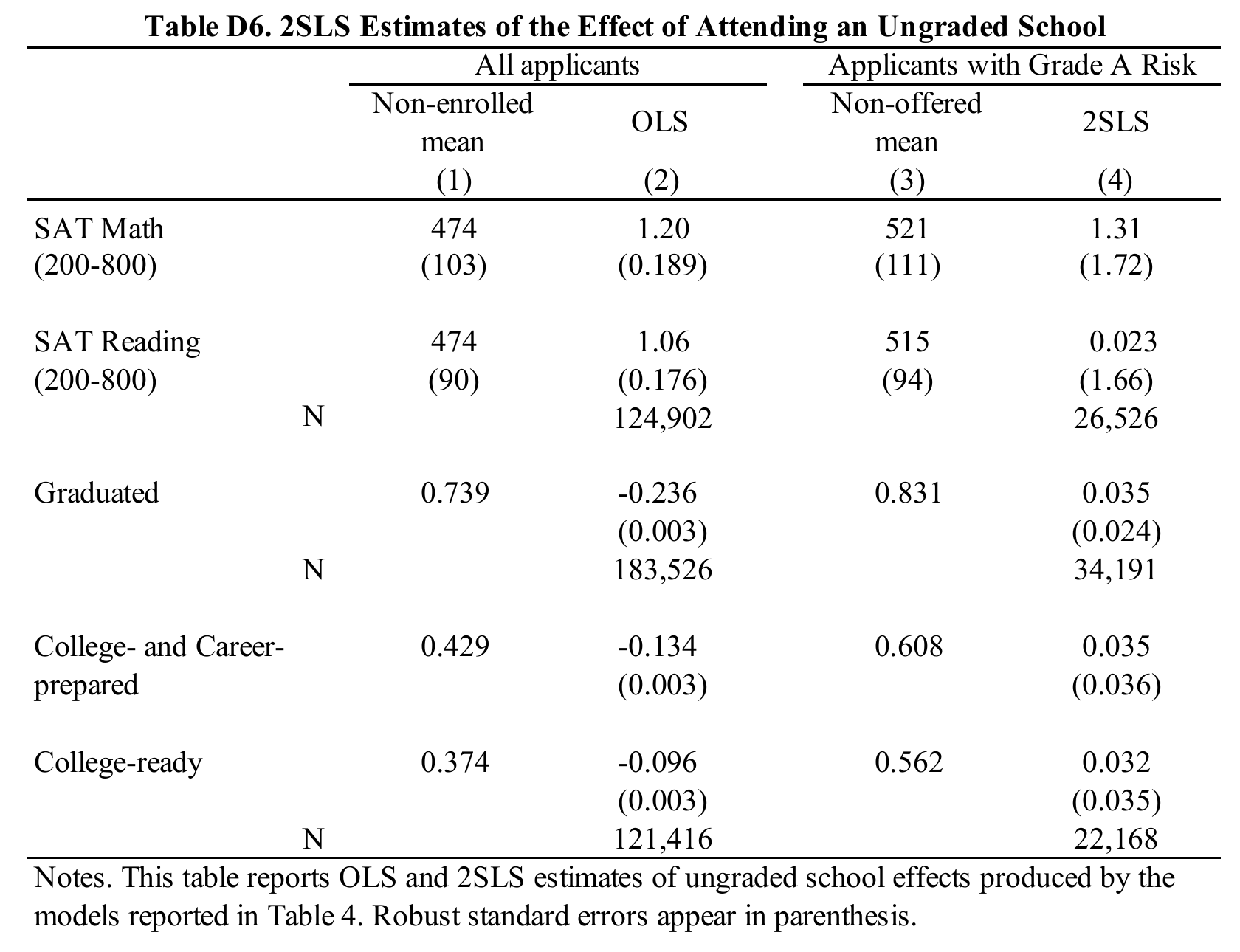} \refstepcounter{tablenums} \label{tabappdx:ungraded} \\
\end{center}

\end{document}